\documentclass[journal]{IEEEtran}

\IEEEoverridecommandlockouts
\usepackage{verbatim}
\usepackage{epsfig}
\usepackage{amsmath}
\usepackage{amsthm}
\usepackage{amssymb}
\usepackage{psfrag}
\usepackage[T1]{fontenc}
\usepackage[ruled]{algorithm2e}

\usepackage{url}
\usepackage{dsfont}
\usepackage{mathtools}
\DeclarePairedDelimiter{\ceil}{\lceil}{\rceil}

\usepackage[usenames,dvipsnames]{pstricks}
\usepackage{pst-grad} 
\usepackage{pst-plot} 
\usepackage[ruled]{algorithm2e}
\usepackage{etoolbox}

\usepackage{graphicx}
\ifCLASSOPTIONcompsoc
    \usepackage[caption=false, font=normalsize, labelfont=sf, textfont=sf]{subfig}
\else
\usepackage[caption=false, font=normalsize]{subfig}
\fi

\usepackage[font=footnotesize,skip=0pt]{caption}

\makeatletter
\patchcmd{\@makecaption}
  {\scshape}
  {}
  {}
  {}
\makeatother

\usepackage[noadjust]{cite}

\newtheorem{theorem}{Theorem}[]
\newtheorem{lemma}{Lemma}[]

\newtheorem{definition}{Definition}
\newtheorem{assump}{Assumption}

\newtheorem{proposition}{Proposition}

\newtheorem{remark}{Remark}

\newcommand{\D}{\textbf{\texttt{DC-DistADMM}}}

\newcommand{\x}{\mathbf{x}}
\newcommand{\y}{\mathbf{y}}

\newcommand{\A}{\mathbf{A}}
\newcommand{\bb}{\mathbf{b}}

\newcommand{\M}{\mathbf{M}}
\newcommand{\F}{\mathbf{F}}
\newcommand{\La}{\mathcal{L}}
\newcommand{\V}{\mathcal{V}}

\DeclareMathOperator*{\argmin}{argmin}
\DeclareMathOperator*{\minimize}{minimize}

\DeclareMathOperator*{\diam}{diam}
\DeclareMathOperator*{\relint}{relint}
\DeclareMathOperator*{\aff}{Aff}

\definecolor{color1}{rgb}{0,0,0}

\IEEEpeerreviewmaketitle

\title{\textbf{\texttt{DC-DistADMM}}:ADMM Algorithm for Constrained Optimization over Directed Graphs}
\author{Vivek Khatana$^1$ and Murti V. Salapaka$^{1}$ \\
\thanks{This work is supported by the Advanced Research Projects Agency-Energy OPEN through the project titled "Rapidly Viable Sustained Grid" via grant no. DE-AR0001016.}
\thanks{$^{1}$ Vivek Khatana \{{\tt\small khata010@umn.edu}\} and Murti V. Salapaka \{{\tt\small murtis@umn.edu\}} are with Department of Electrical and Computer Engineering, University of Minnesota, Minneapolis, USA,}
}
\begin{document}

\maketitle

\begin{abstract}
\textcolor{black}{
This article reports an algorithm for multi-agent distributed optimization problems with a common decision variable, local linear equality and inequality constraints and set constraints with convergence rate guarantees. \textcolor{black}{The algorithm accrues all the benefits of the Alternating Direction Method of Multipliers (ADMM) approach}. It also overcomes the limitations of existing methods on convex optimization problems with linear inequality, equality and set constraints by allowing directed communication topologies. Moreover, the algorithm can be synthesized distributively. The developed algorithm has: (i) a $O(1/k)$ rate of convergence, where $k$ is the iteration counter, when individual functions are convex but not-necessarily differentiable, and (ii) a geometric rate of convergence to any arbitrary small neighborhood of the optimal solution, when the objective functions are smooth and restricted strongly convex at the optimal solution. The  efficacy of the algorithm is evaluated by a comparison with state-of-the-art constrained optimization algorithms in solving a constrained distributed $\ell_1$-regularized logistic regression problem, and unconstrained optimization algorithms in solving a $\ell_1$-regularized Huber loss minimization problem. Additionally, a comparison of the algorithm's performance with other algorithms in the literature that utilize multiple communication steps is provided.}
\\\\
\textit{keywords}: Distributed optimization, constrained optimization, alternating direction method of multipliers
(ADMM), directed graphs, multi-agent networks, finite-time consensus.

\end{abstract}

\begin{section}{Introduction}\label{sec:introduction}
\textcolor{black}{Consider a group of $n$ agents connected through a directed graph, $\mathcal{G}(\mathcal{V},\mathcal{E})$,  where $\mathcal{V}$ and $\mathcal{E}$ are the set of vertices and directed edges respectively. Each agent can transmit information to other agents restricted by the directed graph $\mathcal{G}(\mathcal{V},\mathcal{E})$; an agent $i$ can transmit to agent $j$ if a directed link $i\rightarrow j$ exists in $\mathcal{E}.$ The agents focus on solving the following distributed optimization problem:}
\begin{align}\label{eq:introprob}
    & \hspace{-1.75in} \textcolor{black}{\minimize \limits_{\tilde{x} \in \mathbb{R}^p} \ \ \ \tilde{f}(\tilde{x}) = \textstyle  \sum_{i=1}^{n} \tilde{f}_i(\tilde{x}) } \\
    \textcolor{black}{\mbox{subject to}} \ \ \textstyle \textcolor{black}{ C_i \tilde{x} = c_i, \ D_i \tilde{x} \leq d_i, \ \forall i \in \mathcal{V}}, & \ \  \textcolor{black}{\tilde{x} \in \textstyle \bigcap\limits_{i=1}^n\mathcal{X}_i}, \nonumber
\end{align}
\textcolor{black}{where, $x \in \mathbb{R}^p$ is a global optimization variable. $\tilde{f}_i:\mathbb{R}^p \rightarrow \mathbb{R}$ is the local objective function of agent $i$. $\mathcal{X}_i$ is a convex constraint set associated with the variables of agent $i$. $C_i \tilde{x} = c_i$ with $C_i \in \mathbb{R}^{m_1 \times p}, c_i \in \mathbb{R}^{m_1}$ and $D_i \tilde{x} \leq d_i$ with $D_i \in \mathbb{R}^{m_2 \times p}, d_i \in \mathbb{R}^{m_2}$ are the local equality and inequality constraints of agent $i$.} Many problems in various engineering fields such as wireless systems, multi-agent coordination and control \cite{nedic2018distributed}, and machine learning \cite{nedic2020distributed} can be posed in the form of problem~(\ref{eq:introprob}). \textcolor{black}{Unlike a class of distributed optimization problems that involve global coupling constraints (see for example \cite{wu2021distributed, falsone2017dual, notarnicola2019constraint}), in problem~(\ref{eq:introprob}) agents seeks to determine a common solution $\tilde{x}$ that is required to be known by each agent.\\
\hspace*{0.01in} Early works on the distributed optimization problem can be found in seminal papers \cite{tsitsiklis1984problems, bertsekas1989parallel}. Current approaches for solving optimization problems can be broadly classified as: (i) primal methods, that update the agent estimates of the solution by utilizing a step-size rule based on the gradient information (at the current estimate) of the objective function to steer towards the optimal solution, and (ii) dual based optimization methods that employ Lagrange multipliers. For optimization problems that are unconstrained (with no equality, inequality and set constraints in problem~(\ref{eq:introprob})), examples \cite{nedic2014distributed, xin2020gradient,zeng2015extrapush, pu2020push, khatana2020gradient} (and references therein) take a primal based approach, whereas \cite{MAL-016, wei2012distributed, makhdoumi2017convergence, shi2014linear, wei20131, chang2014multi, iutzeler2015explicit, mansoori2019flexible} take a dual based approach. Articles \cite{wei2012distributed, makhdoumi2017convergence, shi2014linear, wei20131, chang2014multi, iutzeler2015explicit, mansoori2019flexible} motivated by the advantages of parallelizability and good convergence results, (see \cite{gabay1983chapter}) adopt the Alternating Direction Method of Multipliers (ADMM) dual based approach; reference \cite{wei2012distributed} proposed ADMM based algorithm for convex objective functions with an $O(1/k)$ rate of convergence, \cite{shi2014linear, makhdoumi2017convergence} established linear rate of convergence with globally strongly convex objective functions, \cite{chang2014multi} is based on minimizing a proximal first-order approximation of smooth and globally strongly convex functions and \cite{iutzeler2015explicit} provides a linear rate of convergence for a twice continuously differentiable and locally strongly convex functions. In contrast to the dual based approaches \cite{MAL-016, wei2012distributed, makhdoumi2017convergence, shi2014linear, wei20131, chang2014multi, iutzeler2015explicit, mansoori2019flexible} for unconstrained problems, this article provides an algorithm with provable convergence guarantees which is not restricted by assumption of bi-directional (or undirected) communication topologies and does not need centralized information for designing the algorithm. \\
\hspace*{0.01in} The works in \cite{chen2020distributed, chen2021fixed, zhu2011distributed, tian2019distributed, yang2016multi,liu2015second, zhou2019adaptive,lin2021angle,liu2017constrained, nedic2010constrained,li2020distributed,xie2018distributed,mota2013d,yuan2015regularized,lei2016primal} consider constrained distributed optimization problems similar to~(\ref{eq:introprob}). For solving constrained convex optimization problems of the form~(\ref{eq:introprob}) distributively, many of the existing state-of-the-art algorithms \cite{chen2020distributed, chen2021fixed, zhu2011distributed, tian2019distributed, yang2016multi,liu2015second, zhou2019adaptive,lin2021angle,liu2017constrained, nedic2010constrained,li2020distributed,xie2018distributed} utilize primal algorithms based on projections onto the constraint set and the (sub)-gradient information of the individual objective functions. Here, \cite{mota2013d, liu2015second, zhou2019adaptive, lin2021angle, nedic2010constrained, yuan2015regularized, lei2016primal} focus only on set constraints; equality constraints are considered in \cite{yang2016multi, liu2017constrained} while inequality constraints are the focus in \cite{chen2020distributed, chen2021fixed, li2020distributed, xie2018distributed, zhu2011distributed, tian2019distributed, yang2016multi}. To establish convergence, the article \cite{tian2019distributed} assumes the differentiability of the individual objective functions. The assumption of individual functions being Lipschitz continuous and convex quadratic over the constraint set is required in \cite{zhou2019adaptive}. Articles \cite{lin2021angle} and \cite{li2020distributed} require the assumption of the set on which gradient of the individual functions is zero to be bounded and a compact optimal solution set respectively. Twice continuous differentiability and bounded Hessian matrix for the individual functions is assumed in \cite{liu2017constrained}. The articles \cite{nedic2010constrained, xie2018distributed}, require bounded (sub)-gradients. In contrast to the assumptions above in the primal based approaches in \cite{chen2020distributed, chen2021fixed, zhu2011distributed, tian2019distributed, yang2016multi,liu2015second, zhou2019adaptive,lin2021angle,liu2017constrained, nedic2010constrained,li2020distributed,xie2018distributed}, the method reported in this article does not require differentiability nor does it require bounded (sub)gradients. Articles \cite{chen2020distributed, chen2021fixed, zhu2011distributed, tian2019distributed, yang2016multi,liu2015second, zhou2019adaptive,lin2021angle,liu2017constrained, nedic2010constrained,li2020distributed,xie2018distributed} establish convergence to an optimal solution; however, unlike the focus of this article, no convergence rate/iteration complexity estimates are determined. Moreover, this article is based on ADMM and thus accrues its advantages; ADMM  is shown to have better
empirical performance than the primal (projected) sub-gradient
methods \cite{MAL-016}. \\
\hspace*{0.01in} References \cite{mota2013d,yuan2015regularized,lei2016primal} utilize a primal-dual approach for constrained problems. Here, \cite{yuan2015regularized} requires bounded (sub)-gradients to establish an $O(k^{-1/4})$ rate of convergence for the objective function residual and \cite{lei2016primal} does not establish rate of convergence in the presence of constraints.  As alluded to earlier, ADMM based approach (which is a Lagrangian dual based approach) has several advantages (\cite{MAL-016}); however, results on distributed constrained optimization problems using ADMM are sparse. Here \cite{mota2013d} which considers dual ADMM-like method does not provide a convergence rate/iteration complexity analysis.\\
\hspace*{0.01in} Note that most existing constrained optimization schemes \cite{mota2013d, chen2020distributed, chen2021fixed, zhu2011distributed, tian2019distributed, yang2016multi,liu2015second, zhou2019adaptive,lin2021angle,liu2017constrained, nedic2010constrained,li2020distributed,yuan2015regularized,lei2016primal} are designed using centralized information and work under the assumption of bi-directional (or undirected) communication networks. Moreover, the algorithms in \cite{chen2020distributed, chen2021fixed, tian2019distributed, yang2016multi,liu2015second, zhou2019adaptive, xie2018distributed,yuan2015regularized,lei2016primal} are continuous-time algorithms that pose the need for further analysis under discretization schemes needed for implementation. Compared to these our proposed algorithm is based on discrete-time iterations which can be easily implemented in a practical setting. \textcolor{black}{However, unlike some existing work in the literature \cite{chen2021fixed, zhu2011distributed}, in the current work we do not consider nonlinear inequality constraints and time-varying graphs.} We now summarize the discussion above to delineate the main contribution of the article as follows. \\
\hspace*{0.01in} The main contribution of the article is an ADMM based algorithm called \textbf{\underline{D}irected \underline{C}onstrained-\underline{Dist}ributed \underline{A}lternating \underline{D}irection \underline{M}ethod of \underline{M}ultipliers ($\D$)} that solves problem (\ref{eq:introprob}). \textcolor{black}{The novel features of the algorithm  for unconstrained and constrained optimization problems are summarized below}:
\begin{enumerate}
    \item [(i)] With respect to unconstrained optimization problems, $\D$ is \textcolor{black}{a Lagrangian} dual based algorithm:\\
    \textbf{1.} \textcolor{black}{that} accommodate\textcolor{black}{s} directed graphs allowing for non symmetric communication topologies. This feature considerably widens the applicability as in most application scenarios, agents do not have the same range of communication. \\
    \textbf{2.} that is amenable to distributed synthesis scenarios and extends the applicability of ADMM method to applications where a plug and play operation is required \cite{patel2017distributed, patel2020distributed} (distributed synthesis  is detailed in Appendix~\ref{sec:distri_synth}). The $\D$ algorithm has column stochastic updates and does not require the agents to know global network interconnection information.
    \item [(ii)] With respect to the constrained optimization problems, $\D$ algorithm:\\
    \textbf{1.} is \textcolor{black}{a Lagrangian} dual based algorithm for constrained optimization problems that accommodates directed graphs and allows for distributed synthesis.\\
    \textbf{2.} to the best of the authors' knowledge $\D$ has the best convergence rates in the constrained optimization literature. Most works \cite{chen2020distributed, chen2021fixed, zhu2011distributed, tian2019distributed, yang2016multi,liu2015second, zhou2019adaptive,lin2021angle,liu2017constrained, nedic2010constrained,li2020distributed,xie2018distributed,mota2013d,lei2016primal} do not provide convergence rate guarantees while the article \cite{yuan2015regularized} provides a worse one. The $\D$ algorithm converges to an optimal solution of problem~(\ref{eq:introprob}) under mild assumptions (Theorem~\ref{thm:convergence}). The algorithm has a $O(1/k)$ rate of convergence, where $k$ is the iteration counter, when the individual functions $\tilde{f}_i$ are convex but not-necessarily differentiable. It has a geometric rate of convergence to any arbitrary neighborhood of the optimal solution when the objective functions are smooth and restricted strongly convex at the optimal solution of problem~(\ref{eq:introprob}). $\D$ does not require the individual functions $\tilde{f}_i$ to be differentiable for convergence to an optimal solution neither bounded sub-gradients and globally strongly convex or twice continuously differentiable is assumed for the derivation of convergence rate estimates unlike existing analysis in the literature (see \cite{pu2020push, shi2015extra, shi2014linear, makhdoumi2017convergence, iutzeler2015explicit,tian2019distributed, zhou2019adaptive,lin2021angle, li2020distributed, liu2017constrained, nedic2010constrained, xie2018distributed, yuan2015regularized }). 
    \item[(iii)] ADMM based approaches and more generally \textcolor{black}{Lagrangian} dual based approaches for distributed constrained problems involve in each iteration an information mixing step among agents and a step that improves the objective. An important technical contribution of the article is that it transforms the original problem~(\ref{eq:introprob}) to an equivalent problem (see~(\ref{eq:distOpt_indifunc1})) that allows for carrying out almost \textit{ complete} mixing via \textit{any} consensus strategy for the information mixing step to be followed by the objective improvement step. This reformulation with its advantages is being employed by other groups including works in articles \cite{jiang2021fully, jiang2021distributed, jiang2021asynchronous} and \cite{rokade2020distributed} with consensus strategies other than the one employed in this article.  
    \item[(iv)] The total number of communication steps for $\D$ by the $k^{th}$ iterate is within a factor of $\log k$ of the optimal lower bound in obtaining a $O(1/k)$ rate of convergence. Empirical data corroborates that $\D$, with respect to the computational performance, outperforms other distributed constrained and unconstrained optimization approaches. 
\end{enumerate}
}
\noindent The authors have introduced a preliminary version \cite{khatana2020cdc} of the method developed here termed as D-DistADMM. The current work presents a significant generalization, $\D$, where a \textit{constrained} distributed optimization problem is considered. The constrained optimization problem poses significant technical challenges over the unconstrained case; here, stronger and more comprehensive results under less restrictive assumptions are established. In contrast to \cite{khatana2020cdc} this article provides the convergence rate guarantees, establishes upper bounds on the total communication steps required for achieving the convergence rate estimates (remarks~\ref{rem:1/k_comm} and~\ref{rem:linear_comm}), and validation of the $\D$ algorithm's applicability using detailed numerical tests and comparison with other state-of-the-art algorithms is provided.

The rest of the article is organized as follows. In  Section~\ref{sec:probform}, the problem under consideration is discussed in detail and some basic definitions and notations used in the article  are presented. Section~\ref{sec:ddistadmm} presents the $\D$ algorithm along with the finite-time $\varepsilon$-consensus protocol in detail; and provide supporting analysis for the $\varepsilon$-consensus protocol. The convergence analysis of the proposed algorithm is provided in Section~\ref{sec:convgAnalysis}. \textcolor{black}{In Section~\ref{sec:results}, numerical simulations comparing the existing state-of-the-art methods and $\D$ algorithm in solving: (i) a $\ell_1$ regularized distributed logistic regression problem and (ii) a \textcolor{black}{$\ell_1$ regularized distributed Huber loss minimization problem, are provided.} The comparison results verify theoretical claims and provide a discussion on the effectiveness and suitability of the proposed algorithm. Section~\ref{sec:conclusion} provides the concluding remarks and discusses some future work. }
\end{section}

\begin{section}{Definitions, Problem Formulation and Assumptions}\label{sec:probform}
\subsection{Definition, Notations and Assumptions}
In this section, definitions and notations that are used later in the analysis are presented. 
Detailed description of most of these notions can be found in \cite{Die06, horn2012matrix, rockafellar2015convex}.

\begin{definition}(Directed Graph)
A directed graph $\mathcal{G}$ is a pair $(\mathcal{V},\mathcal{E})$ where $\mathcal{V}$ is a set of vertices (or nodes) and $\mathcal{E}$ is a set of edges, which are ordered subsets of two distinct elements of $\mathcal{V}$. If an edge from $j \in \mathcal{V}$ to $i \in \mathcal{V}$ exists then it is denoted as $(i,j)\in \mathcal{E}$. 
\end{definition}


\begin{definition}(Strongly Connected Graph) A directed graph is strongly connected if for any pair $(i,j),\ i\not =j$, there is a directed path from node $i$ to node $j$.
\end{definition}

\begin{definition}(Column-Stochastic Matrix) A matrix $M=[m_{ij}] \in \mathbb{R}^{n\times n}$ 
is called a column-stochastic matrix if $0 \leq m_{ij}\leq 1$ and $\sum_{i=1}^{n}m_{ij}=1$ for all $ 1 \leq i,j \leq n$. 
\end{definition}




\begin{definition}(Diameter of a Graph)
The diameter of a directed graph $\mathcal{G}(\mathcal{V},\mathcal{E})$ is the longest shortest directed path between any two nodes in the the graph. 
\end{definition}

\begin{definition}(In-Neighborhood) The set of in-neighbors of node $i \in \mathcal{V}$ is called the in-neighborhood of node $i$ and is denoted by $\mathcal{N}^{-}_i = \{j \ | \ (i,j)\in \mathcal{E}\}$ not including the node $i$.
\end{definition}

\begin{definition}(Out-Neighborhood) The set of out-neighbors of node $i \in \mathcal{V}$ is called the out-neighborhood of node $i$ and is denoted by $\mathcal{N}^{+}_i = \{j \ | \ (j,i)\in \mathcal{E}\}$ not including the node $i$.
\end{definition}



\begin{definition}(Lipschitz Differentiability)
A differentiable function $f: \mathbb{R}^p \to \mathbb{R}$ is called Lipschitz differentiable with constant $L > 0$, if the following inequality holds:  
\begin{align*}
    \| \nabla f(x) - \nabla f(y)\| \leq L \|x-y\|, \ \forall \ x,y \in \mathbb{R}^p.
\end{align*}
\end{definition}
\vspace{-0.25in}
\textcolor{black}{ \begin{definition}(Restricted Strongly Convex Function \cite{wu2020second}) Given, $\tilde{x} \in \mathbb{R}^p, C \subseteq \mathbb{R}^p$, a differentiable convex function $f: \mathbb{R}^p \to \mathbb{R}$ is called restricted strongly convex with respect to $\tilde{x}$ on $C$ with parameter $\sigma > 0$, if the following inequality holds: 
\begin{align*}
    \langle \nabla f(x)-\nabla f(\tilde{x}), x-\tilde{x} \rangle \geq \sigma \|x-\tilde{x}\|^2, \ \forall \ x \in C.
\end{align*}
\end{definition}
}

\begin{definition}(Diameter of a set) For a norm $\|\cdot \|$ and a set $K \subset \mathbb{R}^p$ define the diameter of $K$ with respect to the norm $\|\cdot \|$,
$\diam_{\|\cdot \|}(K) := \sup\limits_{x, y \in K} \|x-y\|$.
\end{definition}

\begin{definition}(Affine hull of a set)
The affine hull $\aff(X)$ of a set $X$ is the set of all affine combinations of elements of $X$,
\begin{align*}
  \aff(X) = \left \{ \textstyle \sum_{i=1}^k \theta_i x_i \ | \ k >0, x_i \in X, \theta_i \in \mathbb{R}, \textstyle \sum_{i=1}^k \theta_i =1 \right\}.
\end{align*}
\end{definition}

\begin{definition}(Relative interior of a set)
Let $X \subset \mathbb{R}^n$. A point $x \in X$ is in the relative interior of $X$, if $X$ contains the intersection of a small enough ball centered at $x$ with the $\aff(X)$, that is, there exists $r>0$ such that 
\begin{align*}
    B_r(x) \cap \aff(X) := \{y \ | \ y \in \aff(X), \|y-x\| \leq r \} \subset X.  
\end{align*}
The set of all relative interior points of $X$ is called its relative interior denoted by $\relint(X)$. 
\end{definition}

\noindent Throughout the article, vectors are assumed to be column vectors unless stated otherwise. Each agent $i \in \mathcal{V}$ maintains two local primal variables \textcolor{black}{$x_i^k \in \mathbb{R}^{(m+p)}$, $y_i^k \in \mathbb{R}^{(m+p)}$ and two dual variables $\lambda_i^k \in \mathbb{R}^{(m+p)}$}, $\mu_i^k \in \mathbb{R}^m$ at all iterations $k \geq 0$ of the $\D$ algorithm. \textcolor{black}{$\mathbf{x}^k = [x_1^{k^\top} \ x_2^{k^\top} \dots\  x_n^{k^\top}]^\top \in \mathbb{R}^{n(m+p)}, \mathbf{y}^k = [y_1^{k^\top} \ y_2^{k^\top} \dots\  y_n^{k^\top}]^\top \in \mathbb{R}^{n(m+p)}$, $\mathbf{\lambda}^k = [\lambda_1^{k^\top} \ \lambda_2^{k^\top} \dots \lambda_n^{k^\top}]^\top \in \mathbb{R}^{n(m+p)}$}, and $\mathbf{\mu}^k = [\mu_1^{k^\top} \ \mu_2^{k^\top} \dots \mu_n^{k^\top}]^\top$ $ \in \mathbb{R}^{nm}$ respectively. $\mathcal{D}$ denotes an upper bound on the diameter of the graph $\mathcal{G}(\mathcal{V},\mathcal{E})$. \textcolor{black}{$\mathbb{R}^m_{\geq 0}$ denotes the non-negative sub-space of $\mathbb{R}^m$}, $\|.\|$ denotes the 2-norm of the vector input unless stated otherwise and $\ceil*.$ denotes the least integer function or the ceiling function, defined as: 
given $x \in \mathbb{R}, \ceil*x = \min \{ m \in \mathbb{Z} | m \geq x\},$ where $\mathbb{Z}$ is the set of integers. \textcolor{black}{Given a set $S$, the indicator function of set $S$ is defined as:}
\begin{align*}
    \textcolor{black}{\mathcal{I}_{S}(x) = \begin{cases}
           0 & \text{if} \ x \in S\\
           +\infty & \text{otherwise}.
           \end{cases} }
\end{align*}

\begin{subsection}{Problem Formulation}
\textcolor{black}{Consider the slack variable $\tilde{x}_s$ and re-write the inequality constraint $D_i \tilde{x} \leq d_i$ as,
    \begin{align*}
        \begin{bmatrix}
            D_i & \mathbf{I}
        \end{bmatrix} \begin{bmatrix}
            \tilde{x} \\ \tilde{x}_s
        \end{bmatrix} = d_i, \ \tilde{x}_s \succeq 0, \ \forall i \in \mathcal{V},
    \end{align*}
where, $\mathbf{I}$ is the identity matrix of appropriate dimension. Let $x = [\tilde{x}^\top \ \tilde{x}_s^\top]^\top \in \mathbb{R}^{(m+p)}$. Using the notation and the definition of the indicator function of $\mathbb{R}^m_{\geq 0}$ problem~(\ref{eq:introprob}) can be rewritten as:}
\begin{align}\label{eq:introprob_v1}
     & \hspace{-1.95in} \textcolor{black}{\minimize \limits_{x \in \mathbb{R}^p} \ \ \ \textstyle  \sum_{i=1}^{n} [\tilde{f}_i(\tilde{x}) + \mathcal{I}_{\mathbb{R}^m_{\geq 0}}(\tilde{x}_s) ] } \\
    \textcolor{black}{\mbox{subject to}} \ \ \textstyle \textcolor{black}{ A_i x = b_i, \ \forall i \in \mathcal{V}}, & \ \  \textcolor{black}{\tilde{x} \in \textstyle \bigcap\limits_{i=1}^n\mathcal{X}_i}, \nonumber
\end{align}
\textcolor{black}{where, $A_i :=[C_i + D_i \ \ \mathbf{I}] \footnote{without loss of generality here we assume $m_1 = m_2 = m$. }, b_i := c_i + d_i$}. Problem~(\ref{eq:introprob_v1}) is recast by creating local copies $x_i$ for all $i \in \mathcal{V}$, of the global variable $x$ and imposing the agreement of the solutions of all the agents via consensus constraint. This leads to an equivalent formulation of~(\ref{eq:introprob_v1}) as described below:
\textcolor{black}{
\begin{align} \label{eq:first_ref}
    & \textcolor{black}{\minimize \limits_{x \in \mathbb{R}^p} \ \ \ \textstyle  \sum_{i=1}^{n} [\tilde{f}_i(\tilde{x}_i) + \mathcal{I}_{\mathbb{R}^m_{\geq 0}}(\tilde{x}_{s_i}) ] } \\
    & \ \text{subject to} \ \ \textstyle \textcolor{black}{A_i x_i = b_i}, \ \tilde{x}_i \in \mathcal{X}_i, \  \forall i \in \mathcal{V}, \nonumber \\
    & \hspace{0.7in} x_i = x_j, \ \forall i,j \in \mathcal{V}, \nonumber
\end{align}}The constraints \textcolor{black}{$\tilde{x}_i \in \mathcal{X}_i$} can be integrated into the objective function using the indicator functions of the sets $\mathcal{X}_i$, as:
\textcolor{black}{
\begin{align}\label{eq:distOpt_const}
    & \minimize \ \ \textstyle \sum_{i=1}^n f_i(x_i) \\
    & \ \text{subject to} \ \ \textstyle \textcolor{black}{A_i x_i = b_i}, \ \forall i \in \mathcal{V}, \nonumber\\ 
    & \hspace{0.7in} x_i = x_j, \ \forall i,j \in \mathcal{V}, \nonumber
\end{align}
}where, \textcolor{black}{$f_i(x_i) := \tilde{f}_i(\tilde{x}_i) + \mathcal{I}_{\mathbb{R}^m_{\geq 0}}(\tilde{x}_{s_i}) + \mathcal{I}_{\mathcal{X}_i}(\tilde{x}_i)$ and $\mathcal{I}_{\mathcal{X}_i}$ is the indicator function of set $\mathcal{X}_i$.}
\begin{align}\label{eq:set_C}
 \hspace{-0.08in} \mbox{Let}, \ \mathcal{C}_\eta :=     \big \{\y = [&y_1^\top \dots  y_n^\top]^\top \in \textcolor{black}{\mathbb{R}^{n(m+p)}} \ \mbox{such that} \ \nonumber \\ &  \|y_i - y_j\| \leq 2\eta,  1 \leq i,j \leq n \big \},
\end{align}
to be a set of vectors $\y \in \textcolor{black}{\mathbb{R}^{n(m+p)}}$ such that the norm of the difference between any two $(m+p)$ dimensional sub-vectors of $\y$ is less than $\eta$. 
\noindent Using definition of set $\mathcal{C}_\eta$,~(\ref{eq:distOpt_const}) is equivalent to:
\begin{align}\label{eq:distOpt_setCons}
    & \minimize \ \ \textstyle \sum_{i=1}^n f_i(x_i) \\
    & \ \text{subject to} \ \ \textcolor{black}{A_i x_i = b_i}, \ \forall i \in \mathcal{V}, \nonumber\\ 
    & \hspace{0.7in} \x = [x_1^\top x_2^\top \dots x_n^\top]^\top, \ \x = \y, \ \y \in \mathcal{C}_0. \nonumber
\end{align}
Problem~(\ref{eq:distOpt_setCons}) can be reformulated as:
\begin{align}\label{eq:distOpt_indifunc1}
    & \minimize \ \ \textstyle \sum_{i=1}^n f_i(x_i) + \mathcal{I}_{\mathcal{C}_0}(\y) \\
    & \ \text{subject to} \ \ \textstyle \textcolor{black}{A_i x_i = b_i}, \ \forall i \in \mathcal{V}, \nonumber \\
    & \hspace{0.7in} \x = [x_1^\top x_2^\top \dots x_n^\top]^\top, \ \x =\y, \nonumber
\end{align}
where, $\mathcal{I}_{\mathcal{C}_0}$ is the indicator function of the set $\mathcal{C}_0$. 
\end{subsection}
\begin{subsection}{Standard ADMM Method}
Here a brief review of the standard ADMM method \cite{MAL-016} utilized to solve optimization problems of the following form:
\begin{align}\label{eq:stdadmm_opt}
    \minimize_{x \in \mathbb{R}^p,y\in \mathbb{R}^q} &\ \ f(x) + g(y)\\
    \text{subject to}& \ \ Sx + Ty = c,\nonumber
\end{align}
where, $S \in \mathbb{R}^{n \times p}, T \in \mathbb{R}^{n \times q}$, and $c \in \mathbb{R}^{n}$ is provided. Consider, the augmented Lagrangian with the Lagrange multiplier $\lambda$ and positive scalar $\gamma$,
\begin{align}\label{eq:stdadmmauglag}
   & \mathcal{L}_{\gamma}(x,y,\lambda) = f(x) + g(y) + \lambda^\top(Sx + Ty - c) \nonumber\\
    &\hspace{1.8in} \textstyle + \frac{\gamma}{2}\|Sx + Ty - c\|^2.
\end{align}
The primal and dual variable updates in ADMM are given as: starting with the initial guess $(x^0,y^0,\lambda^0)$, at each iteration, 
\begin{align}
   x^{k+1} &= \argmin\limits_{x}  \mathcal{L}_{\gamma}(x,y^k,\lambda^k) \label{eq:stdadmm_x}\\
  y^{k+1} &= \argmin\limits_{y}  \mathcal{L}_{\gamma}(x^{k+1},y,\lambda^k) \label{eq:stdadmm_y} \\
  \lambda^{k+1} &= \lambda^{k} + \gamma(Sx^{k+1} + Ty^{k+1} - c). \label{eq:stdadmm_lam}
\end{align}
\end{subsection}
\end{section}

\begin{section}{The Proposed $\D$ Method}\label{sec:ddistadmm}
\textcolor{black}{Define $\F(\x) := \sum_{i=1}^n f_i (x_i)$, \textcolor{black}{$\A:= \mathbf{I} \otimes A_i, i=1,\dots,n, \in \mathbb{R}^{nm \times n(m+p)}$, where $\mathbf{I}$ is $n \times n$ identity matrix and $\otimes$ denote the matrix Kronecker product, and $\bb := [b_1^\top \dots b_n^\top]^\top \in \mathbb{R}^{nm}$}. The Lagrangian function $\La$ for problem~(\ref{eq:distOpt_indifunc1}) is given by:
\begin{align} \label{eq:lag}
    \La (\x,\y,\lambda,\mu) = \F(\x) + \mathcal{I}_{\mathcal{C}_0}(\y) + \lambda&^\top(\x - \y) \nonumber \\
    & + \mu^\top (\A\x - \bb). 
\end{align}
Note that the standard augmented Lagrangian associated with~(\ref{eq:distOpt_indifunc1}) at any iteration $k$ is:
\begin{align}\label{eq:Aug_Lag_split}
 \La_{\gamma}(\x^k,\y^k,\lambda^k,\mu^k) = \La &(\x^k,\y^k, \lambda^k,\mu^k) + \textstyle \frac{\gamma}{2}\|\x^k - \y^k\|^2 \nonumber \\
 & \textstyle + \frac{\gamma}{2} \|\A\x - \bb\|^2. 
\end{align}}Based on ADMM iterations~(\ref{eq:stdadmm_x})-(\ref{eq:stdadmm_lam}) the primal and dual updates corresponding to augmented Lagrangian~(\ref{eq:Aug_Lag_split}) are: 
\textcolor{black}{\begin{align}
   \x^{k+1} = \textstyle & \argmin\limits_{\x} \Big\{ \F(\x) + \lambda^{k \top} (\x - \y^{k}) \textstyle + \mu^{k \top} (\A \x - \bb) \nonumber \\
   & \textstyle + \frac{\gamma}{2} \| \x - \y^k\|^2 \textstyle + \frac{\gamma}{2} \|\A \x - \bb\|^2 \Big\}.
\end{align}
Note that the above $\x^{k+1}$ update can be obtained in a decentralized manner by computation at each agent as follows:
\begin{align}
   x_i^{k+1} = \textstyle & \argmin\limits_{x_i} \Big\{ f_i(x_i) + \lambda_i^{k^\top} (x_i - y_i^{k}) \textstyle + \mu_{i}^{k^\top} \textcolor{black}{(A_i x_i - b_i)} \nonumber \\
   & \textstyle + \frac{\gamma}{2} \| x_i - y_i^k\|^2 \textstyle + \frac{\gamma}{2} \textcolor{black}{\|A_i x_i - b_i\|^2} \Big\},  \forall i\in \mathcal{V}. \label{eq:dadmm_x}
\end{align}
The update for the primal variable $\y$ is given as:
\begin{align}
   \overline{\y}^{k+1}  = \textstyle & \argmin\limits_{\y} \Big\{\mathcal{I}_{\mathcal{C}_0}(\y) + \lambda^{k^\top}(\x^{k+1} - \y) + \textstyle \frac{\gamma}{2}\|\x^{k+1} - \y \|^2 \Big\} \nonumber\\ 
   = \textstyle & \argmin\limits_{\y} \textstyle \Big\{\mathcal{I}_{\mathcal{C}_0}(\y) + \frac{\gamma}{2}\|\x^{k+1} - \y + \frac{1}{\gamma} \lambda^{k}\|^2\Big\},
  \label{eq:dadmm_y}
\end{align}
The $\overline{\y}^{k+1}$ update in~(\ref{eq:dadmm_y}) is the projection of $\x^{k+1} + \frac{1}{\gamma}\lambda^{k}$ on the set $\mathcal{C}_0$. It can be verified (see Appendix~{B}) that
\begin{align}
   & \overline{\y}^{k+1} = [\overline{y}^{k+1^\top} \dots \overline{y}^{k+1^\top}]^\top \in \textcolor{black}{\mathbb{R}^{n(m+p)}}, \mbox{where,} \label{eq:y_over}\\
    &\overline{y}^{k+1} = \textstyle \frac{1}{n} \sum_{i=1}^n [x_i^{k+1} + \frac{1}{\gamma} \lambda_i^{k}], \nonumber
\end{align}
is the solution of the update~(\ref{eq:dadmm_y}). However, obtaining the optimal solution (the exact average of the local variables) in a distributive manner over a directed network where each agent only has access to its own local information is a challenge. Hence, the exact consensus constraint is relaxed to a requirement of a $\eta_{k+1}$-closeness among the variables of all the agents, i.e., the solution $\y^{k+1}$ is allowed to lie in the set $\mathcal{C}_{\eta_{k+1}}$ and satisfy $\|y^{k+1}_i - y^{k+1}_j\| \leq 2 \eta_{k+1}, \forall i,j \in \V$.
The parameter $\eta_{k+1}$ can be interpreted as a specified tolerance on the quality of consensus among the agent variables and can be chosen appropriately to find a corresponding near optimal solution of~(\ref{eq:dadmm_y}). A distributed finite-time terminated approximate consensus protocol  is employed to find an \textit{inexact} solution of~(\ref{eq:dadmm_y}), such that the obtained solution $\y^{k+1} \in \mathcal{C}_{k+1}$. The "approximate" consensus protocol is discussed in detail in the next section (section \ref{sec:econs}).\\
After the primal variables, each agent updates the dual variables in the following manner:
\begin{align} 
    \lambda_i^{k+1} & = \lambda_i^{k} \textstyle + \gamma(x_i^{k+1} - y_i^{k+1}), \forall i\in \mathcal{V}, \label{eq:dadmm_lam} \\
   \mu_i^{k+1} & = \mu_i^{k} + \textstyle \gamma (\textcolor{black}{A_i} x_i^{k+1} - \textcolor{black}{b_i}), \forall i\in \mathcal{V}. \label{eq:dadmm_mu}
\end{align}
Define $\overline{\lambda}_i^{k+1}$ based on the solution $\overline{\y}^{k+1}$:
\begin{align} 
    \overline{\lambda}_i^{k+1} & = \lambda_i^{k} + \textstyle \gamma(x_i^{k+1} - \overline{y}_i^{k+1}), \forall i\in \mathcal{V}. \label{eq:dadmm_lam_ne}
\end{align}
Next, the $\varepsilon$-consensus protocol is described.
}
\begin{subsection}{Finite-time $\varepsilon$-consensus protocol}\label{sec:econs}
Here, a finite-time "approximate" consensus protocol called the $\varepsilon$-consensus protocol is presented. The protocol was first proposed in an earlier work \cite{melbourne2020geometry} by the authors. Consider, a set of $n$ agents connected via a directed graph $\mathcal{G}(\mathcal{V},\mathcal{E})$. Every agent $i$ has a vector $u_i^0 \in \mathbb{R}^p$. Let  $\widetilde{u}  = \frac{1}{n} \sum_{i=1}^n u_i^0$ denote the average of the vectors. The objective is to design a distributed protocol such that the agents are able to compute an approximate estimate of $\widetilde{u}$ in finite-time. This approximate estimate is parametrized by a tolerance $\varepsilon$ that can be chosen arbitrarily small to make the estimate as precise as needed. To this end, the agents maintain state variables $u_i^k  \in \mathbb{R}^p, v_i^k  \in \mathbb{R}$ that undergo the following update: for $k \geq 0$,
\begin{align}
   u_i^{k+1} &= \textstyle p_{ii}u_i^{k} +  \sum_{j\in\mathit{\mathcal{N}^-_{i}}}p_{ij}u_j^{k} \label{eq:cons_u}\\
   v_i^{k+1} &= \textstyle p_{ii}v_i^{k}+  \sum_{j\in\mathit{\mathcal{N}^-_{i}}}p_{ij}v_j^{k}\label{eq:cons_v} \\
   w_i^{k+1} &= \textstyle \frac{1}{v_i^{k+1}}u_i^{k+1}, \label{eq:cons_w}
\end{align}
where, $w_i^0 = u_i^0$, $v_i^0 = 1$ for all $i \in \mathcal{V}$ and $\mathcal{N}^{-}_i$ denotes the set of in-neighbors of agent $i$. The updates~(\ref{eq:cons_u})-(\ref{eq:cons_w}) are based on the push-sum (or ratio consensus) updates (see \cite{kempe2003gossip} or \cite{dominguez2011distributed}). The following assumption on the graph $\mathcal{G}(\mathcal{V},\mathcal{E})$ and weight matrix $\mathcal{P}$ is made:
\begin{assump}\label{ass:graph_ass}
The directed graph $\mathcal{G}(\mathcal{V},\mathcal{E})$ is strongly-connected. Let the weighted adjacency matrix $\mathcal{P} = [ p_{ij} ]$, associated with the digraph $\mathcal{G}(\mathcal{V},\mathcal{E})$, be column-stochastic. 
\end{assump}

\noindent Note that $\mathcal{P}$ being a column stochastic matrix allows for a distributed synthesis of the consensus protocol. The variable $w_i^{k} \in \mathbb{R}^{p}$ is an estimate of $\widetilde{u}$ with each agent $i$ at any iteration $k$. It is established in prior work that the vector estimates $w_i^k$ converge to the average $\widetilde{u}$ asymptotically.
\begin{theorem} \label{thm:consensusconv}
Let Assumption~\ref{ass:graph_ass} hold. Let $\{w_i^{k}\}_{k \geq 0}$ be the sequence generated by~(\ref{eq:cons_w}) at each agent $i \in \mathcal{V}$. Then $w_i^{k}$ asymptotically converges to $\widetilde{u} = \frac{1}{n}\sum_{i=1}^n u_i^0$ for all $i \in \mathcal{V}$, i.e.,
\begin{align*}
    \lim_{k \rightarrow \infty} w_i^{k} = \textstyle \frac{1}{n}\sum_{i=1}^n u_i^0, \ \text{for all} \ i \in \mathcal{V}.
\end{align*}
\end{theorem}
\begin{proof}
Refer \cite{melbourne2020geometry}, \cite{kempe2003gossip} for the proof.
\end{proof}
\noindent The $\varepsilon$-consensus protocol is a distributed algorithm 
to determine when the agent states $w_i^k$, for all $i \in \mathcal{V}$ are $\varepsilon$-close to each other and hence from Theorem~\ref{thm:consensusconv}, $\varepsilon$-close to $\widetilde{u}$.
To this end, each agent maintains a scalar value $R_i$ termed as the radius of agent $i$. The motivation behind maintaining such a radius variable is as follows:
consider an open ball at iteration $k$ that encloses all the agent states, $w_i^k$, for all $i \in \mathcal{V}$, with a minimal radius (the existence of such a ball can be shown due to bounded nature of the updates~(\ref{eq:cons_u})-(\ref{eq:cons_w}), see Lemma 4.2 in \cite{benezit2010weighted}). Theorem~\ref{thm:consensusconv} guarantees that all the agent states $w_i^k$ converge to a single vector, which implies that the minimal ball enclosing all the states will also shrink  with iteration $k$ and eventually the radius will become zero. The radius variable $R_i$ is designed to track the radius of this minimal ball. Starting at an iteration $s$, the radius $R_i^k(s)$ is updated to $R_i^{k+1}(s)$ as follows: for all $k = 0,1,2 \dots$, with $R_i^{0}(s):=0$,
\begin{align}\label{eq:radius}
    R_i^{k+1}(s) = \textstyle \max_{j \in \mathcal{N}_i^-} \big \{ \| w_i^{s+k+1} - w_j^{s+k}\| + R_j^k(s) \big \},
\end{align}
for all $i \in \mathcal{V}$. Denote, $\mathcal{B}(w_i^{s+k}, R_i^{k}(s))$ as the ball of radius $R_i^{k}(s)$ centered at $w_i^{s+k}$. It is established in \cite{melbourne2020geometry} that, after $\mathcal{D}$ number of iterations of the update~(\ref{eq:radius}) the ball $\mathcal{B}(w_i^{s + \mathcal{D}}, R_i^{\mathcal{D}}(s))$ encloses the states $w_i^{s}$ of all the agents $i \in \mathcal{V}$. 
\begin{lemma}\label{lem:state_inside_ball}
Let $\{w_i^{k}\}_{k \geq 0}$ be the sequence generated by~(\ref{eq:cons_w}) at each agent $i \in \mathcal{V}$. Given, $s \geq 0$, let $R_i^{k}(s)$ be updated as in~(\ref{eq:radius}) for all $k \geq 0$ and $i \in \mathcal{V}$. Under Assumption~\ref{ass:graph_ass}, 
\begin{align*}
    w_j^{s} \in \mathcal{B}(w_i^{\mathcal{D}+s}, R_i^{\mathcal{D}}(s)), \ \text{for all} \ j,i \in \mathcal{V}. 
\end{align*}
\end{lemma}
\begin{proof}
See Appendix~{C}.
\end{proof} \noindent
\begin{align}\label{eq:radiushat}
 \mbox{Let}, \ \widehat{R}_i^m := R_i^\mathcal{D}(m\mathcal{D}), \ m = 0,1,2,\dots,    
\end{align}
where, $R_i^k(s)$ follows the update rule~(\ref{eq:radius}) for any $s \geq 0$. The next result establishes the fact that the sequence of radii $\{\widehat{R}_i^{\mathcal{D}}(m\mathcal{D})\}_{m \geq 0}$, for $i \in \mathcal{V}$ converges to zero as $m \rightarrow \infty$. 
\begin{theorem}\label{thm:eCons}
Let $\{w_i^{k}\}_{k \geq 0}, \{\widehat{R}_i^{m}\}_{m \geq 0}$ be the sequences generated by~(\ref{eq:cons_w}) and~(\ref{eq:radiushat}) respectively. Under Assumption~\ref{ass:graph_ass},  
\begin{align*}
  \textstyle \lim_{m \rightarrow \infty} \widehat{R}_i^{m} = 0, \ \mbox{for all} \ i \in \mathcal{V}. 
\end{align*}
Further, $ \lim\limits_{m \rightarrow \infty} R_i^{m} = 0$ if and only if $\lim\limits_{m \rightarrow \infty} \max_{i,j\in \mathcal{V}} \| w^{m \mathcal{D}}_i - w_j^{m \mathcal{D}}\| = 0$. 
\end{theorem}
\begin{proof}
See Appendix~{D}.
\end{proof}
 \begin{algorithm}[t]
\small
    \SetKwBlock{Initialize}{Initialize:}{}
    \SetKwBlock{Input}{Input:}{}
    \SetKwBlock{Repeat}{Repeat for $ k = 0,1,2, \dots$}{}
    \Input{ Pre-specified tolerance $\varepsilon > 0$;\\ Diameter upper bound $\mathcal{D}$; 
     \\
    } 
    \Initialize{$w_i^0= u_i^0$; $v_i^0=1$; $R_i^{0}= 0$; \\ $m := 1$}
    \Repeat {
    \tcc{consensus updates (\ref{eq:cons_u})-(\ref{eq:cons_w})}
    $u_i^{k+1} = p_{ii}u_i^{k}+\sum_{j\in\mathit{\mathcal{N}^-_{i}}}p_{ji}u_j^{k}$\\
    $v_i^{k+1} =p_{ii}v_i^{k}+\sum_{j\in\mathit{\mathcal{N}^-_{i}}}p_{ji}v_j^{k}$ \\
    $w_i^{k+1} = \frac{1}{v_i^{k+1}}u_i^{k+1}$\\
    \tcc{radius update (\ref{eq:radius})}
        $R_i^{k+1} = \max \limits_{j \in \mathcal{N}_i^-} \Big \{ \| w_i^{k+1} - w_i^{k}\| + R_j^{k} \Big \}$\\
        \If {$ k =  m \mathcal{D} - 1$} {
            $\widehat{R}_i^{m-1}=R_i^{k+1}$;\\
            \uIf {$\widehat{R}_i^{m-1} < \varepsilon$} {\textbf{break} \tcp*{$\varepsilon$-consensus achieved}
                }
            \uElse { 
             $R_{i}^{k+1} = 0$;\\
             $m = m + 1$;
            }
        }
    }
    \caption{Finite-time $\varepsilon$-consensus protocol at each node $i \in \mathcal{V}$ \cite{melbourne2020geometry}}
    \label{alg:econs}
\end{algorithm}
\noindent Theorem~\ref{thm:eCons} gives a criterion for termination of the consensus iterations~(\ref{eq:cons_u})-(\ref{eq:cons_w}) by utilizing the radius updates at each agent $i \in \mathcal{V}$ given by~(\ref{eq:radius}). In particular, by tracking the radii $\widehat{R}_i^m, m =0,1,2\dots$ each agent can determine an estimate of the radius of the minimal ball enclosing all the agent states (Lemma~\ref{lem:state_inside_ball}) distributively. Further, monitoring the value of the $\widehat{R}_i^m, m =0,1,2\dots$ each agent can have the knowledge of the state of consensus among all the agents. A protocol to determine $\varepsilon$-consensus among the agents is given in Algorithm~\ref{alg:econs}. Algorithm~\ref{alg:econs} is initialized with $w_i^0 = u_i^0, v_i^0 = 1$ and $R_i^0 = 0$ respectively for all $i \in \mathcal{V}$. Each agent follows rules~(\ref{eq:cons_u})-(\ref{eq:cons_w}) and~(\ref{eq:radius}) to update its state and radius variables respectively. 
$R_i^k$ is reset at each iteration of the form $ k = m \mathcal{D}, m = 1,2,\dots$ to have a value equal to $0$. The sequence of radii $\{ \widehat{R}_i^m \}_{m \geq 0}$ is determined by setting the value $\widehat{R}_i^{m-1}$ equal to $R_i^{k}$ for all iterations of the form $k = m\mathcal{D}, m = 1, 2, \dots$, i.e, $\widehat{R}_i^{m-1}=R_i^{m \mathcal{D}}$ for all $m \geq 1$. From Lemma~\ref{lem:state_inside_ball}, the ball $\mathcal{B}(w_i^{m \mathcal{D}}, \widehat{R}_i^{m-1})$ will contain all estimates $w_i^{(m-1) \mathcal{D}}$. This ball is the estimate of the minimal ball enclosing all agent states $w_i^{(m-1) \mathcal{D}}$. Therefore, as a method to detect $\varepsilon$-consensus, at every iteration of the form $m \mathcal{D}$, for $m = 1, 2, \dots$,
$\widehat{R}_i^{m-1}$ is compared to the parameter $\varepsilon$, if  $\widehat{R}_i^{m-1} < \varepsilon $ then all the agent states $w_i^{(m-1) \mathcal{D}}$ were $\varepsilon$-close to $\widetilde{u}$ (from Theorem~\ref{thm:consensusconv}) and the iterations~(\ref{eq:cons_u})-(\ref{eq:cons_w}) are terminated. Proposition~\ref{prop:finite-time} establishes that the $\varepsilon$-consensus converges in finite number of iterations.

\begin{proposition}\label{prop:finite-time}
Under the Assumption~\ref{ass:graph_ass}, $\varepsilon$-consensus is achieved in finite number of iterations at each agent $i \in \mathcal{V}$.
\end{proposition}
\begin{proof}
Note, $\widehat{R}_i^m \rightarrow 0$ as $m \rightarrow \infty$. Thus, given $\varepsilon >0, i \in \mathcal{V}$ there exists finite $\overline{k}_i^\varepsilon$ such that for $ m \geq \overline{k}_i^\varepsilon, \widehat{R}_i^m < \varepsilon$.
\end{proof}

\noindent Note, that the radius estimates $\widehat{R}_i^m$ can be different for some of the agents. Thus, the detection of $\varepsilon$-consensus can happen at different iterations for some nodes. In order to, have a global detection, each agent can generate "converged flag" indicating its own detection. Such a flag signal can then be combined by means of a distributed one-bit consensus updates (see \cite{melbourne2020geometry}), thus allowing the agents to achieve global $\varepsilon$-consensus. 

\end{subsection}

\begin{subsection}{ $\y$ variable updates in $\D$}\label{sec:y_update_sec}
\begin{algorithm}[b]
\small
    \SetKwBlock{Input}{Input:}{}
    \SetKwBlock{Initialize}{Initialize:}{}
    \SetKwBlock{Repeat}{Repeat for $ k = 0,1,2, \dots$}{}
    \SetKwBlock{STEPTWO}{STEP 2:}{}
    \newcommand{\inparallel}{\textbf{In parallel}}
    \newcommand{\until}{\textbf{until}}
    \Input{$ \text{choose} \ \gamma>0, \{ \eta_k \}_{k \geq 0}$ }
    \Initialize{$\textcolor{black}{x_i^0 \in \mathbb{R}^{(m+p)}, y_i^0 \in \mathbb{R}^{(m+p)}, \lambda_i^0 = \mathbf{0}_{m+p} \in \mathbb{R}^{(m+p)}} , \mu_i^0 \in \mathbb{R}^m,  \forall i \in \V $}
    \Repeat {
            \For{$ i = 1,2,\dots n$, (\inparallel)} 
            {$x_i^{k+1} = \textstyle \argmin\limits_{x_i} \Big\{ f_i(x_i) + \lambda_i^{k^\top} (x_i - y_i^{k}) \textstyle + \frac{\gamma}{2} \| x_i - y_i^k\|^2 + \mu_{i}^{k^\top} (\textcolor{black}{A_i} x_i - \textcolor{black}{b_i}) \textstyle + \frac{\gamma}{2} \| \textcolor{black}{A_i} x_i - \textcolor{black}{b_i} \|^2 \Big\},$\\
            $y_i^{k+1} \xleftarrow[]{\text{$\eta_{k+1}$  - consensus protocol}}(x_i^{k+1} + \frac{1}{\gamma}  \lambda_i^{k})$\\
            $\lambda_i^{k+1} = \lambda_i^{k} + \gamma(x_i^{k+1} -  y_i^{k+1}) $ \\
            $\mu_i^{k+1} = \mu_i^{k} + \gamma (\textcolor{black}{A_i} x_i^{k+1} - \textcolor{black}{b_i})$
            }
            \until \ a stopping criterion is met
    }
        \caption{$\D$ Algorithm}
        \label{alg:distADMM}
\end{algorithm}
The finite-time $\varepsilon$-consensus protocol discussed above is utilized to determine an inexact solution to the update~(\ref{eq:dadmm_y}). At any iteration $k \geq 0$ of the $\D$ algorithm, each agent $i \in \mathcal{V}$ runs an $\varepsilon$-consensus protocol with the tolerance $\varepsilon = \eta_{k+1}$ and the following initialization:
\begin{align} \label{eq:cons_opt_initialization}
    \textstyle u_i^0 = x_i^{k+1}+ \frac{1}{\gamma} \lambda_i^{k}, \ v_i^0 = 1, \ \mbox{and} \ w_i^0 = u_i^0. 
\end{align}
From Proposition~\ref{prop:finite-time}, there exists finite $t_k$ such that $\forall i,j \in \mathcal{V}$,
\begin{align}
    \| w_i^{t_{k}}  -  w_j^{t_{k}}\| \leq \eta_{k+1}, \ \ \mbox{and} \ \ \| w_i^{t_{k}} - \overline{y}^{k+1} \| \leq \eta_{k+1}, \label{eq:cons_sol}
\end{align}
where, $\overline{y}^{k+1} = \frac{1}{n}\sum_{i=1}^n [x_i^{k+1} + \frac{1}{\gamma} \lambda_i^k]$. The variable $y_i^{k+1}$ of agent $i$ is then set equal to the consensus variable $w_i^{t_k}$, i.e. $y_i^{k+1} = w_i^{t_k}$. Therefore, the obtained vector $\y^{k+1} = [y_1^{k+1^\top}, \dots, y_n^{k+1^\top}]^\top \in \textcolor{black}{\mathbb{R}^{n(m+p)}}$ is an \textit{inexact} solution to the update~(\ref{eq:dadmm_y}) (see Appendix~{E}). Further, $\y^{k+1}$ is $\sqrt{n}\eta_{k+1}$-close to the exact solution $\overline{\y}^{k+1}$, i.e.,
\begin{align}\label{eq:error}
     \y^{k+1} = \overline{\y}^{k+1} + e^{k+1}, \ \mbox{with} \ \|e^{k+1}\| \leq \sqrt{n}\eta_{k+1}.
\end{align}
The $\D$ algorithm is presented in~Algorithm~\ref{alg:distADMM}.
\end{subsection}
\end{section}

\begin{section}{Convergence Results for $\D$}\label{sec:convgAnalysis}
In this section, convergence result for the proposed $\D$ algorithm is presented. Let \textcolor{black}{$\mathcal{X} := \mathcal{X}_1 \times \dots \times \mathcal{X}_n, ^n\mathbb{R}^m_{\geq 0} := \underbrace{\mathbb{R}_{\geq 0}^m \times \dots \times \mathbb{R}_{\geq 0}^m}_{\text{n \ times}}$} and \textcolor{black}{$x^* \in \mathbb{R}^{(m+p)}$} denote an optimal solution of problem~(\ref{eq:introprob}). Throughout the rest of the paper the following assumptions hold:
\begin{assump}\label{ass:constraint_ass}
Each $\mathcal{X}_i$ is a closed bounded convex set with diameter $\diam_{\|\cdot \|}(\mathcal{X}_i) = \mathcal{M}_i$. 
\end{assump}
\noindent Note that Assumption~\ref{ass:constraint_ass} can be relaxed if the set of optimal solutions of the original distributed problem~(\ref{eq:introprob}) is bounded (see \cite{bertsekas1997nonlinear}); in this case  the existence of a bound on the local variables can be inferred from the bound on optimal solutions. Here this route is not taken, as in many multi-agent distributed optimization applications the variables are required to remain within specified bounds. Assumption~\ref{ass:constraint_ass} is motivated by the constraint set requirement in many practical applications.
\begin{assump}\label{ass:fconvex}
The function $\tilde{f}_i, i \in \{1,\dots,n\}$, is a proper closed convex function, which is not necessarily differentiable.
\end{assump} 
\begin{assump}\label{ass:prob_ass1}
There exists a saddle point $(\x^*,\y^*,\lambda^*,\mu^*)$ for the Lagrangian function $\La$ defined in~(\ref{eq:lag}), i.e., for all $(\x,\y) \in \textcolor{black}{\mathbb{R}^{n(m+p)} \times \mathbb{R}^{n(m+p)} }$, $(\lambda, \mu) \in  \textcolor{black}{ \mathbb{R}^{n(m+p)} \times \mathbb{R}^{nm} }$, where, 
\begin{align*}
    \La(\x^*,\y^*,\lambda,\mu) \leq \La(\x^*,\y^*,\lambda^*,\mu^*) \leq \La(\x,\y,\lambda^*,\mu^*). 
\end{align*}
\end{assump}
\noindent Note that since~(\ref{eq:distOpt_indifunc1}) is a convex optimization problem, the existence of dual optimal solutions is guaranteed if a constraint qualification condition like Slater's CQ holds for problem~(\ref{eq:distOpt_indifunc1}) \cite{rockafellar2015convex}. \textcolor{black}{Note that the indicator function of a set is a convex but not differentiable function.}
\noindent For further analysis of the $\D$ algorithm consider the following relation: for $x_1, x_2, x_3, x_4 \in \mathbb{R}^p$ 
\begin{align}\label{eq:loc}
    (x_1 - x_2)^\top(x_3 - x_4) & = \textstyle \frac{1}{2}\big(\|x_1 - x_4 \|^2 - \|x_1 - x_3 \|^2\big) \nonumber  \\
    & \hspace{-0.2in} + \textstyle \frac{1}{2}\big(\|x_2 - x_3 \|^2 - \|x_2 - x_4 \|^2 \big).
\end{align}

\noindent Under Assumption~\ref{ass:fconvex}, the first-order optimality and primal feasibility conditions for~(\ref{eq:distOpt_indifunc1}) are:
\textcolor{black}{
\begin{align}
   \hspace{-0.2in} - \lambda_i^* - \textcolor{black}{A_i}^\top\mu_{i}^* & \in \partial f_i(x_i^*), \ i = 1,2,\dots,n, \label{eq:xsolnoptcond}\\
    \lambda^* & \in \partial \mathcal{I}_{\mathcal{C}_0}(\y^*),\label{eq:ysolnoptcond} \\
     x_i^* &= y_i^* , \ i  = 1,2,\dots,n, \\
    \textstyle \textcolor{black}{A_i}x_i^* - \textcolor{black}{b_i} &= 0, \ i  = 1,2,\dots,n,\label{eq:eqconst_opt}
\end{align}
where, $\partial f_i(x_i^*)$ and $\partial \mathcal{I}_{\mathcal{C}_0}(\y^*)$ are the set of all the sub-gradients of $f_i$ at $x_i^*$, and $\mathcal{I}_{\mathcal{C}_0}$ at $\y^*$ respectively.} Similarly, for the $\D$ updates~(\ref{eq:dadmm_x}), and~(\ref{eq:dadmm_y}) by the first-order optimality condition t follows that, for all $i \in \{1,\dots, n\}$:
\textcolor{black}{
\begin{align}
    & - (\lambda_i^k + \textcolor{black}{A_i}^\top\mu_{i}^k + \gamma (x_i^{k+1} - y_i^k) + \gamma \textcolor{black}{A_i}^\top(\textcolor{black}{A_i}x_i^{k+1} - \textcolor{black}{b_i})) \nonumber \\
    & \hspace{2in} \in \partial f_i(x_i^{k+1}), \label{eq:iteroptcond1}\\
    & - ( \lambda_i^{k+1} + \gamma (y_i^{k+1} - y_i^k) + \textcolor{black}{A_i}^\top \mu_i^{k+1} )  \in \partial f_i(x_i^{k+1}), \label{eq:iteroptcond2}\\
    & \hspace{0.15in} \gamma (\x^{k+1} - \overline{\y}^{k+1}) + \lambda^k \in \partial \mathcal{I}_{\mathcal{C}_0}(\overline{\y}^{k+1}) \nonumber \\
    & \hspace{0.75in} \implies  \overline{\lambda}^{k+1} \in \partial \mathcal{I}_{\mathcal{C}_0}(\overline{\y}^{k+1}),  \label{eq:iteroptcond3}
\end{align}}where, $\partial \mathcal{I}_{\mathcal{C}_0}(\overline{\y}^{k+1})$ is the set of all sub-gradients of $\mathcal{I}_{\mathcal{C}_{0}}$ at $\overline{\y}^{k+1}$. The relation~(\ref{eq:iteroptcond2}) results from combining~(\ref{eq:iteroptcond1}) with the update~(\ref{eq:dadmm_lam}), and in~(\ref{eq:iteroptcond3})~(\ref{eq:error}) is employed.

\begin{lemma}\label{lem:ylambdabound}
Let $\{\x^{k}\}_{k \geq 1}$, $\{\y^{k}\}_{k \geq 1}$, $\{\lambda^{k}\}_{k \geq 1}$, and $\{\mu^{k}\}_{k \geq 1}$, be the sequences generated by Algorithm~\ref{alg:distADMM}. Let $\{\eta_k\}_{k \geq 1}$ be the sequence of tolerances in Algorithm~\ref{alg:distADMM}. Under assumptions~\ref{ass:graph_ass} and~\ref{ass:constraint_ass}, there exists $\mathbf{Q}$, independent of $k$, such that,
\begin{align*}
 \|\lambda^k\| \leq k \mathbf{Q}, \ \mbox{for all} \ k.  
\end{align*}
\end{lemma}
\begin{proof}
See Appendix~{F}.
\end{proof}
\begin{lemma}[\textbf{Iteration complexity of the $\y^{k+1}$ update~(\ref{eq:cons_sol})}]\label{lem:cons_comm} 
Under assumptions~\ref{ass:graph_ass} and~\ref{ass:constraint_ass} at iteration $k$ of Algorithm~\ref{alg:distADMM}, after $\overline{t_k} := \ceil*{  \textstyle \frac{\log\big(\frac{1}{\eta_{k+1}}\big)}{-\log \alpha} + \frac{\log\big(\frac{8((\mathcal{M} + \textcolor{black}{\|x^0\|}) n + k \sqrt{n} \mathbf{Q}/\gamma)}{\beta} \big)}{-\log \alpha}}$ iterations of the consensus protocol (updates~(\ref{eq:cons_u})-(\ref{eq:cons_w})) with the initial condition, $u_i^0 = x_i^{k+1} + \frac{1}{\gamma} \lambda_i^{k}, \ v_i^0 = 1,$ for all $i \in \mathcal{V}$ and $\varepsilon = \eta_{k+1}$ it follows that: 
\begin{align*}
 \|y_i^{k+1} - \overline{y}_i^{k+1}\|  \leq \textstyle \eta_{k+1}, \ \forall i \in \mathcal{V}.
\end{align*}
Here, $\mathcal{M} := \max_{1 \leq i \leq n} \mathcal{M}_i$, $\eta_{k+1}$ denotes the inaccuracy introduced in solving relaxed version of problem~(\ref{eq:dadmm_y}), $\mathbf{Q}$ is a finite constant as given in Lemma~\ref{lem:ylambdabound}, and $\beta >0$, $\alpha \in (0,1)$ are parameters of the graph $\mathcal{G}(\mathcal{V},\mathcal{E})$ satisfying $
    \beta \geq \frac{1}{n^n}, \alpha \leq \textstyle \left(1 - \frac{1}{n^n} \right)$.
\end{lemma}
\begin{proof}
See Appendix~{G}.
\end{proof}

Let the sequence of accuracy tolerances $\{\eta_k\}_{k \geq 1}$ in Algorithm \ref{alg:distADMM} be such that 
\begin{align}\label{eq:summable}
   \textstyle \mbox{for all}  \ k \geq 0, \eta_{k+1} \leq \eta_{k}, \sum_{k=1}^{\infty} k \eta_k < \infty.
\end{align}
Some examples of sequences that satisfy~(\ref{eq:summable}) include:
\begin{align*}
    \textstyle \mbox{i)} \ \eta_k = \frac{1}{k^{2+q}}, q > 0, \ \mbox{ii)} \ \eta_k = \rho^k, \rho \in (0,1), \ \mbox{iii)} \ \eta_k = \frac{\rho^k}{k!}.
\end{align*}
In the following, the global convergence of the proposed $\D$ algorithm to a solution of problem~(\ref{eq:distOpt_indifunc1}) is established.
\begin{theorem}[\textbf{Convergence of iterates generated by $\D$ algorithm to an optimal solution}]\label{thm:convergence}
Let $\{\x^{k}\}_{k \geq 1}$, $\{\y^{k}\}_{k \geq 1}$, $\{\lambda^{k}\}_{k \geq 1}$, $\{\mu^{k}\}_{k \geq 1}$, be the sequences generated by Algorithm~\ref{alg:distADMM}. Let the consensus tolerance sequence $\{\eta_k\}_{k \geq 1}$ satisfy~(\ref{eq:summable}). Under assumptions~\ref{ass:graph_ass}-\ref{ass:prob_ass1}, $(\x^k,\y^k,\lambda^k,\mu^k)$ converges to a solution $(\x^\infty,\y^\infty,\lambda^\infty,\mu^\infty)$ of~(\ref{eq:distOpt_indifunc1}), i.e., $\F(\x^\infty) = \F(\x^*), \y^\infty \in \mathcal{C}_0, \x^\infty = \y^\infty$, and $\A \x^\infty = \bb$.  
\end{theorem}
\begin{proof}
See Appendix~\ref{sec:proof_convergence}.
\end{proof} 
\noindent Next two estimates of rate of convergence for the proposed $\D$ algorithm are provided. 

\begin{subsection}{Sub-linear Rate of Convergence}
Here, the convergence of the proposed $\D$ algorithm for the case when individual functions $\tilde{f}_i$ are convex but not necessarily differentiable is analyzed.
\begin{theorem}[\textbf{Sub-linear rate of convergence}]\label{thm:order1/k}
Let $\{\x^{k}\}_{k \geq 1}$, $\{\y^{k}\}_{k \geq 1}$, $\{\lambda^{k}\}_{k \geq 1}$, and $\{\mu^{k}\}_{k \geq 1}$, be the sequences generated by Algorithm~\ref{alg:distADMM}. Let $\eta_k = \frac{1}{k^{2+q}}, q \in (0,1)$. Let assumptions~\ref{ass:graph_ass}-\ref{ass:prob_ass1} hold, then for all $k \geq 1$, 
\begin{align*}
    & \textstyle \F(\widehat{\x}^k) - \F(\x^*) = O(1/k), \ \  \|\widehat{\x}^k - \widehat{\y}^k\| = O(1/k), \\
    & \hspace{0.65in} \mbox{and}, \  \textstyle \| \A\widehat{\x}^k - \bb \| = O(1/k).
\end{align*}
where, $\widehat{\x}^k := \frac{1}{k}\sum_{s=0}^{k-1}\x^{s+1}$, $\widehat{\y}^k := \frac{1}{k}\sum_{s=0}^{k-1} \y^{s+1}$.
\end{theorem}
\begin{proof}
See Appendix~\ref{sec:proof_order1/k}.
\end{proof}
\noindent Theorem~\ref{thm:order1/k} establishes that the objective function evaluated at the ergodic average of the optimization variables obtained by the proposed $\D$ algorithm converges to the optimal value. In particular, the objective function value evaluated at the ergodic average converges to the optimal value at a rate of $O(1/k)$.
\begin{remark}
Let, $\hat{x}_i^{k} = \frac{1}{k} \sum_{s=0}^{k-1} x_i^s$ and $\hat{y}_i^{k} = \frac{1}{k} \sum_{s=0}^{k-1} y_i^s$ denote the ergodic averages of the variables $x_i^{k}$ and $y_i^{k}$. Since, $\|\widehat{\x}^k - \widehat{\y}^k \| = O(1/k)$ it implies that for all $i \in \mathcal{V}, \|\hat{x}_i^k - \hat{y}_i^k \| = O(1/k)$. \textcolor{black}{Further, since, for all $k \geq 0$ and $i,j \in \mathcal{V}, \|y_i^k - y_j^k\| \leq \eta_k$ (see~(\ref{eq:cons_sol})) it implies that $\| \hat{y}_i^k - \hat{y}_j^k \| \leq \frac{1}{k} \sum_{s=0}^{k-1} \| y_i^s - y_j^s \| \leq \frac{1}{k} \sum_{s=0}^{k-1} \eta_s \leq \frac{\zeta(2+q)}{k} $.} Thus, $\| \hat{y}_i^k - \hat{y}_j^k \| = O(1/k)$. Therefore, for all $ i,j \in \mathcal{V}$,
\begin{align*}
    \| \hat{x}_i^k - \hat{x}_j^k \| = O(1/k) \ \mbox{and} \  \lim_{k \rightarrow \infty} \| \hat{x}_i^k - \hat{x}_j^k \| = 0.
\end{align*}
Thus, in practice, the $\D$ algorithm can be implemented with an additional variable $\hat{x}_i^k$ tracking the ergodic average of $x_i^{k}$ at each agent $i \in \mathcal{V}$ to achieve a consensual solution at the $O(1/k)$ rate of convergence given in Theorem~\ref{thm:order1/k} in a fully distributed manner. Note, that the variable $\hat{x}_i^{k}$ is computed locally by each agent without any additional cost of communication and a minor addition in cost of computation.
\end{remark}

\begin{remark}\label{rem:1/k_comm}
By Lemma~\ref{lem:cons_comm} the total number of communication iterations performed until iteration $k$ of the $\D$ algorithm in Theorem~\ref{thm:order1/k}, is upper bounded by $\mathcal{K}:= \sum_{s=1}^k \overline{t_s} := O(k \log k).$
\textcolor{black}{The communication complexity is within a factor of $\log 
k$ of the optimal lower bound $O(k)$.}  
\end{remark}
\end{subsection}
\vspace{-0.2in}
\begin{subsection}{\textcolor{black}{Geometric} rate of Convergence}
\textcolor{black}{Here, a geometric rate of convergence for the proposed $\D$ algorithm under the following assumption is established:}

\begin{assump}\label{ass:lipschitzgrad_f_strconv}
Each function $\tilde{f}_i$  is Lipschitz differentiable with constant $L_{f_i}$ and \textcolor{black}{restricted strongly convex with respect to the optimal solution $x^*$ on $\mathcal{X}_i$ with parameter $\sigma_i > 0$}.
\end{assump}
\vspace{-0.25in}
\textcolor{black}{
\begin{remark}
Under Assumption~\ref{ass:lipschitzgrad_f_strconv} the problem~(\ref{eq:introprob}) has a unique optimal solution. However, Assumption~\ref{ass:lipschitzgrad_f_strconv} is less restrictive than the standard global strong convexity assumption and makes the analysis presented here applicable to a bigger class of functions \cite{zhang2015restricted}. For example, the widely used logistic regression objective function is restricted strongly convex but not globally strongly convex \cite{bach2014adaptivity}.
\end{remark}
\begin{theorem}[\textbf{Geometric rate of convergence}]\label{thm:linearrate}
Let $\{\x^{k}\}_{k \geq 1}$, $\{\y^{k}\}_{k \geq 1}$, $\{\lambda^{k}\}_{k \geq 1}$, and $\{\mu^{k}\}_{k \geq 1}$, be the sequences generated by Algorithm~\ref{alg:distADMM}. Let assumptions~\ref{ass:graph_ass},\ref{ass:constraint_ass},\ref{ass:prob_ass1}, and \ref{ass:lipschitzgrad_f_strconv} hold. Let $\Delta:= (1 - \frac{1}{\delta}) \min\{ 1, \nu_{\min}(\A\A^\top)\}$ for any $\delta \in \big(1, 1 + \frac{2\sigma \gamma}{L^2 + \gamma^2}\big],$ where $L:= \max_{1 \leq i \leq n} L_{f_i}, \sigma := \min_{1 \leq i \leq n} \sigma_i$, and $\nu_{\min}(\A\A^\top)$ is the minimum eigenvalue of $\A\A^\top$. Let $\eta_k = \rho^k$, where, $\rho \in \left[\frac{1}{1+\Delta},1\right)$. Let $x_i^{k} \in \relint(\mathcal{X}_i \times \mathbb{R}_{\geq 0}^m), \forall i \in \mathcal{V}, k \geq 0.$ Then the agent solution residual $s_k:= \frac{\gamma}{2}  \|\x^{k} - \x^* \|^2 + \frac{1}{2\gamma} \|\lambda^{k} - \lambda^*\|^2 + \frac{1}{2 \gamma} \| \mu^{k} - \mu^*\|^2$ has the following relation: for any $K \geq 0$ and $\epsilon > 0$,
\begin{align*}
     \textstyle  s_K \leq \textstyle  \Upsilon \rho^K + O(\epsilon ),
\end{align*}
where, $\Upsilon$ is a finite constant defined in (\ref{eq:upslion}).
\end{theorem}
\begin{proof}
See Appendix~\ref{sec:proof_linearrate}.
\end{proof}
}
\begin{remark}\label{rem:rho_interval}
Although, Theorem~\ref{thm:linearrate} provides a tight bound on the requirement of $\rho$, in practice the tolerance sequence parameter $\rho$ can be chosen from the interval $\left[\frac{2}{3}, 1 \right)$  which does not require the knowledge of the restricted strong convexity parameter $\sigma_i$ of the individual functions and the optimal solution $x^*$. We explain this below:\\
    From the definition of $\Delta$ we have, 
    $\Delta \leq (1- 1/\delta) $. Using the condition on $\delta$ in Theorem~\ref{thm:linearrate} we get,
    \begin{align*}
        1 & < \delta \leq \frac{L^2 + \gamma^2 + 2 \sigma \gamma}{L^2 + \gamma^2 } \\
        1 & > \frac{1}{\delta} \geq \frac{L^2 + \gamma^2}{L^2 + \gamma^2 + 2 \sigma \gamma} \\
        -1 & < -\frac{1}{\delta} \leq -\frac{L^2 + \gamma^2}{L^2 + \gamma^2 + 2 \sigma \gamma} \\
        0 & < 1 - \frac{1}{\delta} \leq \frac{2\sigma \gamma}{L^2 + \gamma^2 + 2 \sigma \gamma} =  \frac{1}{\frac{L^2 + \gamma^2}{2\sigma \gamma} + 1} \leq \frac{1}{2},
    \end{align*}
    where, the last inequality follows from the AM-GM inequality and the fact that $L \geq \sigma$. Thus, for $\frac{1}{\Delta + 1} = \left(1 + \left(1 - \frac{1}{\delta}\right) \right)^{-1}$
    \begin{align*}
       \frac{2}{3} \leq \frac{1}{\Delta + 1} < 1.
    \end{align*}
    Note that the $\rho \in \left[\frac{1}{1+\Delta} , 1 \right)$. Thus, $\rho \in \left[\frac{2}{3}, 1 \right)$. We have utilized $\rho = 0.75$ (and hence, $\eta_k = 0.75^k$) in the numerical simulations presented in the article (see Section~\ref{sec:results}). The simulation results demonstrate the theory's applicability and show good performance for the $\D$ algorithm with $\rho = 0.75 \in \left[\frac{2}{3}, 1 \right)$.
\end{remark}
\begin{remark}\label{rem:linear_comm}
By Lemma~\ref{lem:cons_comm} the total number of communication iterations performed until iteration $k$ of the $\D$ algorithm in Theorem~\ref{thm:linearrate}, is upper bounded by $\mathcal{K}:= \sum_{s=1}^k \overline{t_s} := O(k^2).$  Thus, compared to Theorem~\ref{thm:order1/k} the improved rate of convergence leads to an increase in the number of communication iterations (Remark~\ref{rem:1/k_comm}). \textcolor{black}{Moreover, compared to algorithms utilizing multiple communication steps in the literature \cite{berahas2018balancing, berahas2021convergence,chen2012fast, li2018sharp,ye2020multi, jakovetic2014fast, johansson2008subgradient, khatana2020gradient} the $\D$ algorithm has the same communication complexity. In particular, methods in \cite{berahas2018balancing, berahas2021convergence, chen2012fast} have communication complexity, $O(k^2)$, in getting a geometric rate. Schemes proposed in \cite{li2018sharp, ye2020multi, jakovetic2014fast, johansson2008subgradient,khatana2020gradient} have the same, $O(k \log k)$, communication complexity in getting a $O(1/k)$ rate of convergence.}
\end{remark}
\begin{remark}
The geometric rate of convergence for $\D$ algorithm does not follow from the existing centralized results (see \cite{lin2015global}, \cite{rockafellar1976augmented}) as problem~(\ref{eq:distOpt_indifunc1}) does not satisfy the assumptions used in these works. In particular, the function $\mathcal{I}_{{\mathcal{C}}_0}(\y)$ is not differentiable and does not meet the requirements of being global Lipschitz diffferentiable, twice continuously differentiable and strongly convex used in \cite{lin2015global}, \cite{rockafellar1976augmented}. Further, unlike \cite{lin2015global}, \cite{rockafellar1976augmented}, the $\y$ variable update in the ADMM scheme of $\D$ algorithm in an inexact manner  is solved via the $\varepsilon$-consensus protocol that adds additional complexities in the analysis. These reasons have motivated us to present the analysis of the $\D$ given in Theorem~\ref{thm:linearrate}.  
\end{remark}

\end{subsection}

\end{section}

\begin{section}{Numerical Simulations}\label{sec:results}
In this section, \textcolor{black}{three} simulation studies are presented for the proposed $\D$ algorithm. First, a performance \textcolor{black}{comparison} of the $\D$ algorithm \textcolor{black}{with two existing algorithms in the literature for solving constrained distributed optimization problems, \cite{liu2017constrained} and \cite{zhu2011distributed}, in solving a distributed $\ell_1$ regularized logistic regression with a local linear equality and set (norm-ball) constraints is presented.} 

Second, a performance comparison of the $\D$ algorithm \textcolor{black}{is provided for solving an unconstrained $\ell_1$ regularized Huber loss minimization problem with the existing state-of-the art unconstrained distributed optimization algorithms on directed graphs. The algorithms used for comparison with the proposed $\D$ algorithm are the following: (i) EXTRAPush \cite{zeng2015extrapush}, (ii) PushPull \cite{pu2020push}, (iii) PushDIGing  \cite{nedic2017achieving}, and (v) the subgradientPush \cite{nedic2014distributed} algorithm.} 

\textcolor{black}{Third, a comparison between $\D$ algorithm and two existing algorithms utilizing multiple communication steps, \cite{berahas2018balancing} and \cite{jakovetic2014fast} in solving the unconstrained distributed least squares problem is presented.} 

A network of $100$ agents connected via:
(i) an undirected graph generated using an Erdos-Renyi model \cite{erdHos1960evolution} with a connectivity probability of $0.3$ in simulation studies one and three and (ii) a directed graph generated using an Erdos-Renyi model with a connectivity probability of $0.2$ in the simulation study two is considered. The weight matrices for the various algorithms are chosen using the equal neighbor model \cite{olshevsky2009convergence}. To provide a comparison, solution residual plots against the total (computation + communication) iteration counts for all the algorithms unless stated otherwise are presented. Further, a comparison based on the CPU time (the amount of time required by a computer (processor) to execute the instructions) between $\D$ and the other algorithms is provided. 

All the numerical examples in this section are implemented in MATLAB, and run on a desktop computer with 16 GB RAM and an Intel Core i7 processor running at 1.90 GHz. The parameters used in the simulations for all the algorithms are reported in Table~\ref{tab:parameters}. \textcolor{black}{In choosing the step-sizes of the algorithms in Table~\ref{tab:parameters} we followed the approach of hand-optimizing the hyper-parameters that produce good performance for these algorithms while maintaining the stability of the algorithmic estimates for the class of problems under consideration.}

\subsection{Performance on constrained optimization problem}
\textcolor{black}{Consider a $\ell_1$ regularized distributed logistic regression with linear equality constraint and local inequality constraints: 
\begin{align}\label{eq:logit}
   \hspace{-0.1in} \textstyle \minimize \sum_{i=1}^n \sum_{j = 1}^{n_i} \log \big( & 1 + \text{exp} \big( -y_{ij}(a^{\top}_{ij} x) \big) \big)  + \theta \|x\|_1 \nonumber \\
    \mbox{subject to} \ H_i x = h_i,\ & x^{\top} x \leq r_i, \mbox{ for all} \ i \in \mathcal{V},
\end{align}
}
\hspace{-0.1in} where, each agent $i \in \mathcal{V}$ has $n_i$ data samples $\{(a_{ij}, y_{ij})\}_{j = 1}^{n_i}$, with, $a_{ij} \in \mathbb{R}^{50}$ being the feature vector and $y_{ij} \in \{-1,+1\}$ is the binary outcome (or class label) of the feature vector $a_{ij}$. The objective is to learn the weight vector $x \in \mathbb{R}^{50}$ based on the available data $\{(a_{ij}, y_{ij})\}_{j = 1}^{n_i}, i \in \mathcal{V}$ such that $x$ is sparse. The parameter $\theta$ enforces the sparsity in $x$. \textcolor{black}{Further, the solution satisfy the linear equality constraints $H_i x = h_i$, where, $H_i \in \mathbb{R}^{500 \times 50}$. Denote $\mathbf{H} := \mathbf{I} \otimes H_i, \mathbf{h} = [h_1^\top, \dots, h_{n}^\top]^\top$}. Each agent $i \in\mathcal{V}$ also has a local set constraint $x^\top x \leq r_i$. For our simulation we generate an instance of problem~(\ref{eq:logit}) where each agent has $n_i = 10^6$ data samples. Each feature vector $a_{ij} \in \mathbb{R}^{50}$ and the `true' weight vector $x_{\text{true}} \in \mathbb{R}^{50}$ are generated to have approximately $40 \%$ zero entries. The non-zero entries of $a_{ij}$ and $x_{\text{true}}$ are sampled independently from the standard normal distribution. The class labels $y_{ij}$ are generated using the equation: $ y_{ij} = \text{sign} (a_{ij}^\top x_{\text{true}} + \delta_{ij})$, where $\delta_{ij}$ is a normal random variable with zero mean and variance $0.1$. The parameter $\theta$ is set to $0.1 \theta_{\text{max}}$, where $\theta_{\text{max}}$ is the critical value above which the solution $x^* = 0$ (see \cite{MAL-016} section 11.2 for the calculation of $\theta_{\text{max}}$). The value of $r_i$ in the constraints is chosen as follows: the unconstrained version of~(\ref{eq:logit}) with a centralized solver and denote the solution as $x_{\text{c}}^*$. The constraint $r_i$ is set as $r_i = (1 + \xi_i)x_{\text{c}}^*$, where $\xi_i$ is drawn from a uniform distribution on $[0,1]$. A value of $\gamma = 10$ is used for the Augmented Lagrangian. \textcolor{black}{The entries of each $H_i$ are sampled independently from the standard normal distribution. The vector $h_i, i \in \mathcal{V}$ is calculated as $h_i = H_i x_{\text{true}}$.} To solve~(\ref{eq:logit}) using the $\D$ algorithm, the fast iterative shrinkage thresholding algorithm (FISTA) \cite{beck2009fast} is used for the updates~(\ref{eq:dadmm_x}) at each agent $i$. A proximal residue (difference between the output of the proximal minimization step and the base-point of the proximal term) \cite{beck2009fast} is utilized as a stopping criterion for FISTA. The FISTA iterations are terminated when the proximal residue becomes less than $10^{-4}$. A fixed accuracy for FISTA is used at all iterations of $\D$ algorithm. For the updates~(\ref{eq:dadmm_y}) the $\varepsilon$-consensus protocol with the tolerance $\varepsilon$ set equal to a desired level of inaccuracy $\eta$ in the solution of~(\ref{eq:dadmm_y}) is employed. 
\begin{figure}[b] 
    \centering
  \subfloat{%
  \hspace{-0.1in}
       \includegraphics[scale=0.205,trim={0.8cm 0cm 0.4cm 0.5cm},clip] {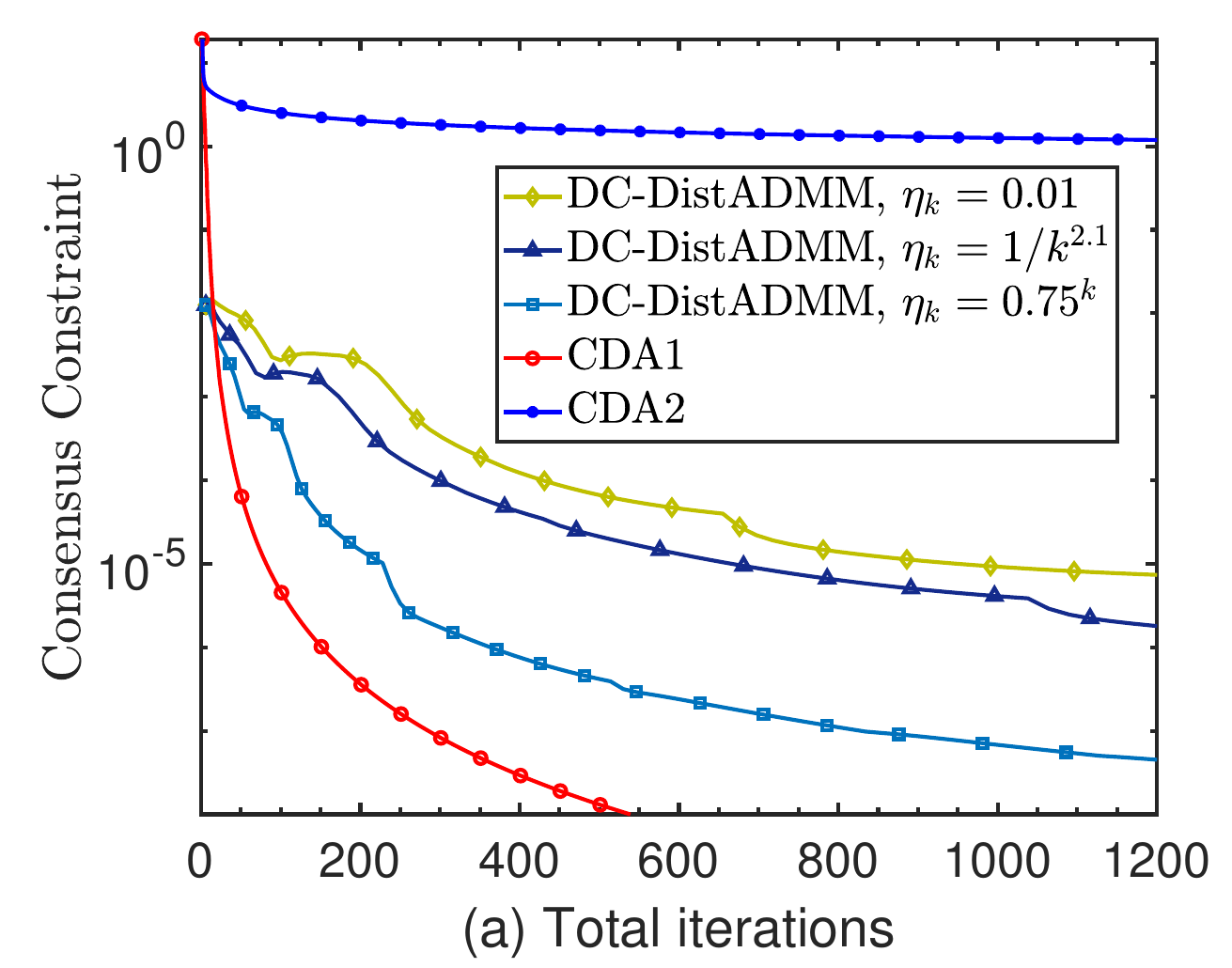}} \hspace{-0.05in}
  \subfloat{%
    \includegraphics[scale=0.205,trim={0.8cm 0cm 0.5cm 0.2cm},clip] {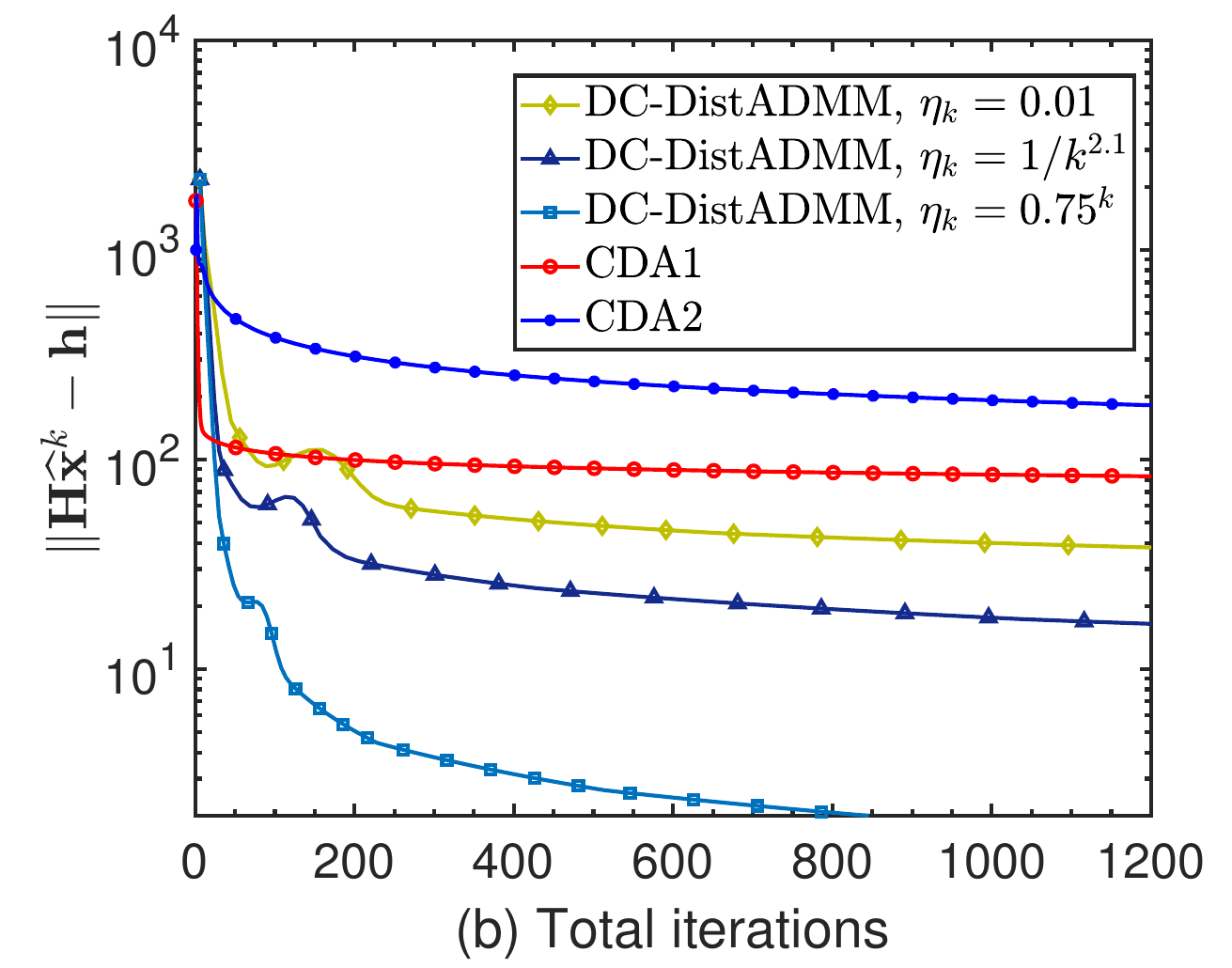}}
  \caption{(a) Comparison of consensus constraint residuals against total iterations (including communication steps) between $\D$ and constrained optimization algorithms (b) Comparison of equality constraint residuals against total iterations between $\D$ and constrained optimization algorithms. $\eta^k=0.75^k$ gives the best performance for the $\D$ algorithm in terms of the equality constraint residual. The $\D$ algorithm is second best in terms of the consensus constraints. }
  \label{fig:logit_constraints} 
\end{figure}
\begin{figure}[t] 
    \centering
  \subfloat{%
  \hspace{-0.1in}
       \includegraphics[scale=0.205,trim={0.2cm 0cm 0cm 0.1cm},clip] {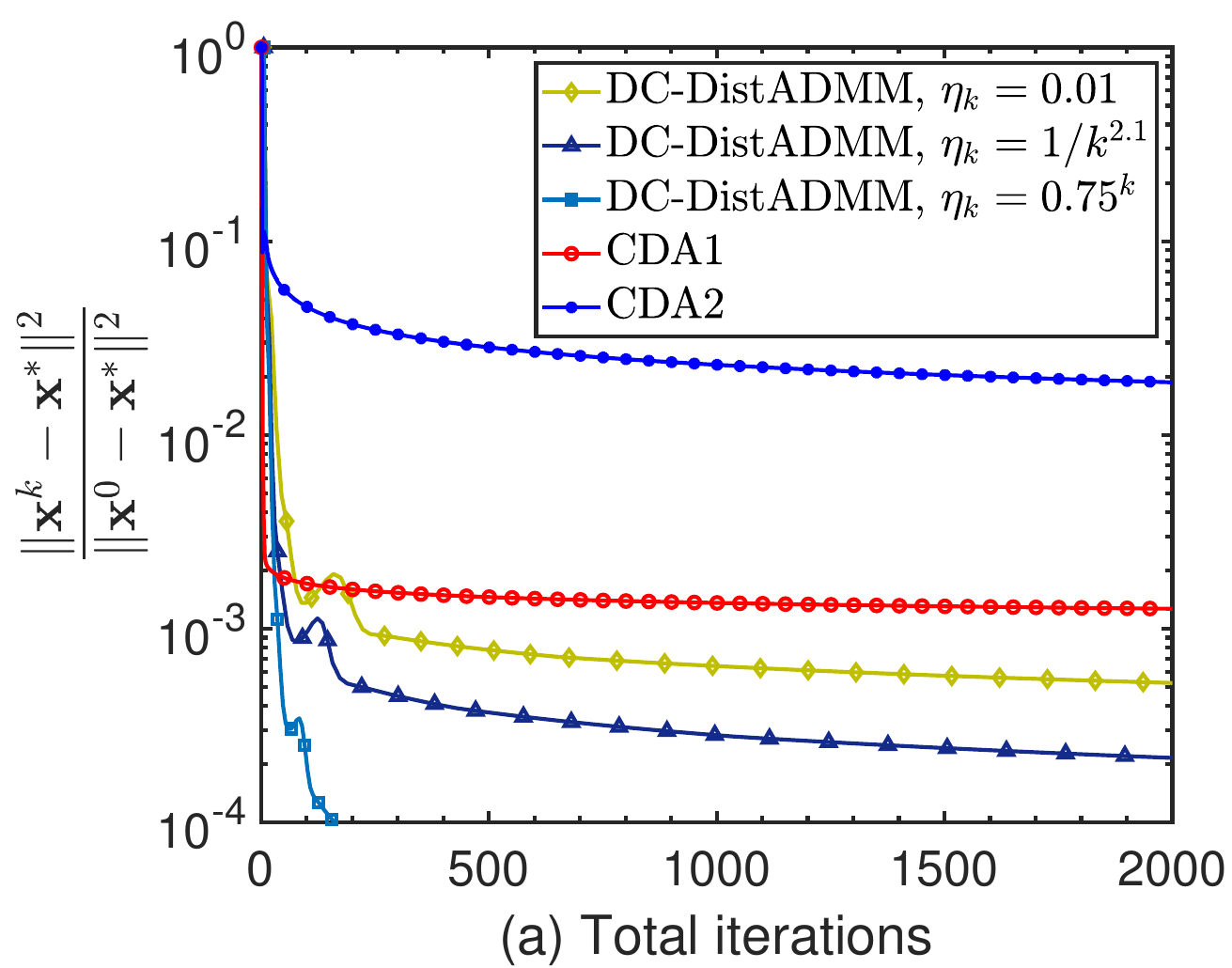}} \hspace{-0.05in}
  \subfloat{%
    \includegraphics[scale=0.205,trim={0.4cm 0cm 0.2cm 0.1cm},clip] {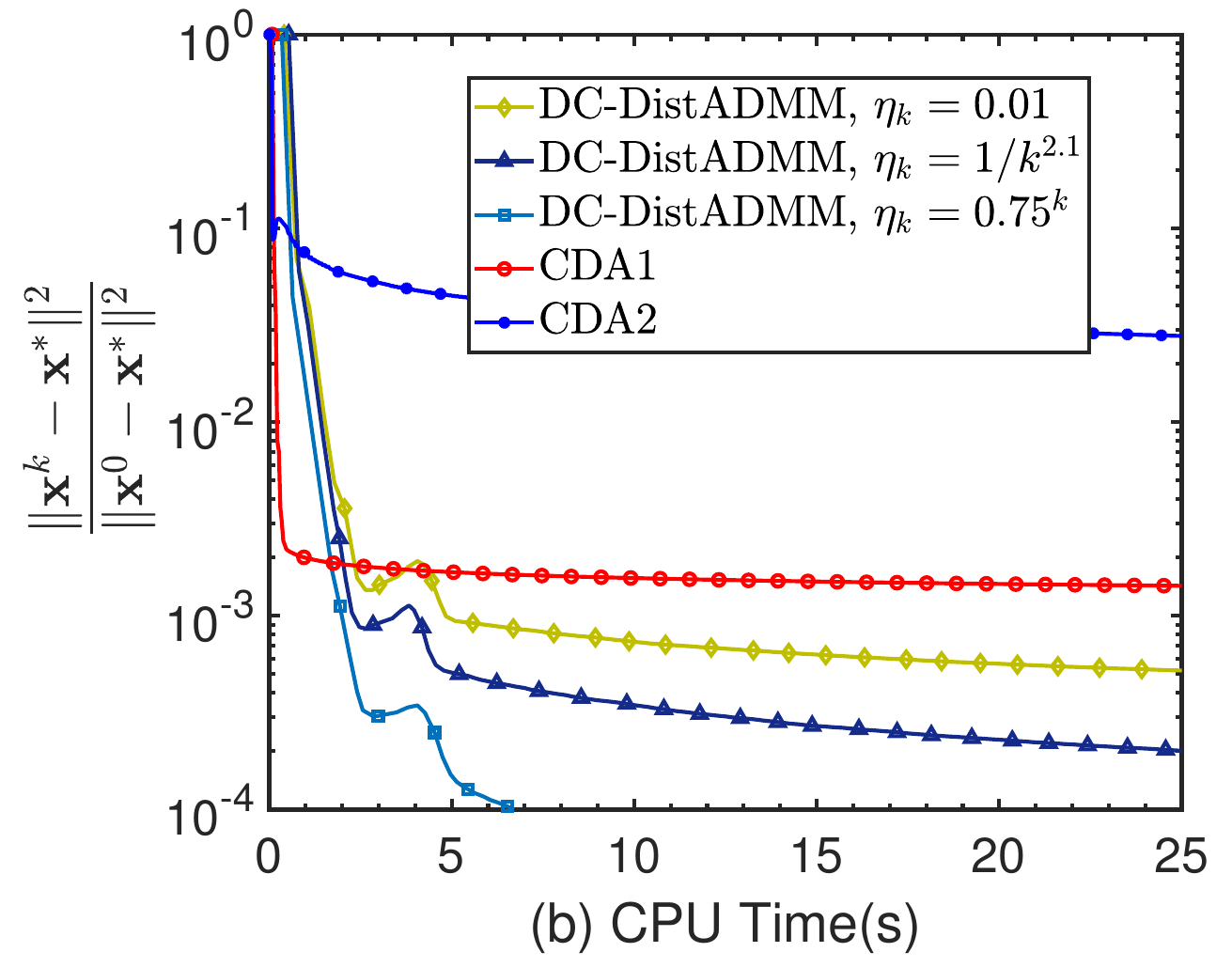}}
  \caption{(a) Comparison of solution residuals against total iterations between $\D$ and constrained optimization algorithms (b) Comparison of solution residuals against CPU time between $\D$ and constrained optimization algorithms. $\D$ algorithm with $\eta^k = 0.75^k$ has the best decrease in the solution residual.}
  \label{fig:logit_sol} 
\end{figure}
The results for three choices of the sequence $\{\eta_k\}_{k \geq 1}$: (i) $\eta_k = 0.01$, (ii) $\eta_k = 1/k^{2.1}$ and (iii) $\eta_k = 0.75^k$, where, $k$ is the iteration counter are presented for Algorithm~\ref{alg:distADMM}. Note, that sequences (ii) and (iii) satisfy condition~(\ref{eq:summable}) and are utilized to derive explicit rate of convergence for the $\D$ algorithm in theorems~\ref{thm:order1/k}, and~\ref{thm:linearrate}. \textcolor{black}{Two existing algorithms \cite{liu2017constrained} and \cite{zhu2011distributed} (termed as Constrained Distributed Algorithm 1 (CDA1) and CDA2 respectively for reference) for solving constrained distributed optimization problems are compared with the $\D$ algorithm. Fig.~\ref{fig:logit_constraints} shows the plots of consensus constraint residual ($\|\widehat{\x}^k - \widehat{\y}^k \|$ for $\D$ algorithm and $\| \x^k - \sum_{i=1}^{n} x_i^k/n\|$ for CDA1 and CDA2) and the equality constraint residual for the three algorithms with respect to the total iterations. The $\D$ performs well as seen in Fig.~\ref{fig:logit_constraints}(a). The consensus residual under $\D$ decreases to a value less than $10^{-4}$ within $200$ iterations for the choice $\eta_k = 0.75^k$. Note, that the constant $\eta_k$ sequence also have good performance with a value less than $10^{-2}$ within $200$ iterations. The plots in Fig.~\ref{fig:logit_constraints}(b) demonstrate the decrease in the equality constraint residual $\|\mathbf{H}\widehat{\x}^k - \mathbf{h} \|$. Fig.~\ref{fig:logit_sol} presents the trajectory of solution residuals $\frac{\|\x^k - \x^*\|^2}{\|\x^0 - \x^*\|^2}$ with respect to the total iterations and the CPU time (secs) for each algorithm. The proposed $\D$ algorithm outperforms the other two algorithms. The solution residual for $\D$ is lesser than CDA1 and CDA2 for all the three choices of the tolerance sequence $\eta_k$ with significant improvement with the choice $\eta_k=0.75^k$. 
Note that both CDA1 and CDA2 cannot handle directed communication topologies and need centralized synthesis for problem~(\ref{eq:logit}) whereas $\D$ does not. }

\subsection{Performance on unconstrained optimization problem}\label{sec:huber_unconstrained}
In the second part of the numerical simulation study a comparison of the proposed $\D$ algorithm is presented with some state-of-the art distributed optimization algorithms \textcolor{black}{for solving unconstrained optimization problems. Consider the following $\ell_1$ regularized Huber loss minimization problem:
\begin{align}\label{eq:huber}
    \minimize\limits \textstyle \ \ \sum_{i=1}^{n} \Phi_i(x) + \theta \|x\|_1,
\end{align}
where, $ \Phi_i(x)$ for all $i$ is the standard Huber loss given by:
\begin{align*}
     \Phi_i(x) = \begin{cases}
     \textstyle \frac{1}{2} \|D_i x - d_i\|^2, & \text{if} \ \|D_i x - d_i\| \leq 1,\\
     \|D_i x - d_i\| - \textstyle \frac{1}{2}, &  \text{if} \ \|D_i x - d_i\| > 1. 
     \end{cases}
\end{align*}
}
\hspace{-0.1in} Each agent $i \in \mathcal{V}$ has a measured data vector $d_i \in \mathbb{R}^{100}$ and a scaling matrix  $D_i \in \mathbb{R}^{100 \times 25}$. The objective is to estimate the unknown signal $x \in \mathbb{R}^{25}$. The entries of the matrices $D_i, i \in \mathcal{V}$ and the observed data $d_i$ are sampled independently from the standard normal distribution. The true solution vector $x^*$ has $70\%$ non-zero entries which are sampled from a standard normal distribution. A value of $\theta = 3$ is chosen for the regularization parameter.
\begin{table}[h]
\centering
\caption{Algorithm parameters}
\begin{tabular}{*{3}{c}} 
\hline\rule{0pt}{2.5ex} \textbf{ALGORITHM} & \textbf{Parameter} \\[0.5ex]
\hline \rule{0pt}{2ex}
\hspace{-0.07in} $\D$ & tolerance $\eta_k$ = 0.01, $1/k^{2.1}$, $0.75^k$, $\gamma = 10$ \\
CDA1 \cite{liu2017constrained}& step-size $\alpha = 0.01$ \\
CDA2 \cite{zhu2011distributed} & step-size $\alpha(k) = 1/k^{1.2}$ \\
EXTRAPush \cite{zeng2015extrapush} & step-size $\alpha = 0.0009 $  \\
PushPull \cite{pu2020push} & step-size $\alpha = 0.05 $  \\
PushDIGing \cite{nedic2017achieving} &  step-size $\alpha = 0.001$ \\
subgradientPush \cite{nedic2014distributed}  & step-size  $\alpha = 0.005 $\\
FDGD \cite{jakovetic2014fast} & step-size $\alpha =  \textcolor{black}{0.01}, \beta(k) =  k/(k+3)$ \\
nearDGD \cite{berahas2018balancing} & step-size $\alpha = \textcolor{black}{0.01} $ \\
\end{tabular}
\label{tab:parameters}
\end{table}
\textcolor{black}{Here, FISTA is employed to solve the sub-problem~(\ref{eq:dadmm_x}) with the same stopping criterion (proximal residue < $10^{-4}$) as in the logistic regression problem.} The update~(\ref{eq:dadmm_y}) is solved using $\varepsilon$-consensus protocol with \textcolor{black}{the tolerance $ \eta_k = 1/k^{2.1}$}. The progression of the solution residual with respect to the total iterations and the CPU time is presented in Fig.~\ref{fig:huber_sol}. In Fig.~\ref{fig:huber_sol}(a) plot of the solution residual with respect to the algorithm iteration $k$ (first curve) for the $\D$ algorithm is also provided. The rationale for this is to provide a comparison between $\D$ and other algorithm both with respect to algorithm iterations and total iterations. \textcolor{black}{It can be seen that the proposed $\D$ algorithm has a significantly better performance compared to other algorithms with respect to the algorithm iterations. In particular, solution residual decreases to a value of $10^{-4}$ in less than $50$ algorithm iterations of $\D$ algorithm whereas the second best method PushDIGing takes around $100$ iterations to reach the same solution residual. In terms of total iterations the $\D$ algorithm has acceptable performance.
\begin{figure}[h] 
    \centering
  \subfloat{%
  \hspace{-0.15in}
       \includegraphics[scale=0.20,trim={0.1cm 0cm 0cm 0.5cm},clip] {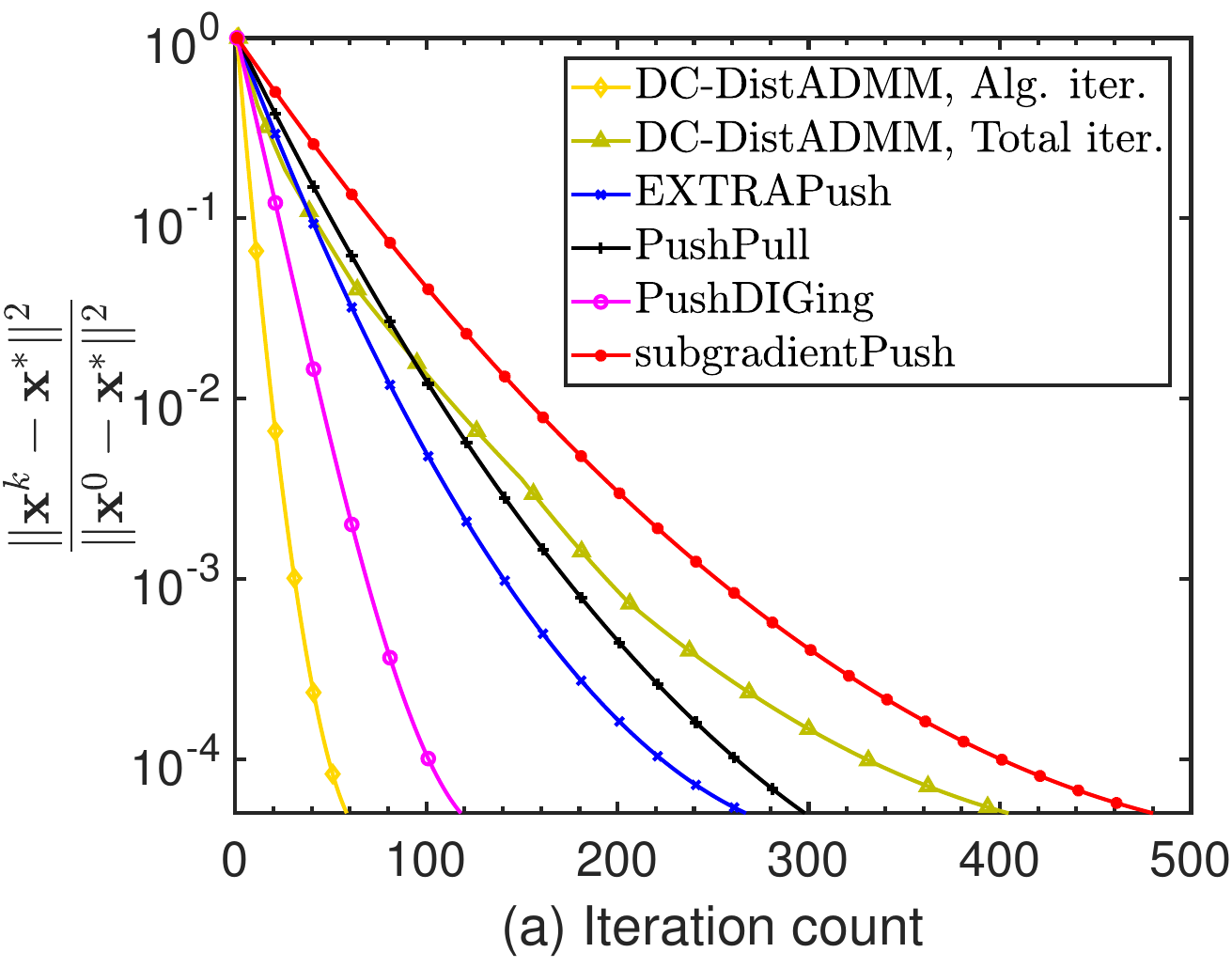}} \hspace{-0.05in}
  \subfloat{%
    \includegraphics[scale=0.201,trim={0.25cm 0cm 0.2cm 0.5cm},clip] {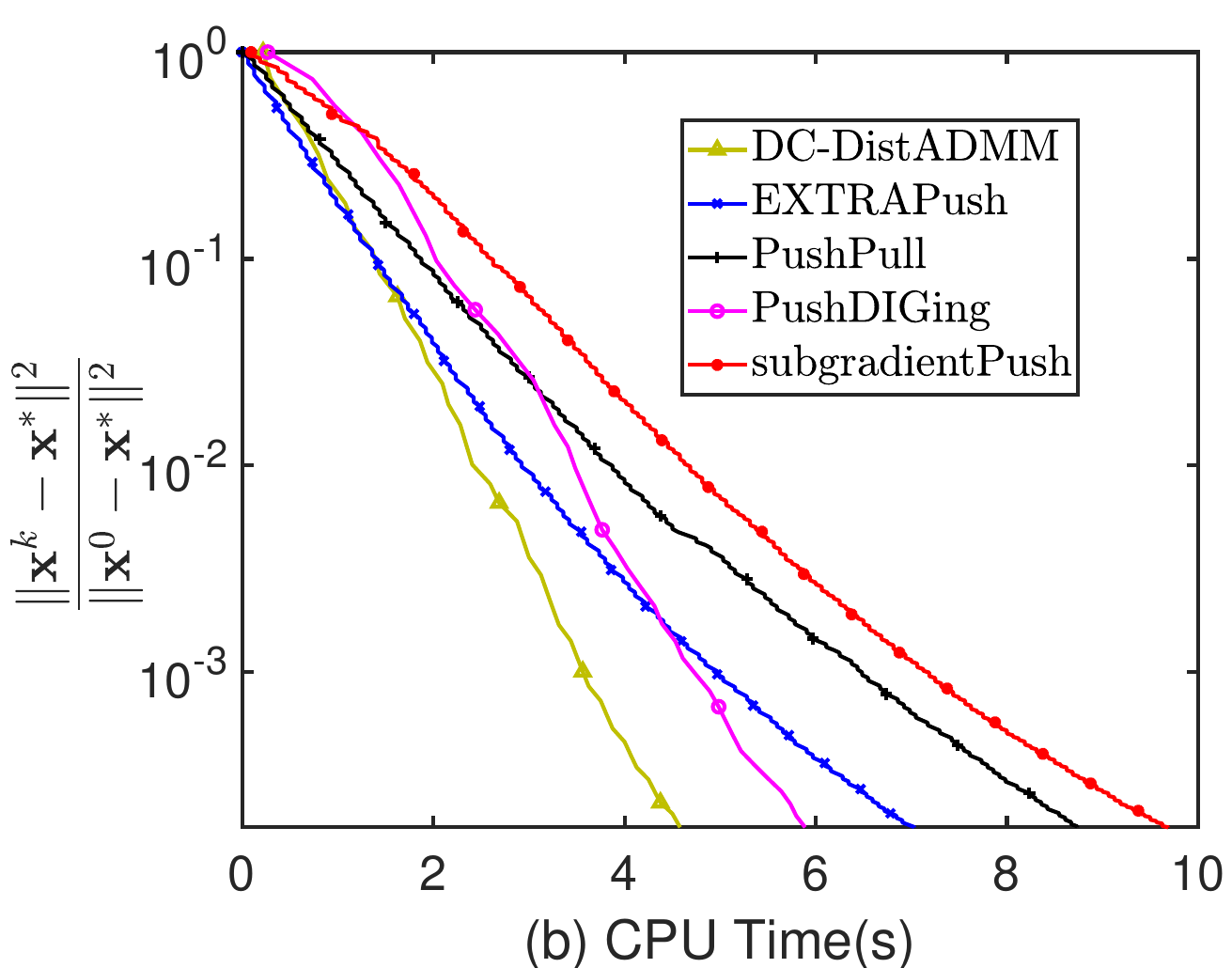}}
  \caption{(a) Comparison of solution residuals against total iterations between $\D$ and unconstrained optimization algorithms (b) Comparison of solution residuals against CPU time between $\D$ and unconstrained optimization algorithms. $\D$ algorithm performs the best in terms of the CPU time while has acceptable performance with respect to total iterations. $\D$ has the best decrease in solution residual with respect to algorithm iterations.}
  \label{fig:huber_sol} 
\end{figure}
The $\varepsilon$-consensus protocol utilized at each iteration of the $\D$ algorithm incurs additional consensus related steps, however, the information mixing steps~(\ref{eq:cons_u})-(\ref{eq:cons_v}) have a low computational footprint and do not significantly increase the overall run time of the $\D$ algorithm. This is illustrated by the comparison between the algorithms based on the required CPU time shown in Fig.~\ref{fig:huber_sol}(b). The residual plots illustrate that $\D$ algorithm performs better in terms of the CPU time requirement compared to the other methods to reach the same level of residual value.} 

\subsection{Comparison with algorithms utilizing multiple communication steps}
Here, a comparison with two unconstrained distributed optimization algorithms that utilize multiple communication steps, nearDGD \cite{berahas2018balancing} and the algorithm in \cite{jakovetic2014fast} (referred here as Fast-DGD (FDGD)) is provided; the reader is directed to \cite{berahas2018balancing, jakovetic2014fast} for the motivations for multiple consensus steps. \textcolor{black}{We consider an unconstrained distributed least squares problem,
\begin{align}\label{eq:least_sq}
    \minimize\limits \textstyle \ \ \frac{1}{2} \sum_{i=1}^{n}\|D_i x - d_i\|^2.
\end{align}
Here, each agent $i \in \mathcal{V}$ has a measured data vector $d_i \in \mathbb{R}^{100}$ and a scaling matrix  $H_i \in \mathbb{R}^{100 \times 25}$. The objective is to estimate the unknown signal $x \in \mathbb{R}^{25}$. The entries of the matrices $D_i, i \in \mathcal{V}$ and the observed data $d_i$ are sampled independently from a standard normal distribution $\mathcal{N}(0,1)$. The true signal $x^*$ also have entries that are sampled from an i.i.d standard normal distribution. In this case, the sub-problem~(\ref{eq:dadmm_x}) takes a closed form solution which can be readily computed by the first order optimality condition of the unconstrained problem. The solution residuals obtained while solving~(\ref{eq:least_sq}) using the three algorithms are presented with respect to the total iterations and CPU time in Fig.~\ref{fig:multi_sol}.
\begin{figure}[h] 
    \centering
  \subfloat{
  \hspace{-0.1in}
       \includegraphics[scale=0.20,trim={0.1cm 0cm 0.2cm 0.8cm},clip] {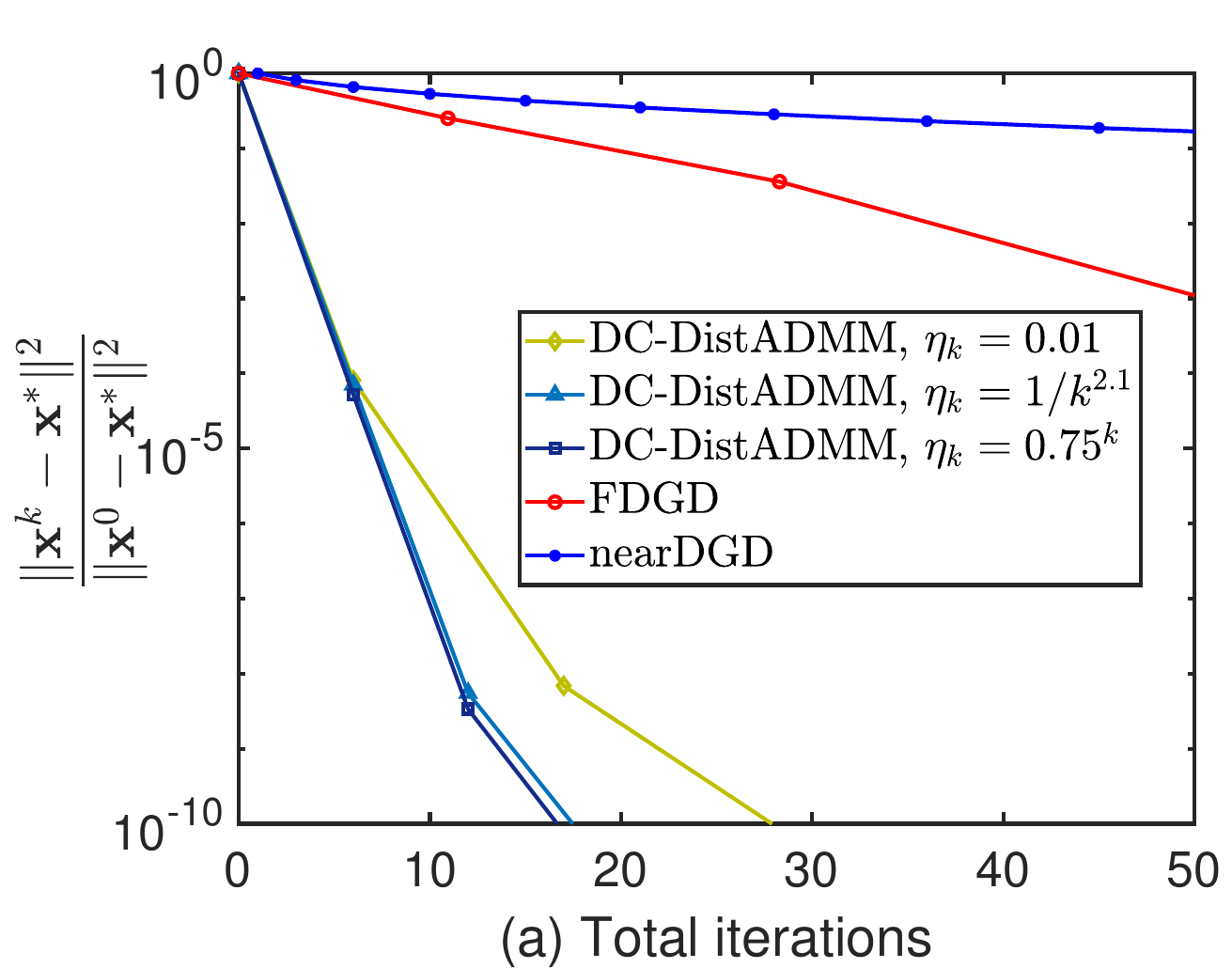}} 
  \subfloat{
    \includegraphics[scale=0.20,trim={0.3cm 0cm 0.2cm 0.8cm},clip] {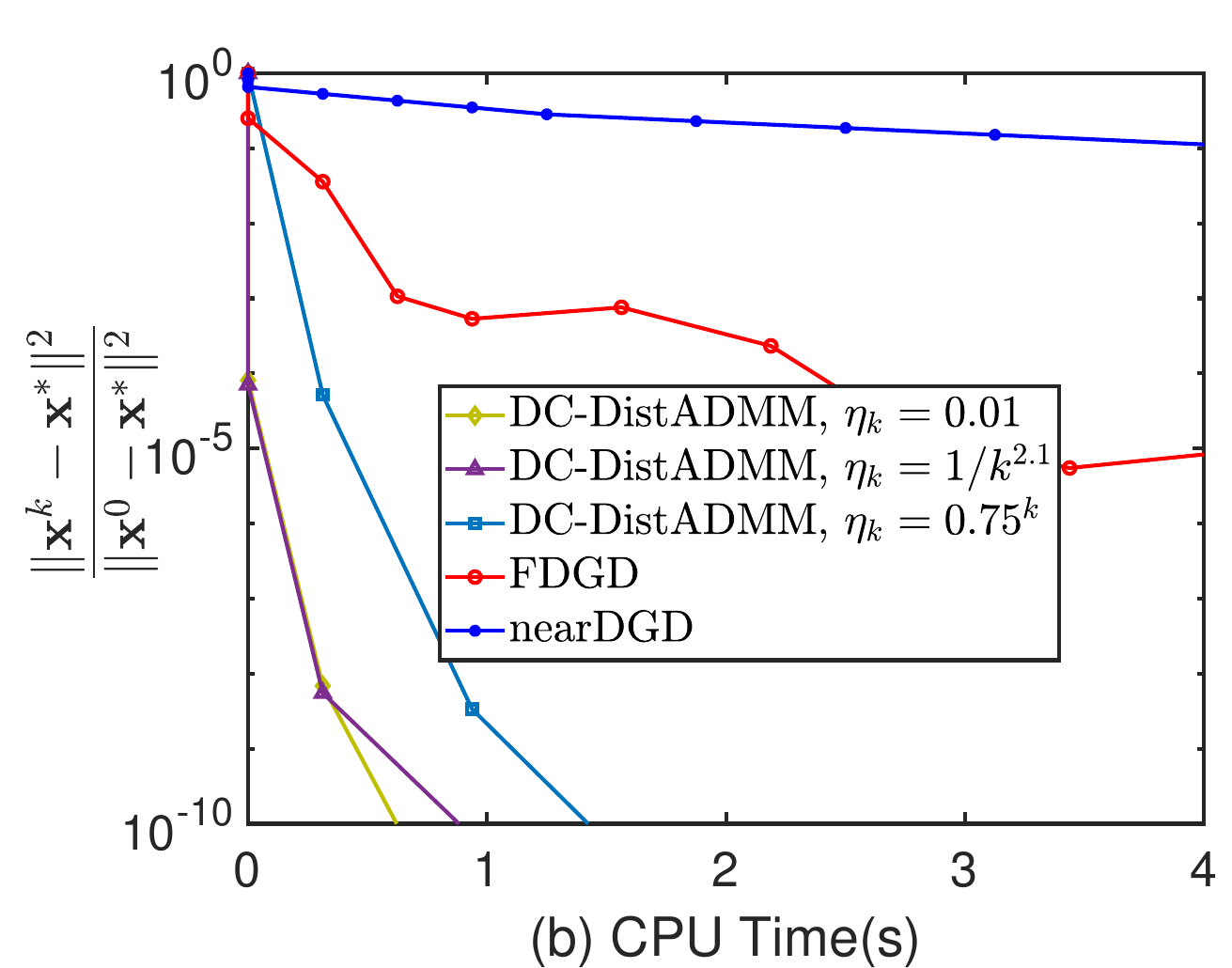}}
  \caption{\textcolor{black}{(a) Comparison of solution residuals against total iterations between $\D$ and algorithms utilizing multiple communication steps (b) Comparison of solution residuals against CPU time between $\D$ and algorithms utilizing multiple communication steps. $\eta_k = 0.75^k$ has the best decrease in solution residual with respect to total iterations. $\D$ with a constant tolerance performs the best with respect to CPU time.} }
  \label{fig:multi_sol} 
\end{figure}
It can be seen that the $\D$ algorithm has the fastest decrease in the solution residual. Fig.~\ref{fig:comm_iterations} presents the number of communication iterations utilized by the $\D$ algorithm while solving~(\ref{eq:least_sq}) for the three error sequences and the FDGD and nearDGD algorithms. The trend suggested by Lemma~\ref{lem:cons_comm} can be seen in the plots in Fig.~\ref{fig:comm_iterations}, where the number of communication iterations for the $\D$ algorithm increase as we tighten the accuracy. Note, the sequence $\eta_k = 1/k^{2.1}$ with convergence rate guarantees provided by Lemma~\ref{thm:order1/k} results in logarithmic increase in the  number of communication iterations and is a suitable choice for getting good performance. The $\D$ algorithm with tolerance $\eta_k = 0.01$ and $\eta_k = 1/k^{2.1}$ has lesser communication cost compared to the FDGD and nearDGD algorithms while the FDGD has better communication complexity than the $\D$ algorithm with the tolerance $\eta_k = 0.75^k$ as the number of iterations increase.}
\begin{figure}[h]
    \centering
    \includegraphics[scale=0.25,trim={0cm 0cm 0cm 0.8cm},clip] {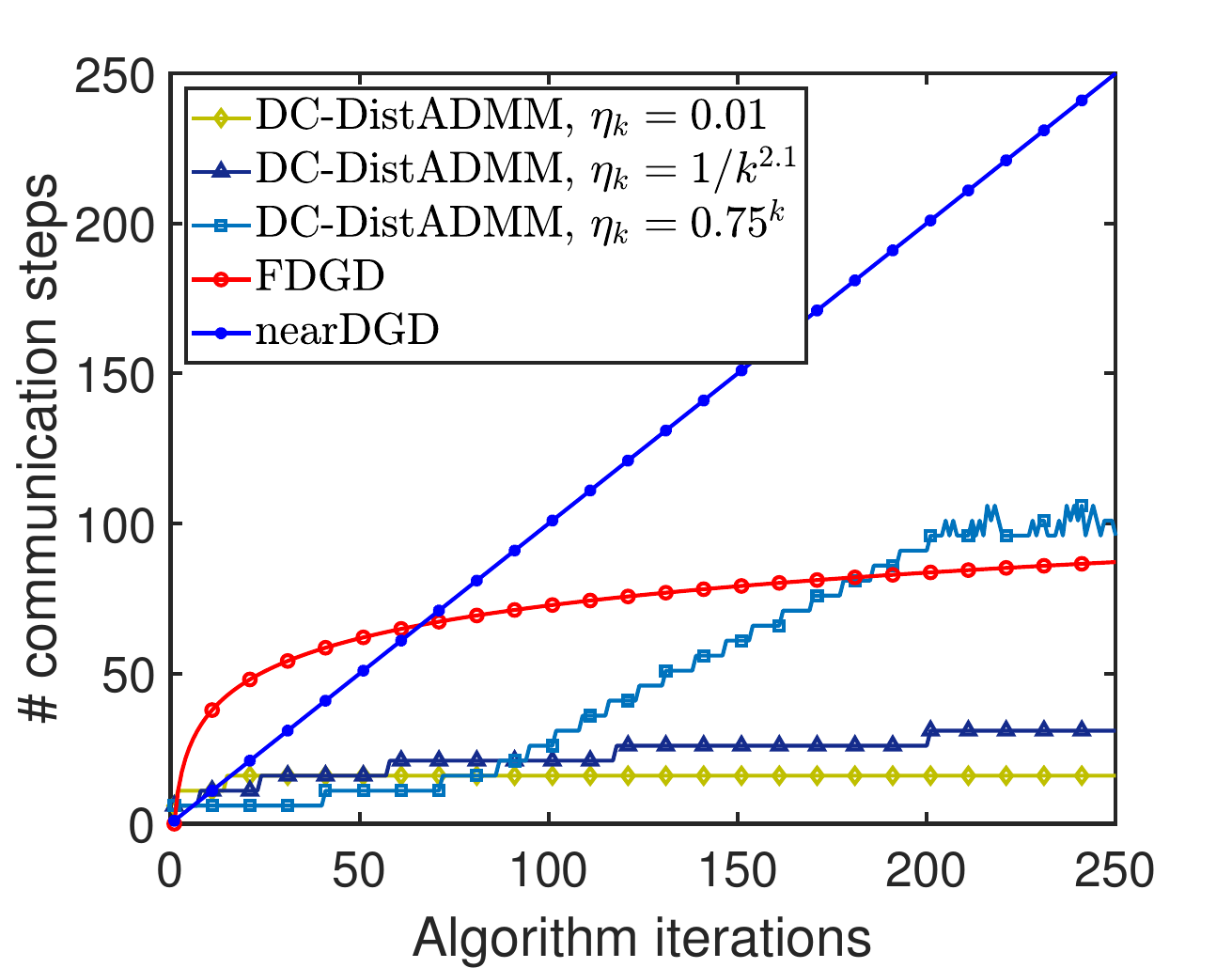}
    \caption{\textcolor{black}{Comparison of the number of communication iterations performed by $\D$ algorithm for the three choices of $\eta_k$ and the FDGD and nearDGD algorithms. $\D$ with $\eta_k = 0.01$ and $\eta_k = 1/k^{2.1}$ has lesser communication cost compared to the FDGD and nearDGD algorithms while the FDGD has better communication complexity than $\D$ with $\eta_k = 0.75^k$ as the number of iterations increase.} }
    \label{fig:comm_iterations}
\end{figure}
Summarizing, the simulation results indicate that with respect to the state-of-the-art \textit{constrained} optimization problems $\D$ provides better performance while encompassing a large class of distributed optimization scenarios with linear equality, inequality and set constraints. Even with respect to the state-of-the-art \textit{unconstrained} optimization frameworks, $\D$ algorithm provides better performance with respect to computations (CPU time) and remains acceptable with respect to total communication steps.
\end{section}

\begin{section}{Conclusion and Future Work}\label{sec:conclusion}
In this article, a novel Directed-Distributed Alternating Direction Method of Multipliers ($\D$) algorithm is presented to solve constrained multi-agent optimization problems with local linear equality, inequality and set constraints over general directed graphs. Moreover, the algorithm is suited for distributed synthesis. The proposed algorithm, to the best of authors' knowledge, is the \textcolor{red}{first} ADMM based algorithm to solve (un)constrained distributed optimization problems over directed graphs. The $\D$ algorithm combines techniques used in \textcolor{black}{Lagrangian} dual based optimization methods along with the ideas in the average consensus literature. In the $\D$ algorithm, each agent solves a local constrained convex optimization problem and utilizes a finite-time $\varepsilon$-consensus algorithm to update its estimate of the optimal solution. The proposed $\D$ algorithm enjoys provable rate of convergence guarantees: (i) a $O(1/k)$ rate of convergence when the individual functions are convex but not-necessarily differentiable, (ii) a \textcolor{black}{geometric decrease to arbitrary small neighborhood of the optimal solution when the objective functions are smooth and restricted strongly convex at the optimal solution}. The proposed $\D$ algorithm eliminates the optimization step involving a primal variable in the standard ADMM setup and replaces it with a less computation intensive $\varepsilon$-consensus protocol which makes it more suitable for distributed multi-agent systems. To show the efficacy of the $\D$ algorithm numerical simulation results \textcolor{black}{comparing the performance of the $\D$ algorithm with the existing state-of-the-art algorithms in the literature of solving constrained and unconstrained distributed optimization problems are presented. A comparison of the $\D$ algorithm and two existing algorithms in the literature utilizing multiple consensus steps is also provided.} Extension of the $\D$ algorithmic framework to networks with time-delays in communication  \cite{prakash2019distributed} between the agents and time-varying connectivity among the agents \cite{saraswat2019distributed} is a future work of this article.
\end{section}

\appendix
\vspace{-0.4in}
\textcolor{black}{\subsection{Distributed Synthesis}\label{sec:distri_synth}
\begin{figure}[b]
    \centering
    \includegraphics[scale=0.25,trim={6cm 4.5cm 6cm 4.2cm},clip] {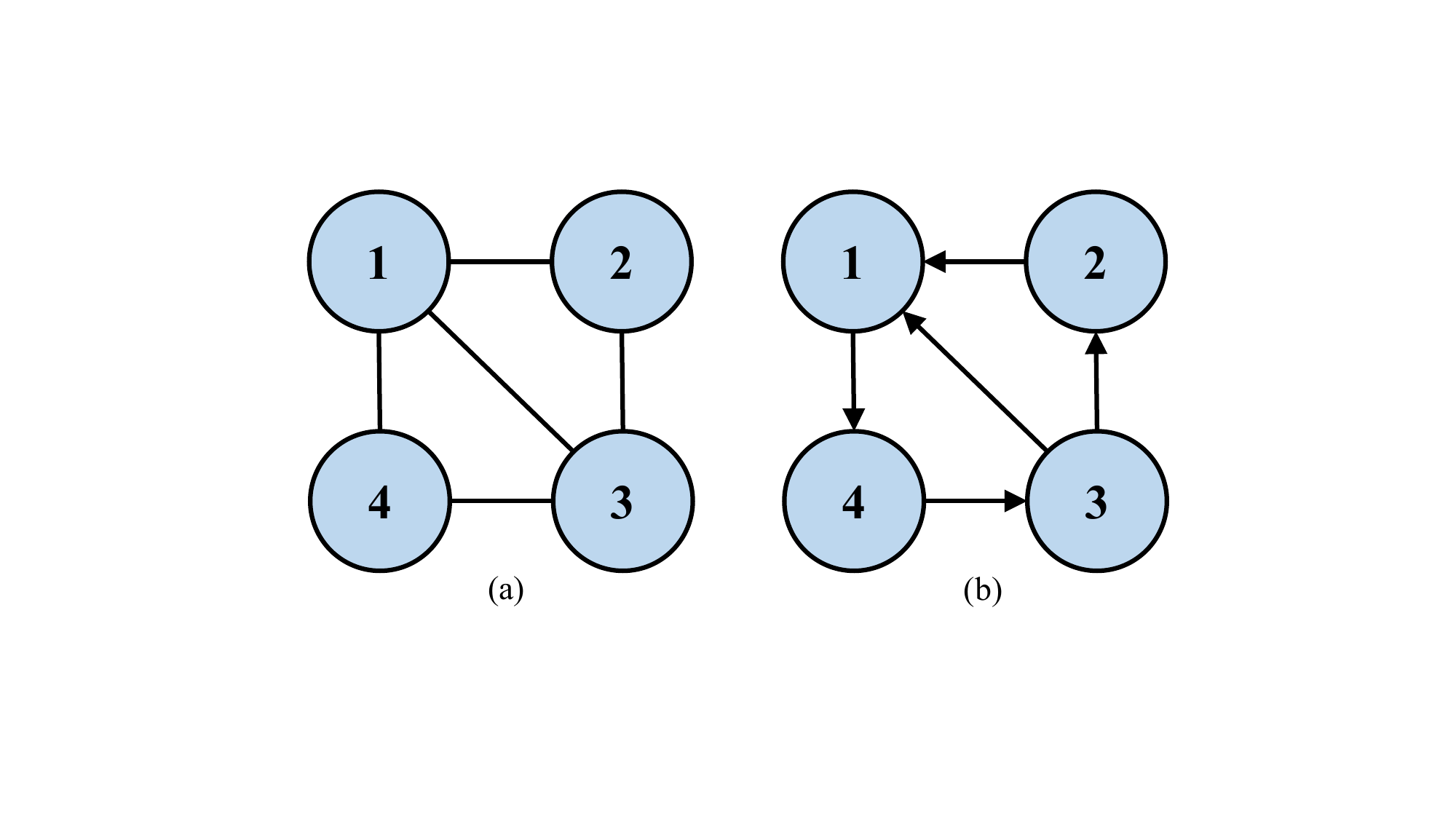}
    \caption{Example directed network}
    \label{fig:graph}
\end{figure}
\noindent To explain the concept of distributed synthesis a simple example in the context of distributed consensus/optimization is presented here. Consider the undirected and directed graphs with four nodes (agents) as shown in Fig.~\ref{fig:graph}(a) and (b). Assume that each agent $i = 1,\dots,4$ updates its estimate $\psi_i$ by taking a weighted combination of its own prior estimate and the information $\psi_j$ received from its neighbors $j$
\begin{align}\label{eq:dist_synthesis}
    \mbox{ i.e.,} \ \ \ \psi_i^+ = p_{ii} \psi_i + \textstyle \sum_{j: \text{neighbors}} p_{ij} \psi_j.
\end{align}
The update protocol~(\ref{eq:dist_synthesis}) is called \textit{ amenable to distributed synthesis} if the weights $p_{ij}$ can be decided by the agents locally and independently without any coordination among them. Let $\mathcal{P} = [p_{ij}], i,j \in \{1,\dots, 4\}$ be the weight matrix with entries $p_{ij}, i,j \in \{1,\dots, 4\} \neq 0$ if and only if there is an  undirected (or a directed) edge between $i$ and $j$ in Fig.~\ref{fig:graph}(a) (or (b)). Hence, due to the structure of the weight matrix $\mathcal{P}$ whether an update scheme of the form~(\ref{eq:dist_synthesis}) is \textit{amenable to distributed synthesis} or not can be inferred by the constraints imposed on the entries of $\mathcal{P}$. In particular, if $\sum_{i=1}^4 p_{ij} = 1 = \sum_{j=1}^4 p_{ij}$, i.e. the matrix is double stochastic then one agent cannot decide the weight it assigns to the information received from its neighbor independently (by only using its \textit{local} neighborhood information). The constraint of sum of the entries of all rows and columns being equal to one necessitates the need for coordination among all the agents (knowledge of \textit{global} network information). It is established in the literature that generating a double stochastic $\mathcal{P}$ matrix is computationally intractable on general directed graphs without any \textit{global} knowledge and coordination \cite{gharesifard2012distributed}. This renders the algorithms using a double stochastic update matrix on directed graphs unsuitable for distributed synthesis. However, a matrix $\mathcal{P}$ with only column (or row) stochasticity constraint can be generated easily with \textit{only} the local information available at each agent. The out-degree based equal neighbor weights rule  \cite{olshevsky2009convergence} is a common choice of weights that results in a distributed synthesis based column stochastic $\mathcal{P}$ matrix. Note that distributed synthesis is a property of an algorithm (update~(\ref{eq:dist_synthesis}) in our example), and has no direct relation to the underlying communication topology of the agent graph.} 

\subsection{Claim: $[\overline{y}^{k+1^\top} \overline{y}^{k+1^\top} \dots \overline{y}^{k+1^\top}]^\top \in \mathbb{R}^{np},  \mbox{where}, \overline{y}^{k+1} = \textstyle \frac{1}{n} \sum_{i=1}^n [x_i^{k+1} + \frac{1}{\gamma} \lambda_i^{k}]$ is the solution to~(\ref{eq:dadmm_y}) }\label{sec:y_exsol}
\begin{proof}
Consider, the update~(\ref{eq:dadmm_y})
\begin{align*}
\overline{\y}^{k+1}  = \textstyle & \argmin\limits_{\y} \textstyle \Big\{\mathcal{I}_{\mathcal{C}_0}(\y) + \frac{\gamma}{2}\|\x^{k+1} - \y + \frac{1}{\gamma} \lambda^{k}\|^2\Big\}.
\end{align*}
Using definition of the set $\mathcal{C}_0$ update~(\ref{eq:dadmm_y}) can be written as:
\begin{align*}
\overline{\y}^{k+1}  = \textstyle \argmin\limits_{\y} \ &\  \textstyle \frac{\gamma}{2}\|\x^{k+1} - \y + \frac{1}{\gamma} \lambda^{k}\|^2 \\
\mbox{subject to} & \ y_i = y_j, \forall i,j \in \V.
\end{align*}
Let $y_i = y_j = u, \forall i,j \in \V$. As the objective is in a block separable form we have the following equivalent problem,
\begin{align*}
    \textstyle \minimize\limits_{u} \frac{\gamma}{2} \sum_{i=1}^n \| x_i^{k+1} + \frac{1}{\gamma} \lambda_i^k - u\|^2.  
\end{align*}
As the above problem is an unconstrained convex optimization problem we can differentiate the objective and set it equal to zero to get the optimal solution. In particular, $\sum_{i=1}^n [x_i^{k+1} + \frac{1}{\gamma} \lambda_i^k - u] = 0$ which implies $u = \frac{1}{n} \sum_{i=1}^n [x_i^{k+1} + \frac{1}{\gamma} \lambda_i^k]$. Therefore, $y_i = y_j = u, \forall i,j \in \V$. Hence, $\overline{\y}^{k+1} = [u^\top \dots u^\top]^\top$ is the solution to~(\ref{eq:dadmm_y}). 
\end{proof}

\begin{subsection}{Proof of Lemma~\ref{lem:state_inside_ball}}\label{sec:proof_state_inside_ball}
\noindent Define, $d_{i,j}$ to be length of the shortest path connecting node $i$ to node $j$. To prove the claim we use an induction argument on $d_{i,j}$. 
Base case: $d_{i,j} = 1$,
Given, $s \geq 0$ for all $i \in \mathcal{V}$,
\begin{align*}
    R_i^{1}(s) = \max\limits_{j \in \mathcal{N}_i^-} \Big \{ \| w_i^{s + 1} - w_j^{s} \| + R_j^{0}(s) \Big \}.
\end{align*}
Therefore, for all $j \in \mathcal{N}_i^-$, $w_j^{s} \in \mathcal{B}( w_i^{s + 1}, R_i^{1}(s))$. Assume the claim holds for $d_{i,j} = k$, that is for all $j$ such that $d_{i,j} \leq k$, $w_j^{s} \in \mathcal{B}( w_i^{s + k}, R_i^{k}(s))$. For any node $\ell \in \mathcal{N}_i^-$, on the path between node $i$ and $j$ with $d_{i,j} \leq k+1$, we have $d_{\ell,j} \leq k$. By the induction hypothesis, $w_j^{s} \in \mathcal{B}( w_\ell^{s + k}, R_\ell^{k}(s))$. It implies,
$\| w_\ell^{{s} + k} - w_j^{s} \| \leq R_\ell^{k}(s)$. From~(\ref{eq:radius}) we have, 
\begin{align*}
    R_i^{k + 1}(s) \geq \| w_i^{s + k + 1} - w_\ell^{s + k } \| + R_\ell^{k}(s).
\end{align*}
Using, the triangle inequality,
\begin{align*}
    \| w_i^{s + k + 1} - w_j^{s} \| & \leq \| w_i^{s + k + 1} - w_\ell^{s + k } \| + \| w_\ell^{s + k} - w_j^{s} \| \\
    & \leq \| w_i^{s + k + 1} - w_\ell^{s + k } \| + R_\ell^{k}(s)\\
    & \leq R_i^{k + 1}(s).
\end{align*}
Therefore, for any $j$ on the path with $d_{i,j} \leq k+1$, $w_j^{s} \in \mathcal{B}( w_i^{s + k+ 1}, R_i^{k+ 1}(s))$. Therefore, induction holds. When $k = \mathcal{D}$, we get the desired result. 
\end{subsection}

\begin{subsection}{Proof of Theorem~\ref{thm:eCons}}\label{sec:proof_eCons}
\noindent Define the element-wise maximum $\overline{\mathbf{w}}^k$ and minimum $\underline{\mathbf{w}}^k$ state variable of the network over all the agents as
\begin{align*}
 \overline{\mathbf{w}}^k & := \Big[ \underset{ 1 \leq j \leq n }{\max} \big\{w_{1[j]}^k \big\} \dots \ \underset{ 1 \leq j \leq n }{\max} \big\{w_{n[j]}^k \big\} \Big]^\top, 
 \\
 \underline{\mathbf{w}}^k & := \Big[ \underset{ 1 \leq j \leq n }{\min} \big\{w_{1[j]}^k \big\} \dots \ \underset{ 1 \leq j \leq n }{\min} \big\{w_{n[j]}^k \big\} \Big]^\top, 
\end{align*}
where, $w_{i[j]}^k$ is the $j^{th}$ entry of $w_i^k$. It is proven in \cite{yadav2007distributed}, \cite{melbourne2020geometry} that $\lim_{m \rightarrow \infty} \overline{\mathbf{w}}^{m \mathcal{D}} = \lim_{m \rightarrow \infty} \underline{\mathbf{w}}^{m \mathcal{D}} = \widetilde{u} = \frac{1}{n}\sum_{i=1}^n u_i^0$. It implies $ \lim_{m \rightarrow \infty}  \|\overline{\mathbf{w}}^{m \mathcal{D}} - \underline{\mathbf{w}}^{m \mathcal{D}} \| = 0$.
Proceeding in the same manner as in \cite{melbourne2020geometry}, Theorem 4.2, we get $\widehat{R}_i^m \leq \mathcal{D} \|\overline{\mathbf{w}}^{m\mathcal{D}} - \underline{\mathbf{w}}^{m \mathcal{D}} \|$. 
Taking $\lim_{m \rightarrow \infty} $ both sides we conclude for all $i \in \mathcal{V}$, $\lim_{m \rightarrow \infty} \widehat{R}_i^{m} = 0$. This gives the first result. 
\noindent From Lemma~\ref{lem:state_inside_ball}, $\max_{i,j \in \mathcal{V}} \| w^{m \mathcal{D}}_i - w_j^{m \mathcal{D}}\| \leq 2\widehat{R}_i^m$. Letting $ m \rightarrow \infty$, gives the second claim.
\end{subsection}

\begin{subsection}{Claim: $\y^{k+1}$ is an inexact solution to~(\ref{eq:dadmm_y})}\label{sec:y_sol}
\begin{proof}
As shown in Appendix~\ref{sec:y_exsol}, the solution to~(\ref{eq:dadmm_y}) is $\overline{\y}^{k+1} = [u^\top \dots u^\top]^\top$, where, $u := \frac{1}{n} \sum_{i=1}^n [x_i^{k+1} + \frac{1}{\gamma} \lambda_i^k]$. 
Recall from~(\ref{eq:cons_sol}) that the $\varepsilon$-consensus algorithm (updates~(\ref{eq:cons_u})-(\ref{eq:cons_w})) at iteration $k$ of the $\D$ algorithm, with the initial condition $u_i^0 = x_i^{k+1} + \frac{1}{\gamma} \lambda_i^{k}, \ v_i^0 = 1, \forall i \in \mathcal{V}$ and $\varepsilon = \eta_{k+1}$ yields a solution of the form $[u^\top u^\top \dots u^\top]^\top + e^{k+1}$, where $\|e^{k+1}\| \leq \sqrt{n} \eta_{k+1}$. Therefore, the consensus step yields an inexact solution of~(\ref{eq:dadmm_y}) of the form $ \y^{k+1}:= [(u+e^{k+1}_1)^\top (u + e^{k+1}_2)^\top \dots (u + e^{k+1}_n)^\top]^\top$, where $\|e^{k+1}\| \leq \sqrt{n} \eta_{k+1}$.
\end{proof}
\end{subsection}
\begin{subsection}{Proof of Lemma~\ref{lem:ylambdabound}}\label{sec:proof_ylambdabound}
\noindent Note that from the initialization in Algorithm~\ref{alg:distADMM} $\lambda_i = 0, \forall i \in \V$. From~(\ref{eq:y_over}), (\ref{eq:dadmm_lam}), (\ref{eq:dadmm_lam_ne}), and,~(\ref{eq:error}), for all $i \in \V$,
\begin{align*}
    \lambda_i^{k+1} &= \lambda_i^k + \gamma (x_i^{k+1} - \overline{y}^{k+1})  - \gamma e^{k+1}_i \\
    & = \textstyle \gamma \sum_{s=0}^{k} (x_i^{s+1} - \overline{y}^{s+1}) - \gamma \sum_{s=0}^{k} e_i^{s+1} \\
    & \hspace{-0.2in}= \textstyle \gamma \sum_{s=0}^{k} \big[ x_i^{s+1} - \frac{1}{n} \sum_{j=1}^n (x_j^{s+1} + \frac{1}{\gamma} \lambda_j^s) \big] - \gamma \sum_{s=0}^{k} e_i^{s+1}\\
    & \hspace{-0.25in} = - \textstyle \frac{1}{n}  
    \sum_{s=0}^{k} \sum_{j=1}^n \lambda_j^s + \gamma \sum_{s=0}^{k} \big[ x_i^{s+1} - \frac{1}{n} \sum_{j=1}^n x_j^{s+1}\big] \\
    & \hspace{-0.15in} \textstyle - \gamma \sum_{s=0}^{k} e_i^{s+1}.
\end{align*}
Let $l_i^{k+1}:= \gamma \sum_{s=0}^{k} \big[ x_i^{s+1} - \frac{1}{n} \sum_{j=1}^n x_j^{s+1}\big] - \gamma \sum_{s=0}^{k} e_i^{s+1}$. Therefore,
\begin{align*}
    \lambda_i^{k+1} = - \textstyle \frac{1}{n}  
    \sum_{s=0}^{k} \sum_{j=1}^n \lambda_j^s + l_i^{k+1}
\end{align*}
Similarly,
\begin{align*}
    \lambda_i^{k} = - \textstyle \frac{1}{n}  
    \sum_{s=0}^{k-1} \sum_{j=1}^n \lambda_j^s + l_i^{k}
\end{align*}
Taking the difference of the above two equations we get,
\begin{align*}
    \lambda_i^{k+1} = \lambda_i^k - \textstyle \frac{1}{n}  
    \sum_{j=1}^n \lambda_j^k + \gamma \big[x_i^{k+1} - \frac{1}{n} \sum_{j=1}^n x_j^{k+1}\big] - \gamma e_i^{k+1}.
\end{align*}
Let $r_i^{k+1} = \gamma \big[x_i^{k+1} - \frac{1}{n} \sum_{j=1}^n x_j^{k+1}\big] - \gamma e_i^{k+1}$. Taking $z$-transform both sides of the above equation we get,
\begin{align*}
    z\hat{\lambda}_i = \hat{\lambda}_i - \textstyle \frac{1}{n}  
    \sum_{j=1}^n \hat{\lambda}_j + \hat{r}_i,
\end{align*}
where, $\hat{(.)}$ denotes the $z$-transform of the corresponding time-domain signal. Let $\hat{\lambda} := [\hat{\lambda}_1^\top \  \hat{\lambda}_2^\top \dots \hat{\lambda}_n^\top]^\top$ and $\hat{r} := [\hat{r}_1 ^\top \  \hat{r}_2^\top \dots \hat{r}_n^\top]^\top$. Writing, the above $n$ equations compactly we get,
\begin{align*}
    \textstyle \underbrace{\begin{bmatrix}
    (z - 1 + \frac{1}{n})\mathbf{I}_{p} & \frac{1}{n}\mathbf{I}_{p} & \dots & \frac{1}{n}\mathbf{I}_{p} \\
     \frac{1}{n}\mathbf{I}_{p}& (z - 1 + \frac{1}{n})\mathbf{I}_{p} & \dots & \frac{1}{n}\mathbf{I}_{p} \\
     \vdots & \vdots  & \ddots & \vdots\\
    \frac{1}{n}\mathbf{I}_{p} & \frac{1}{n}\mathbf{I}_{p} & \dots & (z - 1 + \frac{1}{n})\mathbf{I}_{p}
    \end{bmatrix}}_{T(z)} \hat{\lambda} = \hat{r},
\end{align*}
where $\mathbf{I}_p$ is the $p$-dimensional identity matrix. Note, that the matrix $T(z)$ has eigenvalues $z$ with algebraic multiplicity $p$ and $(z-1)$ with algebraic multiplicity $(n-1)p$. Therefore, there exists a (constant) matrix $Q$ such that 
\begin{align*}
    \hat{\lambda} = T^{-1} \hat{r} = Q^{-1} \begin{bmatrix}
    \frac{1}{z}\mathbf{I}_{p} & \mathbf{0}_{p \times p} & \dots & \mathbf{0}_{p \times p} \\
     \mathbf{0}_{p \times p} & \frac{1}{z-1}\mathbf{I}_{p} & \dots & \mathbf{0}_{p \times p} \\
     \vdots & \vdots  & \ddots & \vdots\\
     \mathbf{0}_{p \times p} & \frac{1}{n}\mathbf{I}_{p} & \dots & \frac{1}{z-1}\mathbf{I}_{p}
    \end{bmatrix} Q \hat{r}.
\end{align*}
Let $ Q \hat{\lambda} =: \hat{q} = [\hat{q}_1^\top \hat{q}_2^\top \dots \hat{q}_n^\top]^\top$, and $ Q\hat{r} =: \hat{t} = [\hat{r}_1^\top \hat{r}_2^\top \dots \hat{r}_n^\top]^\top$. Therefore,
\begin{align*}
    \hat{q}_1 = \textstyle \frac{1}{z} \hat{t}_1, \ \hat{q}_i = \textstyle \frac{1}{z-1} \hat{t}_i, \ \mbox{for} \ i = 2,3,\dots,n.
\end{align*}
Taking inverse $z$-transform on both sides of the above equations we get,
\begin{align*}
    q_1^{k+1} = t_1^{k}, \ q_i^{k+1} = q_i^{k} + t_i^{k}, \ \mbox{for} \ i = 2,3,\dots,n.
\end{align*}
Therefore,
\begin{align*}
    q_1^{k+1} = t_1^{k}, \ q_i^{k+1} = q_i^0 + \textstyle \sum_{s=0}^{k}t_i^{s}, \ \mbox{for} \ i = 2,3,\dots,n.
\end{align*}
Let $Q_i \in \mathbb{R}^{np}$ denote the $i^{th}$ row of $Q$. Note, that $\|t_i^{s}\| = \|Q_i r^{s} \|= \|Q_i \gamma [[ x_1^{s+1} - \frac{1}{n} \sum_{j=1}^n x_j^{s+1} - e_1^{s+1} ]^\top \dots [ x_1^{s+1} - \frac{1}{n} \sum_{j=1}^n x_j^{s+1} - e_1^{s+1} ]^\top ]^\top  ) \| \leq \gamma \sqrt{n} \|Q_i\| (\mathcal{M} + \eta_{s+1}) \leq \gamma \sqrt{n} \|Q_i\| (\mathcal{M} + \eta_{0}), \forall i \in \V$. Let, $\|Q\|_{\max} := \max_{1 \leq i \leq n} \|Q_i\|$.\\
Therefore, $\| \lambda^k\| = \|Q^{-1} q^k\| \leq k n \gamma \|Q^{-1}\|  \|Q\|_{\max} (\mathcal{M} + \eta_{0})$. Note, that here 
\begin{align}\label{eq:q}
    \mathbf{Q}:= n \gamma \|Q^{-1}\| \|Q\|_{\max} (\mathcal{M} + \eta_{0}) < \infty,
\end{align}
associated with, $\|\lambda^k\| \leq k\mathbf{Q}$.
\end{subsection}
\subsection{Proof of Lemma~\ref{lem:cons_comm}} \label{sec:proof_cons_comm}
\noindent To begin, we will present a modification of an existing result (Lemma 1 \cite{nedic2014distributed}). The result in Lemma 1 \cite{nedic2014distributed}, with the perturbation term being zero, reduces to a convergence result for the push-sum protocol. Using this property, we conclude that the updates~(\ref{eq:cons_u})-(\ref{eq:cons_w}) converges at a geometric rate to the average of the initial values. Note that the $\D$ algorithm at every iteration $k$ utilizes $\varepsilon$-consensus protocol (updates~(\ref{eq:cons_u})-(\ref{eq:cons_w})) with the initial values $x_i^{k+1} + \frac{1}{\gamma} \lambda_i^k$. Therefore, we conclude,
\begin{align}\label{eq:nedic_result}
    \hspace{-0.098in} \|y_i^{k+1} - \overline{y}_i^{k+1}\| \leq \textstyle \frac{8 \sqrt{n} \alpha^{t_k}}{\beta} \|\mathbf{z}^k\|, \ \forall i \in \V,
\end{align} 
where, $\mathbf{z}^k := \x^{k+1} + \frac{1}{\gamma} \lambda^k$ and, $\beta >0, \alpha \in (0,1)$ satisfy: $
\beta \geq \frac{1}{n^n}, \alpha \leq \textstyle \left(1 - \frac{1}{n^n} \right)$.
Here, the variables $\alpha$ and $\beta$ are parameters of the graph $\mathcal{G}(\mathcal{V},\mathcal{E})$. The parameter $\alpha$ measures the speed at which the graph $\mathcal{G}(\mathcal{V},\mathcal{E})$ diffuses the information among the agents over time.
Further, the parameter, $\beta$ measures the imbalance of influences in $\mathcal{G}(\mathcal{V},\mathcal{E})$ \cite{nedic2014distributed}. 
Under Assumption~\ref{ass:constraint_ass} using the result in Lemma~\ref{lem:ylambdabound}, $\|\mathbf{z}^k \| \leq \sqrt{n} (\mathcal{M} + \|x^0\|) +  k \frac{\mathbf{Q}}{\gamma}$. If $\frac{\beta \eta_{k+1}}{8 \sqrt{n} \alpha^{t_k} } = \|\mathbf{z}^k\| \leq \sqrt{n} (\mathcal{M} + \|x^0\|) +  k \frac{\mathbf{Q}}{\gamma}$  it implies that, $\frac{\eta_{k+1} \beta}{8((\mathcal{M} + \|x^0\|) n + k \frac{\sqrt{n} \mathbf{Q}}{\gamma})} \leq \alpha^{t_k}$. Therefore, we have,
\begin{align*}
    t_k &\leq \textstyle \frac{-1}{\log \alpha} \Big[\log\big(\frac{1}{\eta_{k+1}}\big) + \log\big(\frac{8((\mathcal{M} + \|x^0\|) n + k \frac{\sqrt{n} \mathbf{Q}}{\gamma})}{\beta} \big)\Big]:= \overline{t_k}.
\end{align*}
Therefore, at the $k^{th}$ iteration of Algorithm~\ref{alg:distADMM}, after $\overline{t_k}$ number of iterations of the consensus protocol we have, $\|y_i^{k+1} - \overline{y}_i^{k+1}\| \leq \eta_{k+1}$, for all $i \in \mathcal{V}$. \qed

\subsection{Proof of Theorem~\ref{thm:convergence}}\label{sec:proof_convergence}
\noindent From~(\ref{eq:xsolnoptcond}),~(\ref{eq:iteroptcond2}),
\begin{align*}
& 0 \leq [\x^{k+1} - \x^*]^\top[d\F(\x^{k+1}) - d\F(\x^*)] \\
& = [\x^{k+1} - \x^*]^\top[-\lambda^{k+1} - \gamma(\y^{k+1} - \y^{k}) - \A^\top \mu^{k+1} \\
& \hspace{1in} + \lambda^* + \A^\top \mu^*] \\
& = [\x^{k+1} - \x^*]^\top[\lambda^* -\lambda^{k+1}] + [\x^{k+1} - \x^*]^\top \A^\top[ \mu^* - \mu^{k+1}] \\
& \hspace{0.2in} - \gamma  [\x^{k+1} - \x^*]^\top [\y^{k+1} - \y^{k}]
\end{align*}
Since, $\mathcal{I}_{\mathcal{C}_0}$ is convex,
\begin{align*}
   \mathcal{I}_{\mathcal{C}_0}(\y_1) & \geq \mathcal{I}_{\mathcal{C}_0}(\y_2) + d^\top (\y_2)(\y_1 - \y_2),  
\end{align*}
where, $d(\y_2)$ is any sub-gradient of $\mathcal{I}_{\mathcal{C}_0}$ at $\y_2$. Using~(\ref{eq:iteroptcond3}) and $\y_1 = \y^*$ and $\y_2 = \overline{\y}^{k+1}$, 
\begin{align}
   \overline{\lambda}^{k+1 ^\top}(\overline{\y}^{k+1} - \y^*) \geq 0. \label{eq:lam_y-ystar1}  
\end{align}
Similarly, for $\y_1 = \overline{\y}^{k+1}, \y_2 = \y^*$ and using~(\ref{eq:ysolnoptcond}),
\begin{align}
   -\lambda^{* ^\top}(\overline{\y}^{k+1} - \y^*) \geq 0. \label{eq:lam_y-ystar2}  
\end{align}
Therefore, using~(\ref{eq:lam_y-ystar1}) and ~(\ref{eq:lam_y-ystar2}) gives,
\begin{align*}
& 0 \leq [\x^{k+1} - \x^*]^\top[\lambda^* -\lambda^{k+1}] + [\x^{k+1} - \x^*]^\top \A^\top[ \mu^* - \mu^{k+1}]  \\
& \hspace{0.1in} - \gamma  [\x^{k+1} - \x^*]^\top [\y^{k+1} - \y^{k}] + [\overline{\lambda}^{k+1} - \lambda^{*}]^\top[\overline{\y}^{k+1} - \y^*] \\
& = [\x^{k+1} - \x^*]^\top[\lambda^* -\lambda^{k+1}] + [\x^{k+1} - \x^*]^\top \A^\top[ \mu^* - \mu^{k+1}]  \\
& \hspace{0.1in} - \gamma  [\x^{k+1} - \x^*]^\top [\y^{k+1} - \y^{k}] + [\lambda^{k+1} - \lambda^{*}]^\top[\y^{k+1} - \y^*] \\
& \hspace{0.1in} + \gamma e^{k+1 ^\top}[\overline{\y}^{k+1} - \y^*] - e^{k+1 ^\top}[\lambda^{k+1} - \lambda^{*}]  \\
& = [\lambda^* -\lambda^{k+1}]^\top[\x^{k+1} - \y^{k+1}] + [\x^{k+1} - \x^*]^\top \A^\top[ \mu^* - \mu^{k+1}]  \\
& \hspace{0.1in} + \gamma[\x^* - \x^{k+1}]^\top [\y^{k+1} - \y^{k}] + \gamma e^{k+1 ^\top}[\overline{\y}^{k+1} - \y^*] \\
& \hspace{0.1in} - e^{k+1 ^\top}[\lambda^{k+1} - \lambda^{*}]\\
& = \textstyle \frac{1}{\gamma}[\lambda^{k+1} - \lambda^*]^\top[\lambda^k - \lambda^{k+1}] + \frac{1}{\gamma}[ \mu^* - \mu^{k+1}]^\top [\mu^{k+1} - \mu^k]  \\
& \hspace{0.1in} + \gamma[\x^* - \x^{k+1}]^\top [\y^{k+1} - \y^{k}] + \gamma e^{k+1 ^\top}[\overline{\y}^{k+1} - \y^*]\\
& \hspace{0.1in} - e^{k+1 ^\top}[\lambda^{k+1} - \lambda^{*}],
\end{align*}
where, the first equality used~(\ref{eq:error}), and the last equality utilized~(\ref{eq:dadmm_lam}) and~(\ref{eq:dadmm_mu}). Using~(\ref{eq:loc}) for $\frac{1}{\gamma}[\lambda^{k+1} - \lambda^*]^\top[\lambda^k - \lambda^{k+1}]$, $\frac{1}{\gamma}[ \mu^* - \mu^{k+1}]^\top [\mu^{k+1} - \mu^k]$ and $\gamma[\x^* - \x^{k+1}]^\top [\x^{k+1} - \x^{k}]$ yields,
\begin{align}
& \textstyle \frac{\gamma}{2}  \|\y^{k+1} - \x^* \|^2 + \frac{1}{2\gamma} \|\lambda^{k+1} - \lambda^*\|^2 + \frac{1}{2\gamma} \|\mu^{k+1} - \mu^*\|^2 \nonumber \\
& \textstyle \leq \frac{\gamma}{2}  \|\y^{k} - \x^* \|^2 + \frac{1}{2\gamma} \|\lambda^k - \lambda^*\|^2 + \frac{1}{2\gamma} \|\mu^k - \mu^*\|^2 \label{eq:for_seq_convg} \\
& \nonumber \hspace{0.1in} \textstyle - \frac{\gamma}{2}\|\x^{k+1} - \y^{k}\|^2 - \frac{1}{2\gamma} \|\mu^{k+1} - \mu^k\|^2 \\
&\hspace{0.1in} + \gamma e^{k+1 ^\top}[\overline{\y}^{k+1} - \y^*] - e^{k+1 ^\top}[\lambda^{k+1} - \lambda^{*}]. \nonumber
\end{align}
Let $R_1:= \gamma e^{k+1 ^\top}[\overline{\y}^{k+1} - \y^*] = \gamma (\y^{k+1} - e^{k+1} - \y^*)^\top e^{k+1} = \gamma \y^{k+1^\top} e^{k+1} - \gamma \|e^{k+1}\|^2 - \gamma \y^{*\top} e^{k+1}$. Therefore, $\|R_1\| \leq \gamma \|\y^{k+1}\|\|e^{k+1}\| + \gamma \|\y^*\| \|e^{k+1}\|$. Using, Lemma~\ref{lem:ylambdabound}, $\|R_1\| \leq n \gamma (\mathcal{M} + \|x^*\|) \eta_{k+1} + 2 \sqrt{n} \mathbf{Q} (k+1)\eta_{k+1} + \gamma \sqrt{n} \|\y^*\| \eta_{k+1}$. 
and $- e^{k+1 ^\top}[\lambda^{k+1} - \lambda^{*}]$. Let $R_2:=  - e^{k+1 ^\top}[\lambda^{k+1} - \lambda^{*}] = -e^{k+1^\top} \lambda^{k+1} + e^{k+1^\top} \lambda^*$. Therefore, using, Lemma~\ref{lem:ylambdabound}, $\|R_2\| \leq \|e^{k+1}\| \|\lambda^{k+1}\| + \|e^{k+1}\|\|\lambda^*\| \leq \sqrt{n} \mathbf{Q} (k+1) \eta_{k+1} + \|\lambda^*\|\eta_{k+1}$.\\
Let $\mathbf{R} := n \gamma (\mathcal{M} + \|x^*\|) + \gamma \sqrt{n} \|\y^*\| + \|\lambda^*\|$. Therefore, using the upper bounds on $R_1$ and $R_2$, 
\begin{align}
a_{k+1} & \leq a_k + \mathbf{R} \eta_{k+1} + 3 \sqrt{n} \mathbf{Q} (k+1) \eta_{k+1} \label{eq:uk_for_convg1} \\ 
     &\leq a_0 + \textstyle \mathbf{R} \sum_{s=0}^\infty \eta_s + 3 \sqrt{n} \mathbf{Q} \sum_{s=0}^\infty s \eta_{s}, \label{eq:uk_for_convg2}
\end{align}
where, $a_k:= \frac{\gamma}{2}  \|\y^{k} - \x^* \|^2 + \frac{1}{2\gamma} \|\lambda^{k} - \lambda^*\|^2 + \frac{1}{2\gamma} \|\mu^{k} - \mu^*\|^2$. The condition~(\ref{eq:summable}) implies that $a_k$ is bounded. Therefore, by the boundedness of $a_k$~(\ref{eq:uk_for_convg2}) there exists sub-sequence $a_{k_t}$ such that, $\lim_{t \to \infty} a_{k_t} < \infty $ exists. Therefore, using~(\ref{eq:for_seq_convg}),  
$\|\x^{k_t+1} - \y^{k_t}\| \to 0, \|\mu^{k_t+1} - \mu^k_t\| \to 0$. Note that, due to~(\ref{eq:uk_for_convg2}) $\|\y^{k_t +1} - \y_{k_t}\| \to 0, \|\x^{k_t+1} - \x_{k_t}\| = \|\x^{k_t+1} - \y^{k_t+1} + \y^{k_t + 1} - \y^{k_t} + \y^{k_t} - \x_{k_t}\| \leq \|\x^{k_t+1} - \y^{k_t+1} \| + \| \y^{k_t + 1} - \y^{k_t} \| + \| \y^{k_t} - \x_{k_t}\| \to 0$. Further, $ \|\lambda^{k_t+1}-\lambda^{k_t}\| = \gamma \|\x^{k_t+1} - \y^{k_t+1}\| \to 0$ and $\|\x^{k_t+1} - \overline{\y}^{k_t+1} \| \leq \|\x^{k_t+1} - \y^{k_t+1} \| + \| \y^{k_t+1} - \overline{\y}^{k_t+1} \| \leq \|\x^{k_t+1} - \y^{k_t+1} \| + \sqrt{n} \eta_{k+1} \to 0$. 
Therefore, there exists a sub-sequence $(\x^{k_t}, \y^{k_t}, \lambda^{k_t}, \mu^{k_t})$ that converges to a limit point $(\x^\infty, \y^\infty, \lambda^\infty, \mu^\infty)$. 
Next, we show that the limit point $(\x^\infty,\y^\infty,\lambda^\infty,\mu^\infty)$ satisfy the constraints in problem~(\ref{eq:distOpt_indifunc1}) and the objective function value $\F(\x^\infty) = \F(\x^*)$. To this end, we will utilize the following Lemma.
\begin{lemma}\label{lem:prox}
Let $\phi: \mathbb{R}^p \rightarrow \mathbb{R}$ be a convex function. Given, $\tilde{z} \in \mathbb{R}^p$ and a positive number $\gamma > 0$, if $\hat{z}$ is a proximal minimization point, i.e., $\hat{z}:= \argmin\limits_{z} \phi(z) + \frac{1}{2\gamma}\|z - \tilde{z}\|^2$, then we have, for any $z \in \mathbb{R}^p$
\begin{align}\label{eq:proxineq}
    \textstyle 2\gamma(\phi(\hat{z})-\phi(z)) \leq \|\tilde{z} - z\|^2 - \|\hat{z} - z\|^2 - \|\hat{z} - \tilde{z}\|^2.
\end{align}
\end{lemma}
\textit{Proof:} Let $\Phi(z) := \phi(z) + \frac{1}{2\gamma} \|z - \tilde{z}\|^2$. Since, $\phi$ is convex it follows that $\Phi(z)$ is strongly convex with modulus $\frac{1}{\gamma}$. By the definition of the proximal minimization step, we have $0 \in \partial \Phi(\hat{z})$. Therefore, 
\begin{align*}
  \Phi(z) &\geq \textstyle \Phi(\hat{z}) + \frac{\gamma}{2} \|z - \hat{z}\|^2 \implies \\
   \textstyle 2\gamma(\phi(\hat{z})-\phi(z)) &\leq \|\tilde{z} - z\|^2 - \|\hat{z} - z\|^2 - \|\hat{z} - \tilde{z}\|^2. \qed
\end{align*}
\noindent Using Lemma~\ref{lem:prox} for the updates~(\ref{eq:dadmm_x}) and~(\ref{eq:dadmm_y}) with $\hat{z} = \x^{k+1}, \tilde{z} = \y^{k}$, $z=\x$ and $\hat{z} = \overline{\y}^{k+1}, \tilde{z} = \x^{k+1}$, $z=\y$ we obtain,
\begin{align*}
    \textstyle \frac{2}{\gamma}&(\La(\x^{k+1},\y^{k},\lambda^k,\mu^{k}) \textstyle + \frac{\gamma}{2} \|\A \x^{k+1} - \bb\|^2 - \La(\x,\y^k,\lambda^k,\mu^{k}) \\
    & \hspace{-0.2in}\textstyle - \frac{\gamma}{2} \|\A \x - \bb\|^2) \leq \|\y^{k} - \x\|^2 - \|\x^{k+1} - \x\|^2 - \|\x^{k+1} - \y^{k}\|^2 \\
    \textstyle \frac{2}{\gamma}&(\La(\x^{k+1},\overline{\y}^{k+1},\lambda^k,\mu^{k}) - \La(\x^{k+1},\y,\lambda^k,\mu^{k})) \\
    & \leq \|\x^{k+1} - \y\|^2 - \|\overline{\y}^{k+1} - \y\|^2 - \|\overline{\y}^{k+1} - \x^{k+1}\|^2.
\end{align*}
Adding the above two inequalities we obtain, 
\begin{align*}
    \textstyle \frac{2}{\gamma}&(\La(\x^{k+1},\overline{\y}^{k+1},\lambda^k,\mu^{k}) + \textstyle \frac{\gamma}{2} \|\A \x^{k+1} - \bb\|^2 - \La(\x,\y,\lambda^k,\mu^{k}) \\
    & \hspace{-0.1in} \textstyle - \frac{\gamma}{2} \|\A \x - \bb\|^2 ) \leq \|\y^{k} - \x\|^2 + \|\x^{k+1} - \y\|^2 - \|\x^{k+1} - \x\|^2 \\
    & \hspace{0.2in} - \|\overline{\y}^{k+1} - \y\|^2 - \|\x^{k+1} - \y^{k}\|^2 - \|\overline{\y}^{k+1} - \x^{k+1}\|^2.
\end{align*}
Taking limit over the appropriate sub-sequences on both sides of the above inequality, we get,
\begin{align}\label{eq:saddle1}
    \hspace{-0.15in} \La(\x^\infty,\y^\infty,\lambda^\infty,\mu^\infty) \leq \La(\x,\y,\lambda^\infty,\mu^\infty) + \textstyle \frac{\gamma}{2} \|\A \x - \bb\|^2
\end{align}
Note that the updates~(\ref{eq:dadmm_mu}) and~(\ref{eq:dadmm_lam_ne}) can be written as the following minimization updates: $\forall i\in \mathcal{V}$,
\begin{align}
    \mu_i^{k+1} &= \textstyle \argmin\limits_{\mu_i} \{\textstyle -\lambda_i^{k^\top} (x_i^{k+1} - y_i^{k+1}) - \mu_{i}^\top(\textcolor{black}{A_i} x_i^{k+1} - \textcolor{black}{b_i})\nonumber \\
   & \hspace{0.7in} \textstyle  + \frac{1}{2\gamma} \|\mu_i - \mu_i^k\|^2 \},  \label{eq:mu_update_min} \\
   & = \mu_i^{k} + \textstyle \gamma (\textcolor{black}{A_i} x_i^{k+1} - \textcolor{black}{b_i}),  \nonumber\\
   & \hspace{-0.35in} \overline{\lambda}_i^{k+1} = \textstyle \argmin\limits_{\lambda_i} \{\textstyle -\lambda_i^{\top} (x_i^{k+1} - \overline{y}_i^{k+1}) \textstyle + \frac{1}{2\gamma} \|\lambda_i - \lambda_i^k\|^2 \} \label{eq:lam_update_min}\\
   & = \lambda_i^{k} + \textstyle \gamma(x_i^{k+1} - \overline{y}_i^{k+1}). \nonumber
\end{align}
Applying Lemma~\ref{lem:prox} to the updates~(\ref{eq:mu_update_min}) and~(\ref{eq:lam_update_min}) with $\hat{z} = \mu^{k+1}, z = \mu, \tilde{z} = \mu^k$ and $\hat{z} = \overline{\lambda}^{k+1}, z = \lambda, \tilde{z} = \lambda^k$ respectively and adding the two inequalities result in,
\begin{align*}
    \textstyle 2\gamma&(\La(\x^{k+1},\overline{\y}^{k+1},\lambda,\mu) - \La(\x^{k+1},\overline{\y}^{k+1},\overline{\lambda}^{k+1},\mu^{k+1})) \\
    & \leq \|\mu^k - \mu\|^2 + \|\lambda^k - \lambda\|^2 - \|\mu^{k+1} - \mu\|^2 \\
    & - \|\mu^{k+1} - \mu^k\|^2- \|\overline{\lambda}^{k+1} - \lambda\|^2 - \|\overline{\lambda}^{k+1} - \lambda^k\|^2.
\end{align*}
Taking limit over the appropriate sub-sequences on both sides of the above inequality, we get,
\begin{align}\label{eq:saddle2}
    \La(\x^\infty,\y^\infty,\lambda,\mu) \leq \La(\x^\infty,\y^\infty,\lambda^\infty,\mu^\infty).
\end{align}
Using~(\ref{eq:saddle1}) with $\x = \x^*, \y=\y^*$,~(\ref{eq:saddle2}) with $\lambda = \lambda^*, \mu = \mu^*$, the saddle point relation in Assumption~\ref{ass:prob_ass1} and~(\ref{eq:eqconst_opt}), we get,
\begin{align}
   &\La(\x^*,\y^*,\lambda^*,\mu^*) \leq  \La(\x^\infty,\y^\infty,\lambda^*, \mu^*) \leq \La(\x^\infty,\y^\infty,\lambda^\infty, \nonumber\\ 
   & \mu^\infty) \leq \La(\x^*,\y^*,\lambda^\infty,\mu^\infty) \leq \La(\x^*,\y^*,\lambda^*,\mu^*).  \nonumber\\
   & \mbox{Therefore, } \ \La(\x^*,\y^*,\lambda^*,\mu^*) = \La(\x^\infty,\y^\infty,\lambda^\infty,\mu^\infty). 
\end{align}
This implies that $\F(\x^\infty) + \mathcal{I}_{\mathcal{C}_0}(\y^\infty) + \lambda^{\infty \top}(\x^\infty - \y^\infty) + \mu^{\infty \top} (\A \x^\infty - \bb) = \F(\x^*) < \infty$. Therefore, $\y^\infty \in \mathcal{C}_0$. Further, since, $\x^\infty = \y^\infty$ as $\|\lambda^{k+1} - \lambda^k\| \to 0$ and $\A \x^\infty = \bb$ since, $\|\mu^{k+1} - \mu^{k} \| \to 0$. Thus, $\F(\x^\infty) = \F(\x^*)$.  To complete the proof, we show that the $(\x^{k},\y^{k},\lambda^{k},\mu^{k})$ has a unique limit point; we utilize the same argument as in \cite{rockafellar1976monotone}, Theorem~{1}. Let $r_1=(\x_1^\infty,\y_1^\infty,\lambda_1^\infty,\mu_1^\infty)$ and $r_2=(\x_2^\infty,\y_2^\infty,\lambda_2^\infty,\mu_2^\infty)$ be any two limit points of $(\x^k,\y^k,\lambda^k,\mu^k)$. As shown above, both $r_1$ and $r_2$ are saddle points of $\La(\x,\y,\lambda,\mu)$. Then, from~(\ref{eq:uk_for_convg2}) for appropriate sub-sequences we have $\lim_{k_t \rightarrow \infty} \|\x^{k_t} - \x_i^\infty\|^2 + \|\y^{k_t} - \y_i^\infty\|^2 + \|\lambda^{k_t} - \lambda_i^\infty\|^2 + \|\mu^{k_t} - \mu_i^\infty\|^2 = \hat{r}_j < 0$, for $j = 1,2.$ Consider,
\begin{align*}
  & \|\x^k - \x_1^\infty\|^2 + \|\y^k - \y_1^\infty\|^2 + \|\lambda^k - \lambda_1^\infty\|^2 + \|\mu^k - \mu_1^\infty\|^2 \\ 
  & - \|\x^k - \x_2^\infty\|^2 - \|\y^k - \y_2^\infty\|^2 - \|\lambda^k - \lambda_2^\infty\|^2 - \|\mu^k - \mu_2^\infty\|^2 \\ 
  & = \|r_1\|^2 - \|r_2\|^2 - 2(r_1 - r_2)^\top(\x^{k\top},\y^{k\top},\lambda^{k\top},\mu^{k\top}).
\end{align*}
Taking the limit both sides for each limit point we obtain,
\begin{align*}
  \hat{r}_1 - \hat{r}_2 &= \|r_1\|^2 - \|r_2\|^2 \\ 
  &- 2 ( \x_1^{\infty \top}(\x_1^\infty - \x_2^\infty) + \y_1^{\infty \top}(\y_1^\infty - \y_2^\infty) \\
   & + \lambda_1^{\infty \top}(\lambda_1^\infty - \lambda_2^\infty) + \mu_1^{\infty \top}(\mu_1^\infty - \mu_2^\infty))\\
   & = -\|r_1-r_2\|^2, \ \mbox{and}, \\
   \hat{r}_1 - \hat{r}_2 &= \|r_1\|^2 - \|r_2\|^2 \\ 
  &- 2 ( \x_2^{\infty \top}(\x_1^\infty - \x_2^\infty) + \y_2^{\infty \top}(\y_1^\infty - \y_2^\infty) \\
   & \hspace{-0.2in} + \lambda_2^{\infty \top}(\lambda_1^\infty - \lambda_2^\infty) + \mu_2^{\infty \top}(\mu_1^\infty - \mu_2^\infty)) = \|r_1- r_2\|^2.
\end{align*}
Therefore, we must have $\|r_1-r_2\| = 0$ and hence, $(\x^\infty,\y^\infty,\lambda^\infty,\mu^\infty)$ is unique. This completes the proof. \qed

\subsection{Proof of Theorem~\ref{thm:order1/k}}\label{sec:proof_order1/k}
\noindent From~(\ref{eq:iteroptcond2}) we have, for $i=1,\dots,n$,
\begin{align*}
    df_i(x_i^{k+1}) + \lambda_i^{k+1} + \textcolor{black}{A_i}^\top \mu_i^{k+1} +  \gamma(y_i^{k+1} - y_i^{k}) = 0, 
\end{align*}
where, $df_i(x_i^{k+1})$ is a sub-gradient of $f_i$ at $x_i^{k+1}$. Writing the above $n$ inequalities compactly,
\begin{align}
    d\F&(\x^{k+1}) + \lambda^{k+1} + \A^\top\mu^{k+1} + \gamma(\y^{k+1} - \y^{k}) = 0 \nonumber,
\end{align}
where, $d\F(\x^{k+1})$ is a vector with subgradients $df_i(x_i^{k+1})$ stacked together respectively. Further, using~(\ref{eq:iteroptcond3}) $d(\overline{\y}^{k+1}) + \overline{\lambda}^{k+1} = 0$, where $d(\overline{\y}^{k+1})$ is a sub-gradient of $\mathcal{I}_{\mathcal{C}_0}$ at $\overline{\y}^{k+1}$. Noticing that $\F$ and $\mathcal{I}_{\mathcal{C}_0}$ are convex functions and using~(\ref{eq:dadmm_lam_ne}) we have,
\begin{align}
    \F(\x&^{k+1}) -  \F(\x^*) + \mathcal{I}_{\mathcal{C}_0}(\overline{\y}^{k+1}) - \mathcal{I}_{\mathcal{C}_0}(\y^*) \nonumber\\
    \leq & - [\x^* - \x^{k+1}]^\top d\F(\x^{k+1}) - [\y^* - \overline{\y}^{k+1}]^\top d(\overline{\y}^{k+1}) \nonumber\\
    \leq & \ [\x^* - \x^{k+1}]^\top[\lambda^{k+1} + \A^\top\mu^{k+1} + \gamma(\y^{k+1} - \y^{k})] \nonumber \\
    &  - [\y^* - \overline{\y}^{k+1}]^\top \overline{\lambda}^{k+1} \nonumber 
\end{align}    
Note that $\mathcal{I}_{\mathcal{C}_0}(\overline{\y}^{k+1}) = \mathcal{I}_{\mathcal{C}_0}(\y^*) = 0$. Using~(\ref{eq:dadmm_lam}) and~(\ref{eq:dadmm_lam_ne}), 
\begin{align*}
  & \F(\x^{k+1}) - \F(\x^*) \\
  & \leq [\x^* - \x^{k+1}]^\top[\lambda^{k+1} + \A^\top\mu^{k+1} + \gamma(\y^{k+1} - \y^{k})] \nonumber \\
    & \hspace{0.2in} - [\y^* - \overline{\y}^{k+1}]^\top[\lambda^{k+1} + \gamma(\y^{k+1} - \overline{\y}^{k+1})] \\ 
    & \leq [\x^* - \x^{k+1}]^\top[\lambda^{k+1} + \A^\top\mu^{k+1} + \gamma(\y^{k+1} - \y^{k})] \nonumber \\
    & \hspace{0.2in} \textstyle - [\y^* - \overline{\y}^{k+1}]^\top \lambda^{k+1} + \gamma e^{k+1^\top}(\overline{\y}^{k+1} - \y^*) \\
    & = [\x^* - \x^{k+1}]^\top[\lambda^{k+1} + \A^\top\mu^{k+1} + \gamma(\y^{k+1} - \y^{k})] \nonumber \\
    & \hspace{0.2in} \textstyle - [\y^* - \y^{k+1}]^\top \lambda^{k+1} + e^{k+1^\top} (-\lambda^{k+1} + \gamma (\overline{\y}^{k+1} - \y^*)) \\
    & = [\y^{k+1} - \x^{k+1}]^\top \lambda^{k+1} +  \gamma[\x^* - \x^{k+1}]^\top[\y^{k+1} - \y^{k}] \nonumber \\
    & \hspace{0.02in} + [\x^* - \x^{k+1}]^\top\A^\top\mu^{k+1} + e^{k+1^\top} (-\lambda^{k+1} + \gamma (\overline{\y}^{k+1} - \y^*)).
\end{align*}
Therefore, for any $\lambda$ and $\mu$, using~(\ref{eq:eqconst_opt}) we have
\begin{align}
    &\F(\x^{k+1}) -  \F(\x^*) + \lambda^\top (\x^{k+1} - \y^{k+1}) + \mu^\top (\A\x^{k+1} - \bb) \leq \nonumber \\
    & [\lambda^{k+1} - \lambda]^\top [\y^{k+1} - \x^{k+1}] +  \gamma[\x^* - \x^{k+1}]^\top[\y^{k+1} - \y^{k}] + \nonumber\\
    & [\mu-\mu^{k+1}]^\top \A (\x^{k+1} - \bb) + e^{k+1^\top} (-\lambda^{k+1} + \gamma (\overline{\y}^{k+1} - \y^*)) \nonumber. 
\end{align}
Using~(\ref{eq:dadmm_lam}) and~(\ref{eq:dadmm_mu}), we get,
\begin{align}
    &\F(\x^{k+1}) -  \F(\x^*) + \lambda^\top (\x^{k+1} - \y^{k+1}) + \mu^\top (\A\x^{k+1} - \bb) \nonumber \\
    & \leq  \textstyle \frac{1}{\gamma}[\lambda - \lambda^{k+1} ]^\top [\lambda^{k+1} - \lambda^{k}] +  \gamma[\x^* - \x^{k+1}]^\top[\y^{k+1} - \y^{k}] \nonumber\\
    & \textstyle \hspace{0.2in} + \frac{1}{\gamma}[\mu-\mu^{k+1}]^\top [\mu^{k+1} - \mu^{k}] \nonumber \\
    &  \hspace{0.2in}  + e^{k+1^\top} (-\lambda^{k+1} + \gamma (\overline{\y}^{k+1} - \y^*)), \label{eq:one}
\end{align} 
Using the identity~(\ref{eq:loc}) for any $s \geq 0$ we get,
\begin{align}
  \textstyle \frac{1}{\gamma}[\lambda - \lambda^{s+1} ]^\top [\lambda^{s+1} - \lambda^{s}] & = \textstyle \frac{1}{2\gamma} (\|\lambda - \lambda^{s}\|^2 \label{eq:two}\\
    - \|\lambda & - \lambda^{s+1}\|^2 - \|\lambda^{s+1} - \lambda^{s}\|^2 ), \nonumber \\
  \gamma [\x^*  -\x^{s+1}]^\top[\y^{s+1} - \y^{s}]  &= \textstyle \frac{\gamma}{2} (\|\x^{s+1} - \y^{s+1}\|^2    \nonumber \\
  - \|\x^{s+1} - \y^{s}\|^2 + \|\x^* & - \y^{s}\|^2 - \|\x^* - \y^{s+1}\|^2 ), \label{eq:three} \\
  & \hspace{-1.65in} \textstyle \frac{1}{\gamma}[\mu-\mu^{k+1}]^\top [\mu^{k+1} - \mu^{k}] = \textstyle \frac{1}{2\gamma} (\|\mu - \mu^k\|^2 - \|\mu - \mu^{k+1}\|^2 \nonumber \\
  & \hspace{0.35in} - \|\mu^{k+1} - \mu^k\|^2) \label{eq:four}.
\end{align}
Note that from Lemma~\ref{lem:ylambdabound}, ~(\ref{eq:dadmm_lam}) and~(\ref{eq:error}), under Assumption~\ref{ass:constraint_ass},
$e^{k+1^\top} (-\lambda^{k+1} + \gamma (\overline{\y}^{k+1} - \y^*)) = - e^{k+1^\top} \lambda^{k+1} + \gamma e^{k+1^\top} \y^{k+1} - \gamma \|e^{k+1}\|^2 - e^{k+1^\top}\y^* \leq \|e^{k+1}\|\|\lambda^{k+1}\| + \gamma \|e^{k+1}\| \|\y^{k+1}\| + \|e^{k+1}\|\|\y^*\| \leq (k+1) \mathbf{Q} \sqrt{n} \eta_{k+1} + (\gamma n (\mathcal{M} + \|x^*\|) + 2(k+1) \sqrt{n} \mathbf{Q})  \eta_{k+1} + \sqrt{n}\|\y^*\| \eta_{k+1}$.  Let $\M:= \gamma n (\mathcal{M} + \|x^*\|) + \sqrt{n}\|\y^*\|$.
From~(\ref{eq:one}),~(\ref{eq:two}),~(\ref{eq:three}) and~(\ref{eq:four}) for any $\lambda, \mu$ and $s \geq 0$ we get,
\begin{align*}
  & \F(\x^{s+1}) -  \F(\x^*) + \lambda^\top (\x^{s+1} - \y^{s+1}) + \mu^\top (\A\x^{k+1} - \bb)\\
  & \leq  \textstyle \frac{1}{2\gamma} (\|\lambda - \lambda^{s}\|^2 - \|\lambda - \lambda^{s+1}\|^2 ) + \frac{\gamma}{2} (\|\x^* - \y^{s}\|^2   \\
  & \textstyle \hspace{0.15in} - \|\x^* - \y^{s+1}\|^2) + \frac{1}{2\gamma} (\|\mu - \mu^s\|^2 - \|\mu - \mu^{s+1}\|^2) \\
  & \textstyle \hspace{0.15in} +  \frac{\M}{s^{2+q}}  + 3\sqrt{n} \mathbf{Q}(s+1) \frac{1}{(s+1)^{2+q}}.
\end{align*}
Summing from $s=0$ to $k - 1$ and dividing by $\frac{1}{k}$ gives,
\begin{align*}
  & \textstyle \frac{1}{k} \sum_{s=0}^{k-1}[ \F(\x^{s+1}) -  \F(\x^*) + \lambda^\top (\x^{s+1} - \y^{s+1}) \\
  & \mu^\top (\A\x^{s+1} - \bb)] \leq \textstyle \frac{\gamma \|\x^* - \y^0\|^2}{2k} + \frac{\|\lambda - \lambda^0\|^2}{2\gamma k} + \frac{\|\mu - \mu^0\|^2}{2\gamma k} \\
  & \textstyle \hspace{.82in} + \frac{\M \sum_{s=0}^{k-1} \frac{1}{s^{2+q}}}{k} + \frac{3\sqrt{n} \mathbf{Q} \sum_{s=0}^{k-1} (s+1) \frac{1}{(s+1)^{2+q}}}{k},
\end{align*}
Let $\widehat{\x}^k := \frac{1}{k} \sum_{s=0}^{k-1}\x^{s+1}$, and $\widehat{\y}^k := \frac{1}{k} \sum_{s=0}^{k-1} \y^{s+1}$. Since, $\F$ is a convex function we have, $\frac{1}{k}\sum_{s=0}^{k-1}\F(\x^{s+1}) \geq \F(\widehat{\x}^k)$. Therefore,
\begin{align}
  & \F(\widehat{\x}^k) -  \F(\x^*) + \lambda^\top (\widehat{\x}^k - \widehat{\y}^k) + \mu^\top (\A\widehat{\x}^k - \bb) \leq \label{eq:reuse_eq} \\
  & \textstyle \frac{\gamma \|\x^* - \y^0\|^2}{2k} + \frac{\|\lambda - \lambda^0\|^2}{2\gamma k} + \frac{\|\mu - \mu^0\|^2}{2\gamma k} + \frac{\M \zeta(2+q)+ 3\sqrt{n} \mathbf{Q} \zeta(1+q)}{k} \nonumber ,
\end{align}
where, $\zeta(.)$ is the Riemann zeta function. 
Let $\mu = 0, \lambda = 0$, therefore, we get,
\begin{align*} 
  & \F(\widehat{\x}^k) -  \F(\x^*) \leq \textstyle \frac{\gamma \|\x^* - \y^0\|^2}{2k} + \frac{\|\lambda^0\|^2 + \|\mu^0\|^2}{2\gamma k} \\
  & \hspace{1.2in} \textstyle + \frac{\M \zeta(2+q)+ 3\sqrt{n} \mathbf{Q} \zeta(1+q)}{k} = O(1/k).
\end{align*}  
Define, $\widehat{\overline{\y}}^k:= \frac{1}{k} \sum_{s=0}^{k-1} \overline{\y}^{s+1}$, where, $\overline{\y}^{s+1}$ is as defined in~(\ref{eq:y_over}). Since, $(\x^*,\y^*,\lambda^*,\mu^*)$ is a saddle point we have, $\La(\x^*,\y^*,\lambda,\mu) \leq \La(\widehat{\x}^k,\widehat{\overline{\y}}^k,\lambda^*,\mu^*)$. Therefore, $\F(\widehat{\x}^k) - \F(\x^*) + \lambda^{*\top}(\widehat{x}^k - \widehat{\y}^k) + \lambda^{*\top}(\widehat{\y}^k - \widehat{\overline{\y}}^k) + \mu^{*\top}(\A \widehat{\x}^k - \bb) \geq 0$. Let $\lambda = \lambda^* + \frac{\widehat{x}^k - \widehat{y}^k}{\|\widehat{x}^k - \widehat{y}^k\|}$ and $\mu = \mu^*$. Therefore,
\begin{align*} 
  & \|\widehat{\x}^k - \widehat{\y}^k\| \leq \F(\widehat{\x}^k) -  \F(\x^*) + \lambda^{*\top} (\widehat{\x}^k - \widehat{\y}^k)  + \nonumber \\
  & \mu^{*\top} (\A\widehat{\x}^k - \bb) + \lambda^{*\top}(\widehat{\y}^k - \widehat{\overline{\y}}^k) + \|\widehat{\x}^k - \widehat{\y}^k\| \leq \textstyle \frac{\gamma \|\x^* - \y^0\|^2}{2k} \nonumber \\
  & \textstyle + \frac{1 + \|\lambda^0 - \lambda^*\|^2}{\gamma k} + \frac{\|\mu^* - \mu^0\|^2}{2\gamma k} + \frac{\M \zeta(2+q)+ 3\sqrt{n} \mathbf{Q} \zeta(1+q)}{k} \\
  & + \lambda^{*\top}(\widehat{\y}^k - \widehat{\overline{\y}}^k).
\end{align*}
Note that, $\|\lambda^{*\top}(\widehat{\y}^k - \widehat{\overline{\y}}^k)\| \leq \|\lambda^*\|\| \frac{1}{k} \sum_{s=0}^{k-1} [\y^{s+1} - \overline{\y}^{s+1}]  \| = \frac{\|\lambda^*\|}{k} \sum_{s=0}^{k-1} \|\y^{s+1} - \overline{\y}^{s+1}\| = \frac{\|\lambda^*\|}{k} \sum_{s=0}^{k-1} \|e^{s+1}\| \leq \frac{\|\lambda^*\|}{k} \sum_{s=0}^{k-1} \sqrt{n} \eta_{s+1} \leq \frac{\sqrt{n}\|\lambda^*\|}{k} \sum_{s=0}^\infty \eta_{s+1} = \frac{\sqrt{n}\|\lambda^*\|}{k} \zeta(2+q).$ Therefore, 
\begin{align*}
    \|\widehat{\x}^k - \widehat{\y}^k\| &\leq \textstyle \frac{\gamma \|\x^* - \y^0\|^2}{2k} + \frac{1 + \|\lambda^0 - \lambda^*\|^2}{\gamma k} + \frac{\|\mu^* - \mu^0\|^2}{2\gamma k} \\
    & \hspace{0.08in} \textstyle + \frac{(\M + \sqrt{n}\|\lambda^*\|)\zeta(2+q)+ 3\sqrt{n} \mathbf{Q} \zeta(1+q)}{k} = O(1/k).
\end{align*}
Alternatively, let, $\mu = \mu^* + \frac{\A\widehat{\x}^k - \bb}{\|\A\widehat{\x}^k - \bb\|}$ and $\lambda = \lambda^*$, we get,
\begin{align*} 
  & \|\A\widehat{\x}^k - \bb\| \leq \F(\widehat{\x}^k) -  \F(\x^*) + \lambda^{*\top} (\widehat{\x}^k - \widehat{\y}^k) + \nonumber \\
  & \mu^{*\top} (\A\widehat{\x}^k - \bb) + \lambda^{*\top}(\widehat{\y}^k - \widehat{\overline{\y}}^k) + \|\A\widehat{\x}^k - \bb\| \leq \textstyle \frac{\gamma \|\x^* - \y^0\|^2}{2k}  \nonumber \\
  & \textstyle + \frac{\|\lambda^0 - \lambda^*\|^2}{2\gamma k} + \frac{1 + \|\mu - \mu^0\|^2}{\gamma k} + \frac{\M \zeta(1+q)}{k} + \lambda^{*\top}(\widehat{\y}^k - \widehat{\overline{\y}}^k) \\
  & \textstyle \leq \frac{\gamma \|\x^* - \y^0\|^2}{2k} + \frac{\|\lambda^0 - \lambda^*\|^2}{2\gamma k} + \frac{1 + \|\mu - \mu^0\|^2}{\gamma k} \\
  & \hspace{0.08in}\textstyle + \frac{(\M + \sqrt{n}\|\lambda^*\|)\zeta(2+q)+ 3\sqrt{n} \mathbf{Q} \zeta(1+q)}{k} = O(1/k). \hspace{0.5in} \qed
\end{align*}

\subsection{Proof of Theorem~\ref{thm:linearrate}}\label{sec:proof_linearrate}
\noindent From the definition of $\F$ and Assumption~\ref{ass:lipschitzgrad_f_strconv} it follows that $\F$ is Lipschitz differentiable with constant $L := \max_{1\leq i \leq n} L_{f_i}$ for all $\tilde{\x} \in \relint(\mathcal{X}), \tilde{\x}_s \in \relint(^n\mathbb{R}^m_{\geq 0})$ and restricted strongly convex with respect to the optimal solution $\x^*$ on the set $\mathcal{X}$ with parameter $\sigma := \min_{1 \leq i \leq n} \sigma_i$. 
For sub-gradients $d\F(\x^{k+1}) \in \partial \F(\x^{k+1})$ and $d\F(\x^*) \in \partial \F(\x^*)$, using~(\ref{eq:xsolnoptcond}) and~(\ref{eq:iteroptcond2}) under Assumption~\ref{ass:lipschitzgrad_f_strconv},
\begin{align*}
& \textstyle \sigma\| \x^{k+1} - \x^* \|^2  \leq [\x^{k+1} - \x^*]^\top[d\F(\x^{k+1}) - d\F(\x^*)] \\
& = [\x^{k+1} - \x^*]^\top[-\lambda^{k+1} - \gamma(\y^{k+1} - \y^{k}) - \A^\top \mu^{k+1} \\
& \hspace{1 in} + \lambda^* + \A^\top \mu^*] \\
& = [\x^{k+1} - \x^*]^\top[\lambda^* - \lambda^{k+1}] - \gamma  [\x^{k+1} - \x^*]^\top [\y^{k+1} - \y^{k}]\\ 
& \hspace{0.1in} + [\x^{k+1} - \x^*]^\top \A^\top [\mu^* - \mu^{k+1}] \\
& = [\x^{k+1} - \x^*]^\top[\lambda^*-\lambda^{k+1}] + [\x^{k+1} - \x^*]^\top \A^\top [\mu^* - \mu^{k+1}]  \\
& \hspace{0.1in} + \gamma[\x^* - \x^{k+1}]^\top [\x^{k+1} - \x^{k}] \\
& \hspace{0.1in} + [\x^* - \x^{k+1}]^\top [2\lambda^k - \lambda^{k+1} - \lambda^{k-1}]. 
\end{align*}
Using,~((\ref{eq:lam_y-ystar1}) and (\ref{eq:lam_y-ystar2}) we get,
\begin{align*}
& \leq [\x^{k+1} - \x^*]^\top[\lambda^* -\lambda^{k+1}] + \gamma[\x^* - \x^{k+1}]^\top [\x^{k+1} - \x^{k}] \\
& \hspace{0.03in} + [\x^{k+1} - \x^*]^\top \A^\top [\mu^* - \mu^{k+1}] + [\overline{\y}^{k+1} - \y^*]^\top [\lambda^{k+1} - \lambda^*] \\
& \hspace{0.1in} - [\overline{\y}^{k+1} - \y^*]^\top \lambda^{k+1} + [\x^* - \x^{k+1}]^\top [2\lambda^k - \lambda^{k+1} - \lambda^{k-1}] \\
& = [\lambda^* - \lambda^{k+1}]^\top[\x^{k+1} - \y^{k+1}] + \gamma[\x^* - \x^{k+1}]^\top [\x^{k+1} - \x^{k}] \\
& \hspace{0.1in} + [\x^{k+1} - \x^*]^\top \A^\top [\mu^* - \mu^{k+1}] - [\overline{\y}^{k+1} - \y^*]^\top \lambda^{k+1}\\
& \hspace{0.1in}  + [\lambda^* - \lambda^{k+1}]^\top  e^{k+1} + [\x^* - \x^{k+1}]^\top [2\lambda^k - \lambda^{k+1} - \lambda^{k-1}]\\
& = \textstyle \frac{1}{\gamma}[\lambda^* - \lambda^{k+1} ]^\top[\lambda^{k+1} - \lambda^{k}] + \gamma[\x^* - \x^{k+1}]^\top [\x^{k+1} - \x^{k}]  \\
& \hspace{0.1in} + \textstyle \frac{1}{\gamma}[\mu^* - \mu^{k+1}]^\top[\mu^{k+1} - \mu^{k}] - [\overline{\y}^{k+1} - \y^*]^\top \lambda^{k+1} \\
& \hspace{0.1in} + [\lambda^* - \lambda^{k+1}]^\top  e^{k+1} + [\x^* - \x^{k+1}]^\top [2\lambda^k - \lambda^{k+1} - \lambda^{k-1}],
\end{align*}
where, in the second equality we used~(\ref{eq:error}), and in the last equality we used~(\ref{eq:dadmm_lam}) and~(\ref{eq:dadmm_mu}). Let $t_k:= 2\lambda^k - \lambda^{k+1} - \lambda^{k-1}$. Using~(\ref{eq:loc}) for $\frac{1}{\gamma}[\lambda^* - \lambda^{k+1} ]^\top[\lambda^{k+1} - \lambda^{k}]$, $\textstyle \frac{1}{\gamma}[\mu^* - \mu^{k+1}]^\top[\mu^{k+1} - \mu^{k}]$, and $\gamma[\x^* - \x^{k+1}]^\top [\x^{k+1} - \x^{k}]$,
\begin{align}
& \textstyle \sigma\| \x^{k+1} - \x^* \|^2 + \frac{\gamma}{2}\|\x^{k+1} - \x^{k}\|^2 + \frac{1}{2\gamma} \|\lambda^{k+1} - \lambda^k\|^2 \nonumber \\
& \hspace{-0.05in} \textstyle \leq \left[ \frac{\gamma}{2}  \|\x^{k} - \x^* \|^2 + \frac{1}{2\gamma} \|\lambda^k - \lambda^*\|^2 + \frac{1}{2\gamma} \|\mu^{k} - \mu^* \|^2 \right] \label{eq:for_lin_rate}\\
& \hspace{-0.05in}  \textstyle - \left[ \frac{\gamma}{2}  \|\x^{k+1} - \x^* \|^2 + \frac{1}{2\gamma} \|\lambda^{k+1} - \lambda^*\|^2 + \frac{1}{2\gamma} \|\mu^{k+1} - \mu^* \|^2  \right] \nonumber  \\ 
& \hspace{-0.05in}  - [\overline{\y}^{k+1} - \y^*]^\top \lambda^{k+1} + [\lambda^* - \lambda^{k+1}]^\top  e^{k+1} + [\x^* - \x^{k+1}]^\top t_k. \nonumber
\end{align}
Observe from~(\ref{eq:xsolnoptcond}),~(\ref{eq:iteroptcond2}), $\nabla \F(\x^{k+1}) - \nabla \F(\x^*) = -\lambda^{k+1} -\gamma(\y^{k+1} - \y^{k}) - \A^\top \mu^{k+1} + \lambda^* + \A^\top \mu^*  = -\lambda^{k+1} + \lambda^* - \A^\top \mu^{k+1} + \A^\top \mu^* + \lambda^{k+1}+\lambda^{k-1} - 2 \lambda^k - \gamma(\x^{k+1} - \x^{k})$, where we used~(\ref{eq:dadmm_lam}) to get a relation between $\y^{k+1}$ and $\y^k$. Therefore, $\gamma (\x^{k+1} - \x^k)$ is the summation of the terms $ \mathbf{q}_1 := \nabla \F(\x^*) - \nabla \F(\x^{k+1})$ and $ \mathbf{r}_1 := \lambda^* -\lambda^{k+1} + \A^\top \mu^* - \A^\top \mu^{k+1} + \lambda^{k+1}+\lambda^{k-1} - 2 \lambda^k$. Therefore, for any $\delta > 1$, we can apply the inequality $\|\mathbf{q}_1+\mathbf{r}_1\|^2 + (\delta- 1)\|\mathbf{q}_1\|^2 \geq (1 - \frac{1}{\delta} )\|\mathbf{r}_1\|^2$,
\begin{align*}
    & \gamma^2 \|\x^{k+1} - \x^k\|^2 + (\delta - 1)
    \|\nabla \F(\x^{k+1}) - \nabla \F(\x^*)\|^2  \\
    & \geq \textstyle (1 - \frac{1}{\delta}) \|\lambda^* -\lambda^{k+1} + \A^\top \mu^* - \A^\top \mu^{k+1} - t_k\|^2 \\
    & \geq \textstyle (1 - \frac{1}{\delta}) \Big[ \| \lambda^{k+1} - \lambda^* \|^2 + \nu_{\min} (\A\A^\top) \| \mu^{k+1} - \mu^*\|^2 \\
    & \hspace{0.45in} + \|t_k\|^2 - 2 \|\lambda^{k+1} -\lambda^* \| \| \A^\top \mu^* - \A^\top \mu^{k+1} - t_k\| \\
    & \hspace{0.45in} - 2 \|t_k\| \|\A^\top \mu^* - \A^\top \mu^{k+1}\|\Big],
\end{align*}
where, $\nu_{\min} (\A\A^\top) > 0$. By Lipschitz differentiablility of $\F$,
\begin{align}
    &\textstyle \frac{\gamma}{2} \|\x^{k+1} - \x^k\|^2 + \frac{(\delta - 1) L^2}{2\gamma} \| \x^{k+1} - \x^*\|^2  \label{eq:sk_1}\\
    & \geq \textstyle \frac{(1 - \frac{1}{\delta})}{2\gamma}  \Big[ \| \lambda^{k+1} - \lambda^*\|^2 + \nu_{\min} (\A\A^\top) \| \mu^{k+1} - \mu^*\|^2 \nonumber \\
    & \hspace{0.62in} - 2 \|\lambda^{k+1} - \lambda^* \| \| \A^\top \mu^* - \A^\top \mu^{k+1} - t_k\| \nonumber \\
    & \hspace{0.62in} - 2 \|t_k\| \|\A^\top \mu^* - \A^\top \mu^{k+1}\| \Big] \nonumber. 
\end{align}
Similarly, $\lambda^{k+1} - \lambda^k = \gamma (\x^{k+1} - \y^{k+1}) = \gamma (\x^{k+1} - \x^*) + \gamma(\y^* - \y^{k+1})$. Therefore, $\lambda^{k+1} - \lambda^k$ is the summation of $\mathbf{q}_2:= \gamma (\x^{k+1} - \x^*)$ and $\mathbf{r}_2:= \gamma (\y^* - \y^{k+1})$. Therefore, for any $\delta > 1$, we can apply the inequality $\|\mathbf{q}_2+\mathbf{r}_2\|^2 + (\delta- 1)\|\mathbf{q}_2\|^2 \geq (1 - \frac{1}{\delta} )\|\mathbf{r}_2\|^2$, and obtain,
\begin{align}
    & \|\lambda^{k+1} - \lambda^k\|^2 + \gamma^2(\delta - 1)
    \|\x^{k+1} - \x^*\|^2 \nonumber \\
    & \geq \textstyle (1 - \frac{1}{\delta}) \|\gamma (\y^{k+1} - \y^*)\|^2 \nonumber \\
    & = \textstyle (1 - \frac{1}{\delta}) \| \lambda^{k+1} - \lambda^k - \gamma (\x^{k+1} - \x^*)\|^2 \nonumber\\
    & \geq \textstyle (1 - \frac{1}{\delta}) \Big[ \| \lambda^{k+1} - \lambda^k\|^2 + \gamma^2 \|\x^{k+1} - \x^*\|^2 \nonumber \\
    & \hspace{0.65in} - 2 \gamma \|\lambda^{k+1} - \lambda^k \| \|\x^{k+1} - \x^*\|\Big] \label{eq:sk_2}.
\end{align}
Therefore, using~(\ref{eq:sk_1}) and~(\ref{eq:sk_2}) we get,
\begin{align}
    & \textstyle \frac{\gamma}{2} \|\x^{k+1} - \x^k\|^2 + \frac{1}{2\gamma}\|\lambda^{k+1} - \lambda^k\|^2 + \frac{(\delta - 1)(\frac{L^2}{\gamma} + \gamma)}{2} 
    \|\x^{k+1} - \x^*\|^2 \nonumber \\
    & \geq \textstyle (1 - \frac{1}{\delta})  \Big[ \frac{1}{2 \gamma}\| \lambda^{k+1} - \lambda^*\|^2 + \frac{\gamma}{2} \|\x^{k+1} - \x^*\|^2 \nonumber \\
    & \textstyle \hspace{0.65in} + \frac{\nu_{\min}(\A\A^\top)}{2 \gamma} \| \mu^{k+1} - \mu^*\|^2 \Big]  \nonumber \\
    & \textstyle - \frac{(1 - 1/\delta)}{\gamma} \Big[  \|\lambda^{k+1} - \lambda^*\| \| \A^\top \mu^* - \A^\top \mu^{k+1} - t_k\|   \nonumber \\
    & \hspace{0.15in} + \|t_k\| \|\A^\top \mu^* - \A^\top \mu^{k+1}\| + \gamma \|\lambda^{k+1} - \lambda^k \| \|\x^{k+1} - \x^*\| \Big]. \nonumber 
\end{align}
Let $\Delta:= (1 - \frac{1}{\delta}) \min\{ 1, \nu_{\min}(\A\A^\top)\}$ and $s_k:= \frac{\gamma}{2}  \|\x^{k} - \x^* \|^2 + \frac{1}{2\gamma} \|\lambda^{k} - \lambda^*\|^2 + \frac{1}{2 \gamma} \| \mu^{k+1} - \mu^*\|^2$. Therefore, for $\delta \in (1, 1 + \frac{2\sigma \gamma}{L^2 + \gamma^2}]$, we have, 
\begin{align}
    &\textstyle \sigma\| \x^{k+1} - \x^* \|^2 + \frac{\gamma}{2}\|\x^{k+1} - \x^{k}\|^2 + \frac{1}{2\gamma} \|\lambda^{k+1} - \lambda^k\|^2 \label{eq:main_1}\\     
    & \geq \textstyle \Delta s_{k+1} - \frac{(1 - 1/\delta)}{\gamma} \big[ \|\lambda^{k+1} - \lambda^* \| \| \A^\top \mu^* - \A^\top \mu^{k+1} - t_k\| \nonumber \\
    & + \|t_k\| \|\A^\top \mu^* - \A^\top \mu^{k+1}\| +  \gamma \|\lambda^{k+1} - \lambda^k \| \|\x^{k+1} - \x^*\| \big]. \nonumber
\end{align}
Therefore, from~(\ref{eq:for_lin_rate}) and~(\ref{eq:main_1}), we have
\begin{align}
    & s_{k+1} \leq \textstyle (\frac{1}{1 + \Delta}) s_k + (\frac{1}{1 + \Delta}) \Big( \|\lambda^{k+1} - \lambda^*\| \|e^{k+1}\| \label{eq:main_12} \\
    & \hspace{0.45in} + \textstyle \|\overline{\y}^{k+1} - \y^*\| \|\lambda^{k+1}\| + \|\x^* - \x^{k+1}\|\|t_k\| \nonumber \\
    & \textstyle + \frac{(1 - 1/\delta)}{\gamma} \big[ \|\lambda^{k+1} - \lambda^* \| \| \A^\top \mu^* - \A^\top \mu^{k+1}  - t_k\| \nonumber  \\
    & \textstyle + \|t_k\| \|\A^\top \mu^* - \A^\top \mu^{k+1}\| +  \gamma \|\lambda^{k+1} - \lambda^k \| \|\x^{k+1} - \x^*\| \big] \Big).\nonumber
\end{align}
Using~(\ref{eq:uk_for_convg2}),
\begin{align}
    \|\lambda^{k} - \lambda^*\| &\leq \textstyle \sqrt{2\gamma (s_0 + \frac{\mathbf{R}}{1-\rho} + \frac{3\sqrt{n} \mathbf{Q} \rho}{(1-\rho)^2})}:= \Xi. \label{eq:xi1_lam}
\end{align}
Let $S_1:= \|\lambda^{k+1} - \lambda^*\| \|e^{k+1}\|, S_2:= \|\overline{\y}^{k+1} - \y^*\|\|\lambda^{k+1}\|, S_3:= \|\x^* - \x^{k+1}\| \|t_k\|, S_4:= \frac{(1 - 1/\delta)}{\gamma} \|\lambda^{k+1} - \lambda^* \| \| \A^\top \mu^* - \A^\top \mu^{k+1} - t_k\|, S_5:= \frac{(1- 1/\delta)}{\gamma} \|t_k\| \|\A^\top \mu^* - \A^\top \mu^{k+1}\|$ and $S_6:= (1- 1/\delta)  \|\lambda^{k+1} - \lambda^k \| \|\x^{k+1} - \x^*\|$. 
\textcolor{black}{Note $S_1 = \|\lambda^{k+1} - \lambda^*\|\|e^{k+1}\| \leq \sqrt{n}\Xi \rho^{k+1}$. Further, from Theorem~\ref{thm:convergence}, there exist $\mathbf{C} < \infty$ and $\mathbf{S}^k$ with $\lim_{k \to \infty} \mathbf{S}^k = 0$ such that $S_2 + S_3 + S_4 + S_5 + S_6 \leq \mathbf{C} \mathbf{S}^k$.} 
\textcolor{black}{Therefore, from~(\ref{eq:main_12}) and using the fact $\rho^{k+1} \leq \rho^{k}, \forall k$
\begin{align}
    & s_{k+1} \leq \textstyle (\frac{1}{1 + \Delta}) s_k + \sqrt{n} \Xi (\frac{1}{1 + \Delta}) \rho^{k} + \mathbf{C}(\frac{1}{1 + \Delta}) \mathbf{S}^k .
\end{align}
By repeated substitution for $k = 0, 1, \dots, K$ we get,
\begin{align*}
   s_K & \leq \textstyle ( \frac{1}{1+\Delta} )^K s_0 + \sqrt{n} \Xi \sum_{t=0}^{K} ( \frac{1}{1+\Delta} )^{t} \rho^{K-t} \\ 
   & \hspace{1.5in} \textstyle + \mathbf{C} \sum_{t=0}^{K} ( \frac{1}{1+\Delta} )^{K-t+1} \mathbf{S}^t  \\
    & = \textstyle ( \frac{1}{1+\Delta} )^K s_0 + \rho^K \sqrt{n} \Xi \sum_{t=1}^{K} ( \frac{1}{1+\Delta} )^{t} \rho^{-t}\\ 
    & \hspace{1.5in} \textstyle + \mathbf{C} \sum_{t=0}^{K} ( \frac{1}{1+\Delta} )^{K-t+1} \mathbf{S}^t  \\
    &\leq \textstyle \rho^K \big(s_0 + \sqrt{n} \Xi \sum_{t=1}^{K} ( \frac{1}{1+\Delta} )^{t} \rho^{-t} \big) \\ 
     & \hspace{1.5in} \textstyle + \mathbf{C} \sum_{t=0}^{K} ( \frac{1}{1+\Delta} )^{K-t+1} \mathbf{S}^t.  
\end{align*}
Given, $\epsilon > 0$, let $K^\prime$ be such that $\mathbf{S}^k < \epsilon$ for all $k \geq K^\prime$ (Theorem~\ref{thm:convergence} guarantees existence of such $K^\prime$). Let $\overline{\mathbf{S}}:= \max_{0 \leq t \leq K^\prime} \mathbf{S}^t$, therefore, 
\begin{align*}
   s_K & \leq \textstyle \rho^K \big(s_0 + \sqrt{n} \Xi \sum_{t=1}^{K} ( \frac{1}{1+\Delta} )^{t} \rho^{-t} + \mathbf{C} \overline{\mathbf{S}} \sum_{t=0}^{K^\prime} \rho^{1-t} \big) \\ 
     & \hspace{1.5in} \textstyle + \epsilon \left(\mathbf{C} \sum_{t={K^\prime}}^{K} \rho^{K-t+1}\right)\\
     & \leq \textstyle \rho^K \big(s_0 +  \frac{\sqrt{n} \Xi\rho(1+\Delta)}{\rho(1+\Delta) - 1}  + \frac{\mathbf{C} \overline{\mathbf{S}} (\rho^{-K^\prime +1} - \rho^2 )}{1-\rho}\big) + \epsilon \left(\frac{\mathbf{C}}{1 - \rho}\right)\\
     & = \Upsilon \rho^K + O(\epsilon).  
\end{align*}
\begin{align} \label{eq:upslion}
    \hspace{-1.75in} \text{where, } \textstyle \Upsilon := s_0 + \frac{\mathbf{S}_2\rho(1+\Delta)}{\rho(1+\Delta) - 1}.
\end{align}
This completes the proof. \qed
}

\bibliography{references}

\begin{thebibliography}{10}
\providecommand{\url}[1]{#1}
\csname url@samestyle\endcsname
\providecommand{\newblock}{\relax}
\providecommand{\bibinfo}[2]{#2}
\providecommand{\BIBentrySTDinterwordspacing}{\spaceskip=0pt\relax}
\providecommand{\BIBentryALTinterwordstretchfactor}{4}
\providecommand{\BIBentryALTinterwordspacing}{\spaceskip=\fontdimen2\font plus
\BIBentryALTinterwordstretchfactor\fontdimen3\font minus
  \fontdimen4\font\relax}
\providecommand{\BIBforeignlanguage}[2]{{%
\expandafter\ifx\csname l@#1\endcsname\relax
\typeout{** WARNING: IEEEtran.bst: No hyphenation pattern has been}%
\typeout{** loaded for the language `#1'. Using the pattern for}%
\typeout{** the default language instead.}%
\else
\language=\csname l@#1\endcsname
\fi
#2}}
\providecommand{\BIBdecl}{\relax}
\BIBdecl

\bibitem{nedic2018distributed}
A.~Nedi{\'c} and J.~Liu, ``Distributed optimization for control,'' \emph{Annual
  Review of Control, Robotics, and Autonomous Systems}, vol.~1, pp. 77--103,
  2018.

\bibitem{nedic2020distributed}
A.~Nedic, ``Distributed gradient methods for convex machine learning problems
  in networks: Distributed optimization,'' \emph{IEEE Signal Processing
  Magazine}, vol.~37, no.~3, pp. 92--101, 2020.

\bibitem{wu2021distributed}
X.~Wu, H.~Wang, and J.~Lu, ``Distributed optimization with coupling
  constraints,'' \emph{arXiv preprint arXiv:2102.12989}, 2021.

\bibitem{falsone2017dual}
A.~Falsone, K.~Margellos, S.~Garatti, and M.~Prandini, ``Dual decomposition for
  multi-agent distributed optimization with coupling constraints,''
  \emph{Automatica}, vol.~84, pp. 149--158, 2017.

\bibitem{notarnicola2019constraint}
I.~Notarnicola and G.~Notarstefano, ``Constraint-coupled distributed
  optimization: a relaxation and duality approach,'' \emph{IEEE Transactions on
  Control of Network Systems}, vol.~7, no.~1, pp. 483--492, 2019.

\bibitem{tsitsiklis1984problems}
J.~N. Tsitsiklis, ``Problems in decentralized decision making and
  computation.'' Massachusetts Inst of Tech Cambridge Lab for Information and
  Decision Systems, Tech. Rep., 1984.

\bibitem{bertsekas1989parallel}
D.~P. Bertsekas and J.~N. Tsitsiklis, \emph{Parallel and distributed
  computation: numerical methods}.\hskip 1em plus 0.5em minus 0.4em\relax
  Prentice hall Englewood Cliffs, NJ, 1989, vol.~23.

\bibitem{nedic2014distributed}
A.~Nedi{\'c} and A.~Olshevsky, ``Distributed optimization over time-varying
  directed graphs,'' \emph{IEEE Transactions on Automatic Control}, vol.~60,
  no.~3, pp. 601--615, 2014.

\bibitem{xin2020gradient}
R.~Xin, S.~Kar, and U.~A. Khan, ``Gradient tracking and variance reduction for
  decentralized optimization and machine learning,'' \emph{arXiv preprint
  arXiv:2002.05373}, 2020.

\bibitem{zeng2015extrapush}
J.~Zeng and W.~Yin, ``Extrapush for convex smooth decentralized optimization
  over directed networks,'' \emph{arXiv preprint arXiv:1511.02942}, 2015.

\bibitem{pu2020push}
S.~Pu, W.~Shi, J.~Xu, and A.~Nedic, ``Push-pull gradient methods for
  distributed optimization in networks,'' \emph{IEEE Transactions on Automatic
  Control}, 2020.

\bibitem{khatana2020gradient}
\BIBentryALTinterwordspacing
V.~Khatana, G.~Saraswat, S.~Patel, and M.~V. Salapaka, ``Gradient-consensus
  method for distributed optimization in directed multi-agent networks,''
  \emph{arXiv preprint arXiv:1909.10070}, 2019. [Online]. Available:
  \url{https://arxiv.org/pdf/1909.10070v7.pdf}
\BIBentrySTDinterwordspacing

\bibitem{MAL-016}
\BIBentryALTinterwordspacing
S.~Boyd, N.~Parikh, E.~Chu, B.~Peleato, and J.~Eckstein, ``Distributed
  optimization and statistical learning via the alternating direction method of
  multipliers,'' \emph{Foundations and Trends® in Machine Learning}, vol.~3,
  no.~1, pp. 1--122, 2011. [Online]. Available:
  \url{http://dx.doi.org/10.1561/2200000016}
\BIBentrySTDinterwordspacing

\bibitem{wei2012distributed}
E.~Wei and A.~Ozdaglar, ``Distributed alternating direction method of
  multipliers,'' in \emph{2012 IEEE 51st IEEE Conference on Decision and
  Control (CDC)}.\hskip 1em plus 0.5em minus 0.4em\relax IEEE, 2012, pp.
  5445--5450.

\bibitem{makhdoumi2017convergence}
A.~Makhdoumi and A.~Ozdaglar, ``Convergence rate of distributed admm over
  networks,'' \emph{IEEE Transactions on Automatic Control}, vol.~62, no.~10,
  pp. 5082--5095, 2017.

\bibitem{shi2014linear}
W.~Shi, Q.~Ling, K.~Yuan, G.~Wu, and W.~Yin, ``On the linear convergence of the
  admm in decentralized consensus optimization,'' \emph{IEEE Transactions on
  Signal Processing}, vol.~62, no.~7, pp. 1750--1761, 2014.

\bibitem{wei20131}
E.~Wei and A.~Ozdaglar, ``On the $\large{O}(1/k)$ convergence of asynchronous
  distributed alternating direction method of multipliers,'' in \emph{2013 IEEE
  Global Conference on Signal and Information Processing}.\hskip 1em plus 0.5em
  minus 0.4em\relax IEEE, 2013, pp. 551--554.

\bibitem{chang2014multi}
T.-H. Chang, M.~Hong, and X.~Wang, ``Multi-agent distributed optimization via
  inexact consensus admm,'' \emph{IEEE Transactions on Signal Processing},
  vol.~63, no.~2, pp. 482--497, 2014.

\bibitem{iutzeler2015explicit}
F.~Iutzeler, P.~Bianchi, P.~Ciblat, and W.~Hachem, ``Explicit convergence rate
  of a distributed alternating direction method of multipliers,'' \emph{IEEE
  Transactions on Automatic Control}, vol.~61, no.~4, pp. 892--904, 2015.

\bibitem{mansoori2019flexible}
F.~Mansoori and E.~Wei, ``A flexible framework of first-order primal-dual
  algorithms for distributed optimization,'' \emph{arXiv preprint
  arXiv:1912.07526}, 2019.

\bibitem{gabay1983chapter}
D.~Gabay, ``Chapter ix applications of the method of multipliers to variational
  inequalities,'' in \emph{Studies in mathematics and its applications}.\hskip
  1em plus 0.5em minus 0.4em\relax Elsevier, 1983, vol.~15, pp. 299--331.

\bibitem{chen2020distributed}
G.~Chen, Q.~Yang, Y.~Song, and F.~L. Lewis, ``A distributed continuous-time
  algorithm for nonsmooth constrained optimization,'' \emph{IEEE Transactions
  on Automatic Control}, vol.~65, no.~11, pp. 4914--4921, 2020.

\bibitem{chen2021fixed}
------, ``Fixed-time projection algorithm for distributed constrained
  optimization on time-varying digraphs,'' \emph{IEEE Transactions on Automatic
  Control}, 2021.

\bibitem{zhu2011distributed}
M.~Zhu and S.~Mart{\'\i}nez, ``On distributed convex optimization under
  inequality and equality constraints,'' \emph{IEEE Transactions on Automatic
  Control}, vol.~57, no.~1, pp. 151--164, 2011.

\bibitem{tian2019distributed}
F.~Tian, W.~Yu, J.~Fu, W.~Gu, and J.~Gu, ``Distributed optimization of
  multiagent systems subject to inequality constraints,'' \emph{IEEE
  transactions on cybernetics}, 2019.

\bibitem{yang2016multi}
S.~Yang, Q.~Liu, and J.~Wang, ``A multi-agent system with a
  proportional-integral protocol for distributed constrained optimization,''
  \emph{IEEE Transactions on Automatic Control}, vol.~62, no.~7, pp.
  3461--3467, 2016.

\bibitem{liu2015second}
Q.~Liu and J.~Wang, ``A second-order multi-agent network for bound-constrained
  distributed optimization,'' \emph{IEEE Transactions on Automatic Control},
  vol.~60, no.~12, pp. 3310--3315, 2015.

\bibitem{zhou2019adaptive}
H.~Zhou, X.~Zeng, and Y.~Hong, ``Adaptive exact penalty design for constrained
  distributed optimization,'' \emph{IEEE Transactions on Automatic Control},
  vol.~64, no.~11, pp. 4661--4667, 2019.

\bibitem{lin2021angle}
P.~Lin, J.~Xu, W.~Ren, C.~Yang, and W.~Gui, ``Angle-based analysis approach for
  distributed constrained optimization,'' \emph{IEEE Transactions on Automatic
  Control}, 2021.

\bibitem{liu2017constrained}
Q.~Liu, S.~Yang, and Y.~Hong, ``Constrained consensus algorithms with fixed
  step size for distributed convex optimization over multiagent networks,''
  \emph{IEEE Transactions on Automatic Control}, vol.~62, no.~8, pp.
  4259--4265, 2017.

\bibitem{nedic2010constrained}
A.~Nedic, A.~Ozdaglar, and P.~A. Parrilo, ``Constrained consensus and
  optimization in multi-agent networks,'' \emph{IEEE Transactions on Automatic
  Control}, vol.~55, no.~4, pp. 922--938, 2010.

\bibitem{li2020distributed}
H.~Li, Q.~L{\"u}, G.~Chen, T.~Huang, and Z.~Dong, ``Distributed constrained
  optimization over unbalanced directed networks using asynchronous
  broadcast-based algorithm,'' \emph{IEEE Transactions on Automatic Control},
  vol.~66, no.~3, pp. 1102--1115, 2020.

\bibitem{xie2018distributed}
P.~Xie, K.~You, R.~Tempo, S.~Song, and C.~Wu, ``Distributed convex optimization
  with inequality constraints over time-varying unbalanced digraphs,''
  \emph{IEEE Transactions on Automatic Control}, vol.~63, no.~12, pp.
  4331--4337, 2018.

\bibitem{mota2013d}
J.~F. Mota, J.~M. Xavier, P.~M. Aguiar, and M.~P{\"u}schel, ``D-admm: A
  communication-efficient distributed algorithm for separable optimization,''
  \emph{IEEE Transactions on Signal Processing}, vol.~61, no.~10, pp.
  2718--2723, 2013.

\bibitem{yuan2015regularized}
D.~Yuan, D.~W. Ho, and S.~Xu, ``Regularized primal--dual subgradient method for
  distributed constrained optimization,'' \emph{IEEE transactions on
  cybernetics}, vol.~46, no.~9, pp. 2109--2118, 2015.

\bibitem{lei2016primal}
J.~Lei, H.-F. Chen, and H.-T. Fang, ``Primal--dual algorithm for distributed
  constrained optimization,'' \emph{Systems \& Control Letters}, vol.~96, pp.
  110--117, 2016.

\bibitem{patel2017distributed}
S.~Patel, S.~Attree, S.~Talukdar, M.~Prakash, and M.~V. Salapaka, ``Distributed
  apportioning in a power network for providing demand response services,'' in
  \emph{2017 IEEE International Conference on Smart Grid Communications
  (SmartGridComm)}.\hskip 1em plus 0.5em minus 0.4em\relax IEEE, 2017, pp.
  38--44.

\bibitem{patel2020distributed}
S.~Patel, V.~Khatana, G.~Saraswat, and M.~V. Salapaka, ``Distributed detection
  of malicious attacks on consensus algorithms with applications in power
  networks,'' 2020.

\bibitem{shi2015extra}
W.~Shi, Q.~Ling, G.~Wu, and W.~Yin, ``Extra: An exact first-order algorithm for
  decentralized consensus optimization,'' \emph{SIAM Journal on Optimization},
  vol.~25, no.~2, pp. 944--966, 2015.

\bibitem{jiang2021fully}
W.~Jiang and T.~Charalambous, ``Fully distributed alternating direction method
  of multipliers in digraphs via finite-time termination mechanisms,''
  \emph{arXiv preprint arXiv:2107.02019}, 2021.

\bibitem{jiang2021distributed}
------, ``Distributed alternating direction method of multipliers using
  finite-time exact ratio consensus in digraphs,'' in \emph{2021 European
  Control Conference (ECC)}.\hskip 1em plus 0.5em minus 0.4em\relax IEEE, 2021,
  pp. 2205--2212.

\bibitem{jiang2021asynchronous}
W.~Jiang, A.~Grammenos, E.~Kalyvianaki, and T.~Charalambous, ``An asynchronous
  approximate distributed alternating direction method of multipliers in
  digraphs,'' \emph{arXiv preprint arXiv:2104.11866}, 2021.

\bibitem{rokade2020distributed}
K.~Rokade and R.~K. Kalaimani, ``Distributed admm over directed graphs,''
  \emph{arXiv preprint arXiv:2010.10421}, 2020.

\bibitem{khatana2020cdc}
V.~Khatana and M.~V. Salapaka, ``D-distadmm: A $\large{O}(1/k)$ distributed
  admm for distributed optimization in directed graph topologies,'' in
  \emph{59th IEEE Conference on Decision and Control (CDC)}.\hskip 1em plus
  0.5em minus 0.4em\relax IEEE, 2020, pp. 2992--2997.

\bibitem{Die06}
R.~Diestel, \emph{Graph Theory}.\hskip 1em plus 0.5em minus 0.4em\relax Berlin,
  Germany: Springer-Verlag, 2006.

\bibitem{horn2012matrix}
R.~A. Horn and C.~R. Johnson, \emph{Matrix analysis}.\hskip 1em plus 0.5em
  minus 0.4em\relax Cambridge university press, 2012.

\bibitem{rockafellar2015convex}
R.~T. Rockafellar, \emph{Convex analysis}.\hskip 1em plus 0.5em minus
  0.4em\relax Princeton university press, 2015.

\bibitem{wu2020second}
X.~Wu, Z.~Qu, and J.~Lu, ``A second-order proximal algorithm for consensus
  optimization,'' \emph{IEEE Transactions on Automatic Control}, vol.~66,
  no.~4, pp. 1864--1871, 2020.

\bibitem{melbourne2020geometry}
\BIBentryALTinterwordspacing
J.~Melbourne, G.~Saraswat, V.~Khatana, S.~Patel, and M.~V. Salapaka, ``On the
  geometry of consensus algorithms with application to distributed termination
  in higher dimension,'' \emph{in the proceedings of International Federation
  of Automatic Control (IFAC)}, 2020. [Online]. Available:
  \url{http://box5779.temp.domains/~jamesmel/wp-content/uploads/2019/11/Vector_Consensus__IFAC_-1.pdf}
\BIBentrySTDinterwordspacing

\bibitem{kempe2003gossip}
D.~Kempe, A.~Dobra, and J.~Gehrke, ``Gossip-based computation of aggregate
  information,'' in \emph{44th Annual IEEE Symposium on Foundations of Computer
  Science, 2003. Proceedings.}\hskip 1em plus 0.5em minus 0.4em\relax IEEE,
  2003, pp. 482--491.

\bibitem{dominguez2011distributed}
A.~D. Dom{\'\i}nguez-Garc{\'\i}a and C.~N. Hadjicostis, ``Distributed
  strategies for average consensus in directed graphs,'' in \emph{2011 50th
  IEEE Conference on Decision and Control and European Control
  Conference}.\hskip 1em plus 0.5em minus 0.4em\relax IEEE, 2011, pp.
  2124--2129.

\bibitem{benezit2010weighted}
F.~B{\'e}n{\'e}zit, V.~Blondel, P.~Thiran, J.~Tsitsiklis, and M.~Vetterli,
  ``Weighted gossip: Distributed averaging using non-doubly stochastic
  matrices,'' in \emph{2010 IEEE International Symposium on Information
  Theory}.\hskip 1em plus 0.5em minus 0.4em\relax IEEE, 2010, pp. 1753--1757.

\bibitem{bertsekas1997nonlinear}
D.~P. Bertsekas, ``Nonlinear programming,'' \emph{Journal of the Operational
  Research Society}, vol.~48, no.~3, pp. 334--334, 1997.

\bibitem{zhang2015restricted}
H.~Zhang and L.~Cheng, ``Restricted strong convexity and its applications to
  convergence analysis of gradient-type methods in convex optimization,''
  \emph{Optimization Letters}, vol.~9, no.~5, pp. 961--979, 2015.

\bibitem{bach2014adaptivity}
F.~Bach, ``Adaptivity of averaged stochastic gradient descent to local strong
  convexity for logistic regression,'' \emph{The Journal of Machine Learning
  Research}, vol.~15, no.~1, pp. 595--627, 2014.

\bibitem{berahas2018balancing}
A.~S. Berahas, R.~Bollapragada, N.~S. Keskar, and E.~Wei, ``Balancing
  communication and computation in distributed optimization,'' \emph{IEEE
  Transactions on Automatic Control}, vol.~64, no.~8, pp. 3141--3155, 2018.

\bibitem{berahas2021convergence}
A.~S. Berahas, R.~Bollapragada, and E.~Wei, ``On the convergence of nested
  decentralized gradient methods with multiple consensus and gradient steps,''
  \emph{IEEE Transactions on Signal Processing}, vol.~69, pp. 4192--4203, 2021.

\bibitem{chen2012fast}
A.~I.-A. Chen, ``Fast distributed first-order methods,'' Ph.D. dissertation,
  Massachusetts Institute of Technology, 2012.

\bibitem{li2018sharp}
H.~Li, C.~Fang, W.~Yin, and Z.~Lin, ``A sharp convergence rate analysis for
  distributed accelerated gradient methods,'' \emph{arXiv preprint
  arXiv:1810.01053}, 2018.

\bibitem{ye2020multi}
H.~Ye, L.~Luo, Z.~Zhou, and T.~Zhang, ``Multi-consensus decentralized
  accelerated gradient descent,'' \emph{arXiv preprint arXiv:2005.00797}, 2020.

\bibitem{jakovetic2014fast}
D.~Jakoveti{\'c}, J.~Xavier, and J.~M. Moura, ``Fast distributed gradient
  methods,'' \emph{IEEE Transactions on Automatic Control}, vol.~59, no.~5, pp.
  1131--1146, 2014.

\bibitem{johansson2008subgradient}
B.~Johansson, T.~Keviczky, M.~Johansson, and K.~H. Johansson, ``Subgradient
  methods and consensus algorithms for solving convex optimization problems,''
  in \emph{2008 47th IEEE Conference on Decision and Control}.\hskip 1em plus
  0.5em minus 0.4em\relax IEEE, 2008, pp. 4185--4190.

\bibitem{lin2015global}
T.~Lin, S.~Ma, and S.~Zhang, ``On the global linear convergence of the admm
  with multiblock variables,'' \emph{SIAM Journal on Optimization}, vol.~25,
  no.~3, pp. 1478--1497, 2015.

\bibitem{rockafellar1976augmented}
R.~T. Rockafellar, ``Augmented lagrangians and applications of the proximal
  point algorithm in convex programming,'' \emph{Mathematics of operations
  research}, vol.~1, no.~2, pp. 97--116, 1976.

\bibitem{nedic2017achieving}
A.~Nedic, A.~Olshevsky, and W.~Shi, ``Achieving geometric convergence for
  distributed optimization over time-varying graphs,'' \emph{SIAM Journal on
  Optimization}, vol.~27, no.~4, pp. 2597--2633, 2017.

\bibitem{erdHos1960evolution}
P.~Erd{\H{o}}s and A.~R{\'e}nyi, ``On the evolution of random graphs,''
  \emph{Publ. Math. Inst. Hung. Acad. Sci}, vol.~5, no.~1, pp. 17--60, 1960.

\bibitem{olshevsky2009convergence}
A.~Olshevsky and J.~N. Tsitsiklis, ``Convergence speed in distributed consensus
  and averaging,'' \emph{SIAM Journal on Control and Optimization}, vol.~48,
  no.~1, pp. 33--55, 2009.

\bibitem{beck2009fast}
A.~Beck and M.~Teboulle, ``A fast iterative shrinkage-thresholding algorithm
  for linear inverse problems,'' \emph{SIAM journal on imaging sciences},
  vol.~2, no.~1, pp. 183--202, 2009.

\bibitem{prakash2019distributed}
M.~Prakash, S.~Talukdar, S.~Attree, V.~Yadav, and M.~V. Salapaka, ``Distributed
  stopping criterion for consensus in the presence of delays,'' \emph{IEEE
  Transactions on Control of Network Systems}, 2019.

\bibitem{saraswat2019distributed}
G.~Saraswat, V.~Khatana, S.~Patel, and M.~V. Salapaka, ``Distributed
  finite-time termination for consensus algorithm in switching topologies,''
  \emph{arXiv preprint arXiv:1909.00059}, 2019.

\bibitem{gharesifard2012distributed}
B.~Gharesifard and J.~Cort{\'e}s, ``Distributed strategies for generating
  weight-balanced and doubly stochastic digraphs,'' \emph{European Journal of
  Control}, vol.~18, no.~6, pp. 539--557, 2012.

\bibitem{yadav2007distributed}
V.~Yadav and M.~V. Salapaka, ``Distributed protocol for determining when
  averaging consensus is reached,'' in \emph{45th Annual Allerton Conf}, 2007,
  pp. 715--720.

\bibitem{rockafellar1976monotone}
R.~T. Rockafellar, ``Monotone operators and the proximal point algorithm,''
  \emph{SIAM journal on control and optimization}, vol.~14, no.~5, pp.
  877--898, 1976.

\end{thebibliography}
\begin{IEEEbiography}[
{
\includegraphics[width=1in,height=1.25in,clip,keepaspectratio]{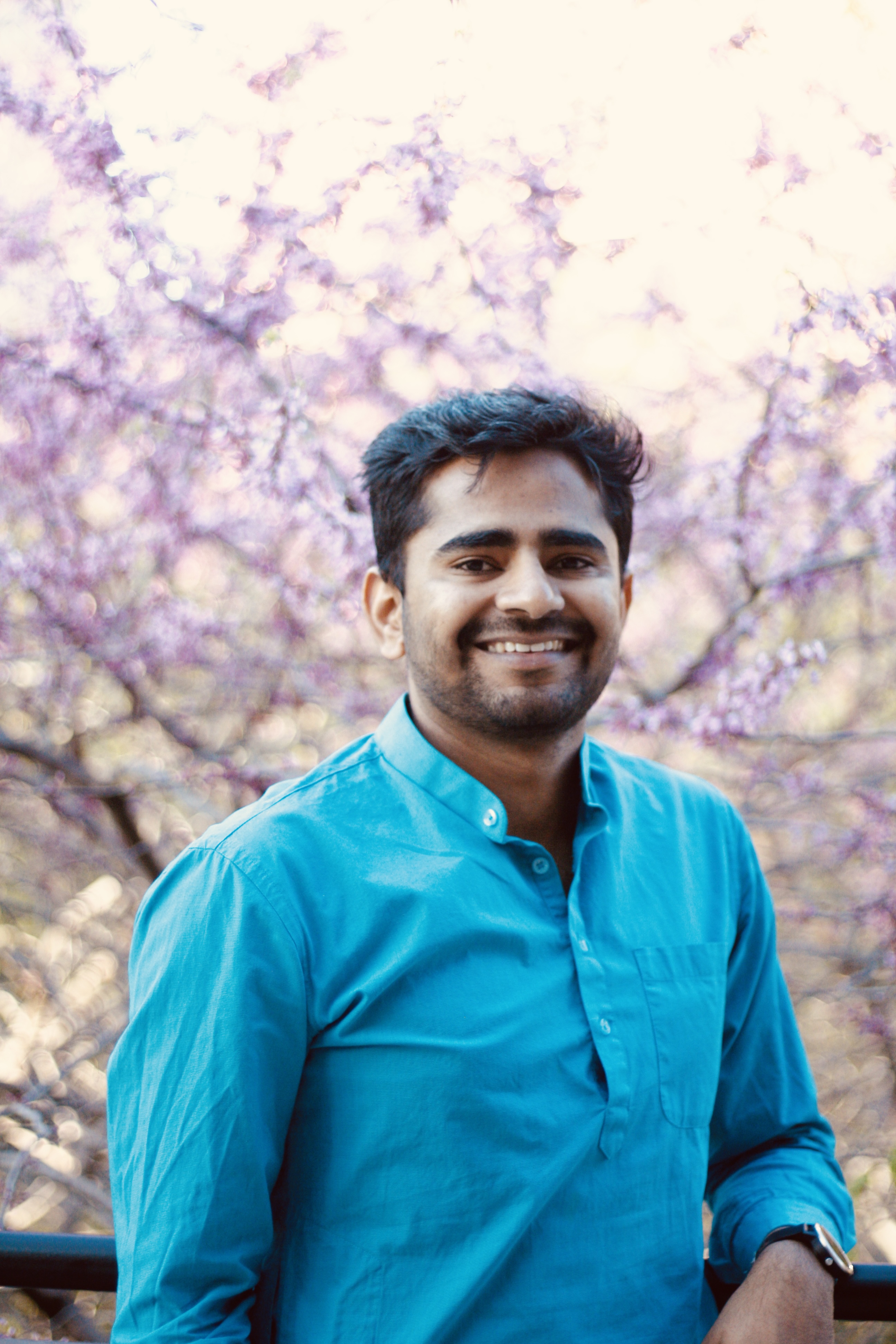}
}
]
{Vivek Khatana} received the B.Tech degree in Electrical Engineering from the Indian Institute of Technology, Roorkee, in 2018. Currently, he is working towards a Ph.D. degree at the department of Electrical Engineering at University of Minnesota. His research interests include distributed optimization, consensus algorithms, distributed control and stochastic calculus, power system analysis and control.
\end{IEEEbiography}
\begin{IEEEbiography}[
{
\includegraphics[width=1in,height=1.25in,clip,keepaspectratio]{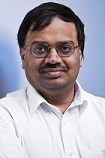}
}
]
{Murti V. Salapaka}
received the B.Tech. degree in Mechanical Engineering from the Indian Institute of Technology, Madras, in 1991 and the M.S. and Ph.D. degrees in Mechanical Engineering from the University of California at Santa Barbara, in 1993 and 1997, respectively. He was a faculty member in the Electrical and Computer Engineering Department, Iowa State University, Ames, from 1997 to 2007. Currently, he is the Director of Graduate Studies and the Vincentine Hermes Luh Chair Professor in the Electrical and Computer Engineering Department, University of Minnesota, Minneapolis. His research interests include control and network science, nanoscience and single molecule physics. Dr. Salapaka was the recipient of the NSF CAREER Award and the ISU Young Engineering Faculty Research Award for the years 1998 and 2001, and is an IEEE fellow. 
\end{IEEEbiography}

\end{document}